\documentclass[reqno]{amsart}

\usepackage[english]{babel}
\usepackage{amsmath}
\usepackage{amssymb}
\usepackage{enumerate}
\usepackage{graphicx}

\usepackage[ citecolor=black, linkcolor=black,
colorlinks, hypertexnames=false, plainpages=false]{hyperref}

\usepackage{color}

\newtheorem{lemma}     {Lemma}[section]
\newtheorem{thm}   [lemma]{Theorem}
\newtheorem{teorema1}   [lemma]{Theorem}
\newtheorem{prop}      [lemma]{Proposition}
\newtheorem{coro}    [lemma]{Corollary}
\newtheorem{cong1}      [lemma]{Conjecture}
\newtheorem{remark1}    [lemma]{Remark}
\newtheorem{defin}     [lemma]{Definition}

\numberwithin{equation}{section}

\newcommand{\und}{\underline}

     \newcommand{\nn}{\nonumber}
\newcommand{\1}{\text{\bf 1}}

\newcommand{\dis}{\displaystyle}

\newcommand{\mmmintone}[1]{{\dis{\int\kern -.43cm
-}}_{\kern-.21cm\substack{#1}}\;\;}
\newcommand{\mmmintwo}[2]{{\dis{\int\kern -.43cm
-}}_{\kern-.21cm\substack{#1}}^{\substack{#2}}\;\;}
\newcommand{\submint}{{\scriptstyle{\int\kern -.66em -}}}
\newcommand{\submintone}[1]{{\scriptstyle{\int\kern -.66em
-}}_{\scriptscriptstyle{\kern-.21em\substack{#1}}}}
\newcommand{\fracmint}{{\textstyle{\int\kern -.88em -}}}
\newcommand{\fracmintone}[1]{{\textstyle{\int\kern -.88em
-}}_{\scriptscriptstyle{\kern-.21em\substack{#1}}}\;}

\def\mintone{\protect\mmmintone}

\newcommand{\eps}{\epsilon}

\newcommand{\ga}{\gamma}
\newcommand{\Ga}{\Gamma}
\newcommand{\Om}{\Omega}
\newcommand{\om}{\omega}
\newcommand{\si}{\sigma}
\newcommand{\Si}{\Sigma}

\newcommand{\la}{\lambda}
\newcommand{\La}{\Lambda}

\newcommand{\nada}[1]{}

\begin{document}


\title[Continuum Potts model]{Coexistence of ordered and disordered
phases in Potts models in the continuum}

\date{\today}
\author [A.\ De Masi] {Anna De Masi}
\address{Anna de Masi, Dipartimento di Matematica Pura ed Applicata,
Universit\`a di L'Aquila, Via Vetoio (Coppito) 67100 L'Aquila,
Italy} \email{demasi@univaq.it}

\author[I.\ Merola] {Immacolata Merola}
\address{Immacolata Merola, Dipartimento di Matematica Pura ed Applicata,
Universit\`a di L'Aquila, Via Vetoio (Coppito) 67100 L'Aquila,
Italy} \email{merola@univaq.it}

\author[E.\ Presutti] {Errico Presutti}
\address{Errico Presutti, Dipartimento di Matematica, Universit\`a
di Roma Tor Ver\-ga\-ta, 00133 Roma, Italy}
\email{presutti@mat.uniroma2.it}

\author[Y.\ Vignaud] {Yvon Vignaud}
\address{Yvon Vignaud, TU Berlin- Fakult\"{a}t II, Institut f\"{u}r
Mathematik D-10623 Berlin, Germany}
\email{vignaud@mail.tu-berlin.de}

\thanks{ \vskip.2em \noindent {\it Key words and phrases.}\/
Phase transition in Continuum Particle Systems. Pirogov Sinai
theory. }

\begin{abstract}

This is the second of two papers on a continuum version of the Potts
model, where particles are points in $\mathbb R^d$, $d\ge 2$, with a
spin which may take $S\ge 3$ possible values. Particles with
different spins repel each other via a Kac pair potential of range
$\ga^{-1}$, $\ga>0$. In this paper we prove phase transition, namely
we prove that if the scaling parameter of the Kac potential is
suitably small, given any temperature there is a value of the
chemical potential such that at the given temperature and chemical
potential there exist $S+1$ mutually distinct DLR measures.

\end{abstract} \maketitle

\tableofcontents

\section{{\bf Introduction}}
        \label{sec:1}

The conjecture that {\em mean field phase diagrams are well
approximated by systems with long range interactions} cannot be
taken literally as  it obviously fails in one dimensional systems
(if the second moment of the interaction is finite), moreover the
mean field critical exponents are [believed to be] different from
those computed for finite range interactions. With proper caveat
however the conjecture is generally regarded as correct and indeed
there are  mathematical proofs mainly referring to specific models
and focused on the occurrence of  phase transitions. The choice of
the approximating hamiltonian is not at all arbitrary and the
results so far have been obtained for reflection  positive
interactions, \cite{chayes},  and for Kac potentials, \cite{CP}. The
former  choice is clearly motivated by a powerful and well developed
theory, the latter class seems more general, in particular includes
systems of particles in the continuum as the one  considered in the
present paper. We will in fact study here a continuum version of the
classical Potts model. Its mean field  free energy is
   \begin{equation}
      \label{z1.1}
F^{\rm mf}_{\beta,\la}(\rho) = \frac 12 \sum_{s\ne s'}
\rho_s\rho_{s'} -\la \sum_s \rho_s - \frac 1 \beta
\mathcal{S}(\rho),\qquad \mathcal{S}(\rho)=- \sum_s \rho_s [\log
\rho_s-1 ]
     \end{equation}
$\rho=\{\rho_1,..,\rho_S\}\in \mathbb R_+^S$, $\rho_s$ represents
the density of particles with spin $s$, $s\in \{1,..,S\}$, $S\ge 3$;
$\beta$ the inverse temperature; $\la$ the chemical potential.

Despite the simplicity of the model its thermodynamics, which is
defined by minimizing $F^{\rm mf}_{\beta,\la}(\rho)$ over $\rho \in
\mathbb R_+^S$, has a rather interesting structure. In  \cite{GMRZ}
and \cite{DMPV2}  it is proved that the resulting phase diagram is
characterized by a critical curve $\la=\la_\beta, \beta >0$, as  in
Figure 1.

\begin{figure}
\centering
\includegraphics[width=.8\textwidth]{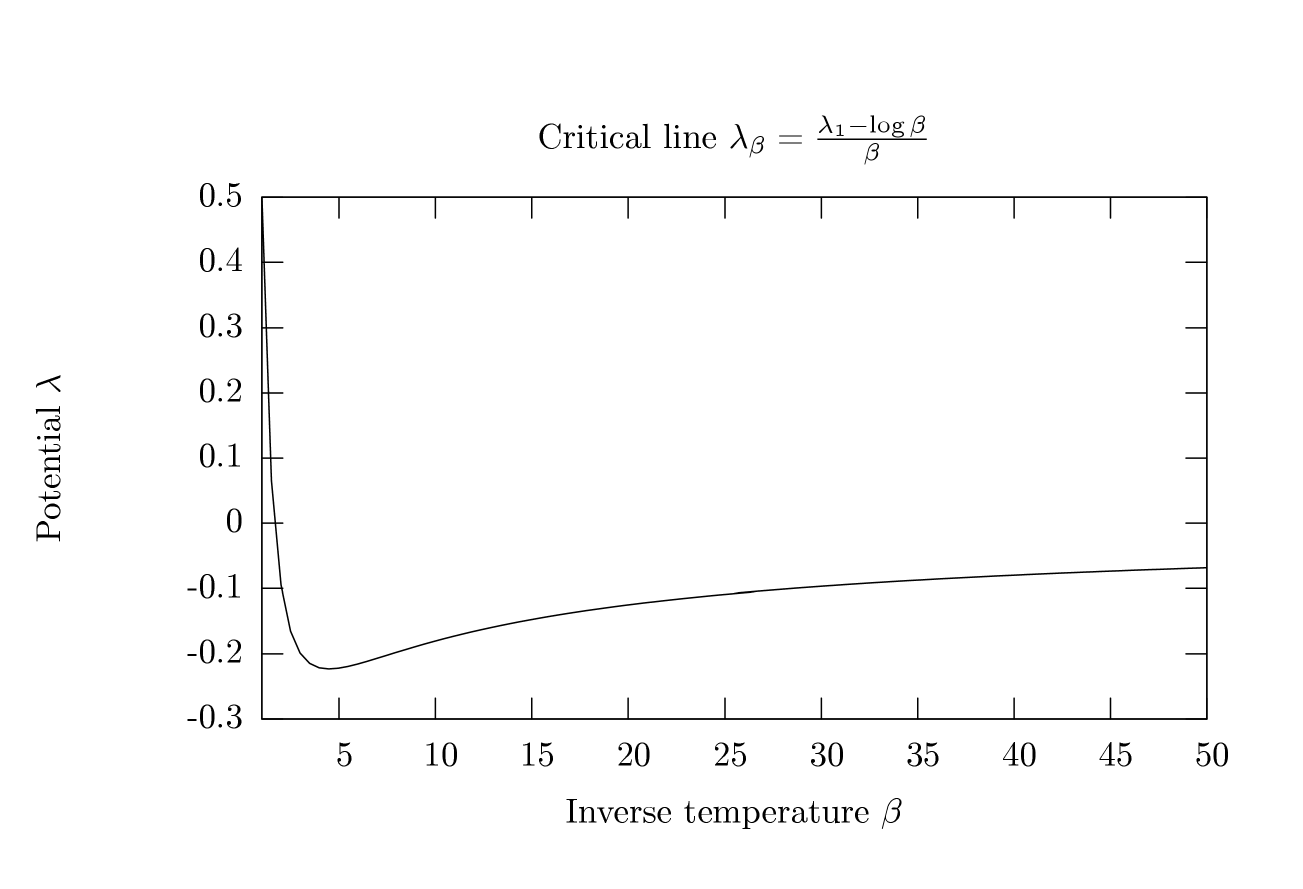}
\caption{Phase Diagram of the Mean field Potts gas}
\label{phasediagram}
\end{figure}

$F^{\rm mf}_{\beta,\la_\beta }$ has $S+1$ minimizers $\rho^{(k)}
=(\rho_s^{(k)}, s=1,..,S)$, $k=1,..,S+1$. There are positive numbers
$a$, $b<c$ so that
    \begin{equation}
    \label{orderdis}
 \rho^{(S+1)}_s=a,\,\,\,    \forall
s,\qquad \text{ for }k\le S:\;\; \rho^{(k)}_s=b\,\,\,    \forall
s\ne k,\,\,\, \quad\rho^{(k)}_k=c
    \end{equation}
 Furthermore
        \begin{equation}
    \label{density}
Sa < b^*,\quad b^*:=(S-1)b+c
         \end{equation}
so that the total density of the state $\rho^{(S+1)}$ is smaller
than the total density in any of the ordered critical points
$\rho^{(k)}$, $k\le S$, which is in fact equal to $b^*$.

When $\la>\la_\beta$, only the ordered states survive and there are
$S$ minimizers, when $\la<\la_\beta$, only the disordered state
survives and there is a unique minimizer. Therefore when crossing
vertically the critical curve the total density jumps, a phenomenon
which can be related to magnetostriction as argued in \cite{DMPV2}.

The Kac proposal applied to \eqref{z1.1} leads to hamiltonians of
the form
    \begin{equation}
      \label{z1.2}
H_{\ga,\la} (q) = \frac 12 \sum_{i\ne j} V_\ga(r_i,r_j) \text{\bf
1}_{s_i\ne s_j} - \la n
     \end{equation}
where   $q=(...,r_i,s_i,...)$, $i=1,..,n$, $r_i\in \mathbb R^d$,
$s_i\in \{1,..,S\}$, is a finite configuration of particles with
spin; $\la$ the chemical potential and $V_\ga(r_i,r_j)= \ga^dV(\ga
r_i,\ga r_j)$, $V(r,r')$ a symmetric probability kernel, say with
range 1. An analysis a la Lebowitz and Penrose, \cite{LP}  (see also
Gates and Penrose, \cite{GP}) proves that the mesoscopic ($\ga \to
0$) behavior of the system with hamiltonian $H_{\ga,\la} (q)$ is
described by
  \begin{equation}
      \label{z1.3}
F_{\beta,\la}(\rho) = \frac 12 \sum_{s,s':s\ne s'}\int
\rho_s(r)V(r,r')\rho_{s'}(r') drdr'- \int  \{\frac {\mathcal
S(\rho(r))} \beta  +\la \sum_s \rho_s(r)\}dr
     \end{equation}
as a functional defined on functions $\rho \in L^\infty(\mathbb R^d,
\mathbb R_+^S)$ with compact support. Let $\La$ a torus in $\mathbb
R^d$, call $F_{\beta,\la,\La}(\rho)$ the functional  \eqref{z1.3} on
$L^\infty(\La, \mathbb R_+^S)$, then obviously
  \begin{equation}
      \label{z1.4}
\inf _{\rho \in L^\infty(\La, \mathbb R_+^S)}F_{\beta,\la,\La}(\rho)
 \le |\La| \inf_{\rho \in \mathbb R_+^S} F^{\rm mf}_{\beta,\la}(\rho)
     \end{equation}
(just restrict the inf on the l.h.s.\ to constant functions).  Thus
a preliminary condition for the particle model to have mean field
behavior is to require that \eqref{z1.4} holds with equality, which
(we suspect) requires extra conditions on $V$.

In \cite{LMP} the Kac proposal has been modified in such a way that
the above condition is automatically satisfied. Call  $e^{\rm
mf}_{\la}(\rho)$ the mean field
 energy, in our case
    \begin{equation}
      \label{z1.5}
e^{\rm mf}_{\la}(\rho) = \frac 12\sum_{s,s':s\ne s'} \rho_s\rho_{s'}
-\la \sum_s \rho_s
     \end{equation}
(i.e.\ the first two terms on the r.h.s.\ of \eqref{z1.1}, the third
one is the contribution of the entropy to the free energy) and set
   \begin{equation}
      \label{z1.6}
H_{\ga,\la} (q) = \int  e^{\rm mf}_{\la}(J_\ga * q(r))\,dr,
\;\;\;\text{$J_\ga * q(r)\in \mathbb R_+^S$, $\dis{ (J_\ga *
q)(r,s)= \sum_i J_\ga(r,r_i) \text{\bf 1}_{s_i=s}}$}
     \end{equation}
$J_\ga(r,r')= \ga^dJ(\ga r,\ga r')$ with $J(r,r')$  a smooth,
symmetric, translational invariant probability kernel with range 1,
$J(r,r')=0$ if $|r-r'|\ge 1$.

Namely the ``modified Kac proposal'' we are adopting is to suppose
that {\em the particle hamiltonian has an energy density at point
$r$ given by the mean field free energy computed on the empirical
density $J_\ga * q(r)$}. Analogous prescription can be applied
whenever the mean field order parameter is a density (or as in this
case a collection of densities).  The free energy functional
associated to \eqref{z1.6} is, supposing $\La$ a torus in $\mathbb
R^d$,
  \begin{equation}
      \label{z1.7}
F_{\beta,\la,\La}(\rho) =     \int_\La \{ e^{\rm mf}_{\la}(J  *
\rho(r)) - \frac{\mathcal  S(\rho(r)) }\beta   \}
     \end{equation}
 which can be rewritten as
   \begin{equation}
      \label{z1.8}
F_{\beta,\la,\La}(\rho) =     \int_\La \{ e^{\rm mf}_{\la}(J  *
\rho(r)) - \frac{\mathcal  S(J*\rho(r)) }\beta   \} + \frac 1\beta
\int_\La \{  \mathcal  S(J*\rho) -J*\mathcal S(\rho) \}
     \end{equation}
By convexity the second integral is non negative and 0 on the
constants; the first one is minimized by taking $\rho(r)$ constantly
equal to the  minimizer of $F^{\rm mf}_{\beta,\la}(\cdot)$.  Thus
\eqref{z1.4} holds in this case with equality.

Notice that the hamiltonian  $H_{\ga,\la} (q)$ of \eqref{z1.6} has
the form \eqref{z1.2} because it can be written as
    \begin{equation}
      \label{z1.9}
H_{\ga,\la} (q) = \frac 12 \sum_{i\ne j}(J_\ga*J_\ga)(r_i,r_j)
\text{\bf 1}_{s_i\ne s_j} - \la n
     \end{equation}
Thus the LMP prescription in this case is just a positivity
assumption on the kernel $V$ (more precisely $V=J_\ga*J_\ga$). In
the sequel we will restrict to the choice
\eqref{z1.6}--\eqref{z1.9}.
 The main result in this paper is

\vskip1cm

    \begin{thm}
    \label{thme1.1}
For any $d\ge 2$, $S\ge 3$, $\beta >0$ there is $\ga^*>0$ and for
any $\ga\le \ga^*$ there are  $\la_{\beta,\ga}$ and $S+1$ DLR
measures at $(\beta,\la_{\beta,\ga})$, denoted by $\mu^{(k)},
k\in\{1,\dots,S+1\}$, with the following properties.

$\bullet$\; Each $\mu^{(k)}$ is a translational invariant, extremal
DLR measure (with trivial $\si$-algebra at infinity);

$\bullet$\; any translational invariant DLR measure is a convex
combination of $\big\{\mu^{(k)}, k\in\{1,..,S+1\}\big\}$;

$\bullet$\; calling $\rho_{\ga,s'}^{(k)},s'\in \{1,..,S\}$ the
average   density of particles with spin $s'$ in $\mu^{(k)}$ and
$\rho^{(k)}_{s}$
 the mean field values, $\dis{\lim_{\ga\to 0}
\rho_{\ga,s'}^{(k)}=\rho^{(k)}_{s'}}$;

$\bullet$\; any measure $\mu^{(k)}$, $k\le S$, is invariant under
any exchange of spin labels which does not involve $k$ while
$\mu^{(S+1)}$ is invariant under any  exchange of spin labels.
    \end{thm}

\vskip1cm

The proof of  Theorem \ref{thme1.1}  uses specific features of the
model besides  the property that \eqref{z1.4} is true with equality.
Which properties are of general nature and which ones are instead
truly specific of the model is difficult to say. To a great extent
the proof follows from the analysis (a la Pirogov-Sinai) of the  LMP
model in Chapter 11 and 12 of \cite{leipzig}, but there are several
points where we need to overcome important difficulties not present
in the  LMP model. Among them the main one is about the  exponential
decay of correlations in the restricted ensemble, Theorem 3.1 of the
companion paper \cite{DMPV2}. How to go from such a result to the
proof of Theorem  \ref{thme1.1} is the content of the present paper.

Theorem \ref{thme1.1} does not claim anything away from
$(\beta,\la_{\beta,\ga})$, this allows to simplify the traditional
Pirogov-Sinai approach.  The conjecture is that when $\la$ varies in
$(\la_{\beta,\ga}-\eps,\la_{\beta,\ga}+\eps)$, $\eps>0$ suitably
small, then we go from uniqueness $\la<\la_{\beta,\ga}$ to $S$
extremal states, $\la>\la_{\beta,\ga}$,
 (always referring to translational invariant DLR states).
 The Potts model does not exactly fall in the class
 considered in \cite{BMPZ} but presumably the analysis
 in \cite{BMPZ} can be extended to prove the above conjecture.
 It is also plausible that the estimates are
uniform in a small neighborhood of $\beta$, in such a case we would
have local  closeness of the mean field and the finite $\ga$  phase
diagrams, thus partially confirming the validity of the conjecture
 in the beginning of the introduction.

\vskip.5cm

In Part I we define the model and establish the main notation,
Section \ref{sec:2}, and then prove Theorem \ref{thme1.1}, Section
\ref{sec:z31}, supposing that the Peierls estimates on contours are
valid.  In Part II we prove the Peierls estimates, this being the
more technical part of the paper.

\vskip1.5cm

\part{}

\section{{\bf Main notation and definitions}}
        \label{sec:2}
We start with the basic definitions. They are quite standard and
consistent with those of the companion paper \cite{DMPV2}.

 \vskip.5cm

\subsection{Geometrical notions}
 \label{subsecZ2.0}

We give the following definitions.

$\bullet$\;  {\em The partitions  $\mathcal D^{(\ell)}$.}

We denote by $\mathcal D^{(\ell)}$, $\ell>0$, the partition
$\{C^{(\ell)}_x, \, x\in \ell \mathbb Z^d\}$ of
$\mathbb R^d$ into the cubes $C^{(\ell)}_x=\{r\in \mathbb R^d: x_i
\le r_i <x_i+\ell, i=1,..,d\}$ ($r_i$ and $x_i$ the cartesian
components of $r$ and $x$). We call $C^{(\ell)}_r$ the cube which
contains $r\in\mathbb R^d$.

\vskip.5cm

$\bullet$\;  {\em  $\mathcal D^{(\ell)}$-measurable sets and functions.}

A set is $\mathcal D^{(\ell)}$-measurable if it is union
of cubes in $\mathcal D^{(\ell)}$. A function $f:\mathbb R^d\times\{1,\dots,S\}\to \mathbb R$
is $\mathcal D^{(\ell)}$-measurable if its inverse images are
$\mathcal D^{(\ell)}$-measurable sets, or, equivalently, if it is constant on the cubes of  $\mathcal D^{(\ell)}$.

\vskip.5cm

$\bullet$\;  {\em  $\mathcal D^{(\ell)}$-boundaries of a set.}

Calling two sets connected if their
closures have non empty intersection,
given a $\mathcal
D^{(\ell)}$-measurable region $\La$ we call $\delta_{\rm
out}^{\ell}[ \La]$ the union of all cubes of $\mathcal D^{(\ell)}$
in $\La^c\equiv\mathbb R^d\setminus \La$ which are connected to
$\La$. Analogously we call $\delta_{\rm in}^{\ell}[ \La]$ the union
of all cubes of $\mathcal D^{(\ell)}$ in $\La$ which are connected
to $\La^c$.

\vskip1cm

\subsection{Phase space, topology and free measure}
 \label{subsecZ2.1}
We start with the definition of the phase space.

$\bullet$\;  {\em The phase space $\mathcal{Q}$.}

It is convenient to represent the phase space $\mathcal{Q}$ of the Potts model as
a spin system on the lattice, the spins taking values in a non compact
space.  With $\{C_i\equiv C_i^{(1)}, i\in \mathbb Z^d\}$  the cubes of the partition $\mathcal D^{(1)}$,  we then define
$\dis{\mathcal{Q}_{C_i}:=\bigcup_{n\ge
0}(C_i \times\{1,\dots,S\})^n}$ and $\mathcal{Q} = \prod _i \mathcal{Q}_{C_i}$.

Thus an element $q\in \mathcal Q$ is a sequence  $q^{(i)}\in  \mathcal{Q}_{C_i}$, if
$q^{(i)}=(r_1,s_1,\dots,r_n,s_n)$ we will then say that in $C_i$ there are $n$ particles at positions $r_j$
 with spins $s_j$, $j=1,..,n$.  As particles are undistinguishable,
 physical observables
are functions symmetric under exchange of particles labels and the
actual physical phase space is $\mathcal{Q}^{\rm sym} = \prod _i
\mathcal{Q}_{C_i}^{\rm sym}$ which is obtained by taking  the
quotient  under permutation of indices.  To simplify notation in the
sequel we will  just write  $\mathcal{Q}$ being clear from the
context if we are referring to $\mathcal{Q}^{\rm sym}$. Since labels
are unimportant we can write a configuration $q\in \mathcal Q$  as a
sequence $q=\{...,r_j,s_j,...\}$, $r_j\in \mathbb R^d$, $s_j\in
\{1,..,S\}$, indeed $q_{C_i}=q\cap C_i$, namely the set of all
$(r_j,s_j)\in q: r_j\in C_i$, identifies the component of $q$ in $
\mathcal{Q}_{C_i}$. Given $q=\{...,r_j,s_j,...\}\in \mathcal Q$ we
write $q(s)=\{(r_j,s_j)\in q: s_j=s\}$ and we call
$\mathcal{Q}_\La:=\{q_\La\in\mathcal{Q}: q_\La=\{(r_j,s_j)\in q:
r_j\in \La\} \}$.  We finally denote by $q\cup q'$ the configuration
which collects all the particles in $q$ and $q'$, evidently
referring here to indistinguishable particle configurations.

\vskip.5cm

$\bullet$\;  {\em Topological properties of $\mathcal{Q}$.}

We consider   $\dis{\mathcal{Q}_{C_i}=\bigcup_{n\ge 0}(C_i
\times\{1,\dots,S\})^n}$ equipped with its natural topology and
$\mathcal{Q}$ with the product topology  calling $\Si$ the
corresponding Borel $\si$-algebra.

While  the product topology in $\mathcal{Q}$ is not physically correct (the path of a particle moving
continuously from a cube $C_i$ to another one is not continuous in the product topology) yet the Borel
structure is not changed and since we are interested in measure theoretically properties the above definition
becomes acceptable.

\vskip.5cm

$\bullet$\;  {\em The free measure.}

We denote by $\nu(dq_{C_i})$ the measure on $\mathcal Q_{C_i}$ which
restricted to $(C_i \times\{1,.,S\})^n$ is equal to $(n!)^{-1}
dr_1.. dr_n$, such that if $f$ is a bounded measurable function on
 $\mathcal Q_{C_i}$
      \begin{equation*}
\int_{\mathcal{Q}_{C_i}}f(q)\nu(dq_{C_i})=\sum_{n=0}^\infty \frac
1{n!}\sum_{s_1,\dots, s_n}\int_{C_i^n} f(r_1,s_1,\dots,
r_n,s_n)dr_1\dots dr_n
    \end{equation*}
If $\La$ is a bounded $\mathcal D^{(1)}$ measurable region we define {\em the free measure}
$\dis{\nu(dq_{\La}) =\prod_{i\in \La\cap \mathbb Z^d}\nu(dq_{C_i})}$ on $\mathcal Q_\La$ observing that
for any measurable set $\Delta \subseteq \La$
       \begin{equation}
    \label{integ}
\int_{\mathcal{Q}_\Delta}f(q)\nu(dq_\La)=\sum_{n=0}^\infty \frac
1{n!}\sum_{s_1,\dots, s_n}\int_{\Delta^n} f_n(r_1,s_1,\dots,
r_n,s_n)dr_1\dots dr_n
    \end{equation}

 \vskip1cm
\subsection{Energy and Gibbs measures}
 \label{subsec1.1}

We have already defined the  energy $H_{\ga,\la} (q)$ (of a finite
configuration), see \eqref{z1.6}.
The energy in a bounded set $\La$ with boundary condition $\bar
q_{\La^c}$  is defined as usual as
   \begin{equation}
      \label{1.3}
H_{\La,\la} (q_\La|\bar q_{\La^c}) = H_\la(q_\La\cup \bar q_{\La^c})
- H_\la(\bar q_{\La^c})= \int_{\mathbb R^d}\big[ e_\la^{\rm
mf}\big(J_\ga\star (q_\La\cup q_{\La^c}) \big)-e_\la^{\rm
mf}\big(J_\ga\star q_{\La^c}\big)\big]dr
     \end{equation}
The expression on the r.h.s.\ depends only on the particles of $\bar
q_{\La^c}$ at distance $\le2 \ga^{-1}$.

In the sequel we will sometimes replace $\bar q_{\La^c}$ by
$\si$-finite measures by
setting
   \begin{equation}
    \label{del}
 J_\ga\star d\mu(r,s)=\int_{\mathbb R^d}J_\ga(r,r')d\mu(r',s)
    \end{equation}
$d\mu(r,s)$ any non negative $\si$-finite measure on
$\mathbb R^d\times\{1,\dots,S\}$.
By identifying  $q\in \mathcal{Q}$
as a sum of Dirac deltas
 we may regard the convolution  $J_\ga\star q$ as a particular case of \eqref{del}.
 In particular we will often consider
      \begin{equation}
    \label{z2.7}
   H_{\La,\la} (q_\La|\chi^{(k)}_{\La^c})=  \int_{\mathbb R^d}\big[ e_\la^{\rm
mf}\big(J_\ga\star q_\La + J_\ga\star \chi^{(k)}_{\La^c})
\big)-e_\la^{\rm mf}\big(J_\ga\star \chi^{(k)}_{\La^c}\big)\big]dr
    \end{equation}
where $\chi^{(k)}_{\La^c}:=\rho^{(k)} \mathbf 1_{\La^c}$ were
$\rho^{(k)}=(\rho^{(k)}_1,\dots \rho^{(k)}_S)$,
$k\in\{1,\dots,S+1\}$ is one of the minimizers of the mean field
free energy $F^{\rm mf}_{\beta,\la_\beta}$.

The Gibbs measure in $\La$ ($\La$ a bounded, measurable set in
$\mathbb R^d$) with boundary conditions $\bar q$ is
      \begin{equation}
    \label{z2.8}
  G_{\la,\La,\bar q} (dq_\La)= Z^{-1}_{\La,\bar q}
  e^{-\beta H_{\La,\la} (q_\La|\bar q_{\La^c})}d\nu(q_\La)
    \end{equation}
where the partition function $Z_{\La,\bar q}$  is the normalization
factor in \eqref{z2.8}. We will also consider more general boundary
conditions with $\bar q$ replaced by $\si$-finite measure, the
formula is again
 \eqref{z2.8} with the energy defined using \eqref{del}.

\vskip1cm
\subsection{Phase indicators and restricted phase space}
    \label{sec:a2.2bis}

\vskip.5cm

As usual in statistical mechanics local equilibrium and deviations
from equilibrium  are defined in terms of ``averages'' and of
``coarse grained'' variables. We briefly recall the main notion
adapted to the present context. Given a configuration $q\in
\mathcal{Q}$ we denote by $n^{(\ell)}(r,s;q)$ the number of
particles in the configuration $q$ which are in the cube
$C_r^{(\ell)}$ and have spin $s$, namely
    \begin{equation}
      \label{a3.1b}
 n^{(\ell)}(r,s;q):= |q(s)\cap C_r^ {(\ell)}|,\qquad s\in \{1,\dots,S\}
     \end{equation}
 We also define the density of particles in $C_r^{\ell}$,
    \begin{equation}
      \label{a3.1}
 \rho^ {(\ell)}(r,s;q):= \frac{n^{(\ell)}(r,s;q)} {\ell^d},\qquad s\in \{1,\dots,S\}
     \end{equation}

The phase indicators are introduced using two scales $\ell_-$
and $\ell_+$ and an accuracy parameter $\zeta$. All these numbers
depend on $\ga$ and there is much flexibility about their choice,
for the sake of definitiveness we fix them as follows:

 \begin{defin} (Choice of parameters).
    \label{contorni3}
{\sl We choose $\ell_-$ and $\ell_+$ as  functions of
$\gamma$:
    \begin{equation}
    \label{elle}
\ell_{-}=\gamma^{-1+\alpha_-},\quad
\ell_{+}=\gamma^{-1-\alpha_+},\qquad \alpha_-< 1,\quad\alpha_+<1
    \end{equation}
supposing for simplicity that $\ga^{-1}$ and
$\ga^{-(1\pm\alpha_\pm)}$ are both in $\{2^n,n\in \mathbb N\}$.
We
also choose
$$\zeta=\gamma^a,\qquad a<1$$}

We require that $\alpha_+>\alpha_->a$ and
    \begin{equation}
 2ad +\alpha_- d^2 <\frac{\alpha_+}{2},\qquad \frac 12
 -2d\alpha_+>0,\qquad \frac 14
 -d(\alpha_+-\alpha_-)>0
 \label{aI.3.11}
     \end{equation}
    \begin{eqnarray}
\frac{\alpha_++\alpha_-}{1-\alpha_-}<\frac 1 d,\qquad
4(\alpha_++\alpha_-)+\frac {\alpha_- }{2}<\frac 1 4
 \label{aI.3.11b}
     \end{eqnarray}
            \end{defin}

    \vskip1cm
 Thus for $\ga$ small, $\ell_-$ is much larger than 1 and much
 smaller than $2\ga^{-1}$ equal to the range of the interaction; it defines  a
 scale large enough to make statistics reliable. Indeed, the scale
$\ell_-$ is used together with the accuracy parameter $\zeta$ to
determine if a configuration (or a density) is close to a mean field
equilibrium value in a cube $C^{(\ell_-)}$. This will be done via
the  {\it phase indicator} that we denote by $\eta$. Local
equilibrium is instead present when the above closeness extends to
regions in the scale $\ell_+$ thus regions with a diameter much
larger than the interaction range. To quantify the local equilibrium
we use the {\it phase indicator on the scale  $\ell_+$} that we
denote by $\Theta$.

\vskip.5cm

For any  $\rho\in L^1(\mathbb R^d\times\{1,\dots,S\})$ we then
define in analogy to \eqref{a3.1}
    $$
 \rho^{(\ell)}(r,s)=\mintone{C^{(\ell_-)}_r} \rho(x,s) dx,
 \qquad \mintone{A}:=\frac 1{|A|}\int_A
    $$
and
   \begin{equation}
      \label{I.1.1}
\eta^{(\zeta,\ell_-)}(r;\rho) = \begin{cases} k &\text{if $|
 \rho^{(\ell)}(r,s) -\rho^{(k)}_s |\le \zeta$},\forall s
\\0&\text{otherwise}\end{cases}
     \end{equation}
and
          \begin{equation}
          \label{Theta}
\Theta^{(\zeta,\ell_-,\ell_+)}(r;\rho) =
    \begin{cases} k & \text{if $\eta^{(\zeta,\ell_-)}(\cdot;\rho)=k$
    in  $ C^{(\ell_+)}_{r} \cup \delta_{\rm
    out}^{\ell_+}[C^{(\ell_+)}_{r}], $}
 \\0 & {\rm otherwise}.
    \end{cases}
     \end{equation}
Recalling \eqref{a3.1}, the previous definitions  extend to particle
configurations $q$ by setting
  \begin{equation}
      \label{aI.1.1}
\eta^{(\zeta,\ell_-)} (r;q) =
\eta^{(\zeta,\ell_-)}(r;\rho^{(\ell_-)}(q;\cdot)), \qquad
\Theta^{(\zeta,\ell_-,\ell_+)}(r;q)=\Theta^{(\zeta,\ell_-,\ell_+)}
(r;\rho^{(\ell_-)}(q;\cdot))
     \end{equation}

We will often drop the suffix ${(\zeta,\ell_-)}$ by writing $\eta$
instead of $\eta^{(\zeta,\ell_-)}$, analogously for $\Theta$.

\vskip1cm

 Given $\La\subseteqq\mathbb{R}^d$ and $k\in\{1,..,S+1\}$,
we define ``the $k$-restricted ensemble'' as
    \begin{equation}
      \label{a3.4bis}
\mathcal X_\La^{(k)}\equiv \mathcal
X_\La^{(k)}(\zeta,\ell_-):=\Big\{ q: \eta^{(\zeta,\ell_-)}(r;q)=k,
\,\,\forall r\in \La\Big\}
     \end{equation}
If $\La=\mathbb{R}^d$ we simply write $\mathcal X^{(k)}$. By an
abuse of notation we also denote by $\mathcal X_\La^{(k)}$ the space
of densities $\rho$ such that $\eta(\cdot;\rho)=k$ in $\La$.

\vskip.5cm
\subsection{Colored Contours}
    \label{sec:3.1}

 First observe that
    $$
 \{r:\Theta^{(\zeta,\ell_-,\ell_+)}(q;r)=k\} \cap
 \{r:\Theta^{(\zeta,\ell_-,\ell_+)}(q;r)=h\}=\emptyset,\qquad h\ne k
    $$
In fact the two regions are separated by the set
$\{r:\Theta^{(\zeta,\ell_-,\ell_+)}(q;r)=0\}$ which is what we call
spatial support of a contour. Given a configuration $q$ such that
$\{r\in \mathbb R^d: \Theta(r;q)=0\}$ is bounded, we call contour a
pair $\Ga=({\rm sp}(\Ga),\eta_\Ga)$ where ${\rm sp}(\Ga)$, the
spatial support of $\Ga$ is
    \begin{equation}
    \label{spga}
{\rm sp}(\Ga)= \text{ maximal connected component of the region
}\{r\in \mathbb R^d: \Theta(r;q)=0\}
    \end{equation}
  and $\eta_\Ga(r)=\eta(r;q),
r\in {\rm sp}(\Ga)$, its specification. Abstract contours  $\Ga$ are
the pairs which arise from some configuration as above.

 We decompose the complement of sp$(\Ga)$ as
$\text{sp}(\Ga)^c = \text{ext}(\Ga) \cup \text{int}(\Ga)$ where
$\text{ext}(\Ga)$ is the unbounded, maximal connected component of
$\text{sp}(\Ga)^c$. We denote by
    \begin{equation}
    \label{ngamma}
 N_\Ga=\frac{|{\rm sp}(\Ga)|}{\ell_+^d}, \;\;c(\Ga) = {\rm
 sp}(\Ga)\cup {\rm int}(\Ga)
    \end{equation}

We omit the proof of the following proposition (which is a
straightforward consequence of the definition of phase indicators
and contours).

\vskip.5cm

  \begin {prop}
  \label{thm2.6}
Suppose $q$ has a  contour $\Ga$, then there is $k\in \{1,..,S+1\}$
such that $\Theta(r;q)=k$ for all $r\in \delta_{\rm
out}^{\ell_+}[c(\Ga)]$, $c(\Ga)$ as in \eqref{ngamma}. Moreover if
$\La$ is any maximal connected component of ${\rm int}(\Ga)$ then
there is  $h\in \{1,..,S+1\}$ such that $\Theta(r;q)=h$ for all
$r\in \delta_{\rm in}^{\ell_+}[\La]$.

  \end{prop}

\vskip.5cm

Proposition \ref{thm2.6} implies that given any $\Ga$,
$\Theta(r;\bar q)$, $r\in \delta_{\rm{out}}[{\rm sp}(\Ga)]$ is
determined by $\eta_\Ga$ and assumes the same value for all  $\bar
q\in\{q:\Ga$ is a contour for $q\}$. We will then say that $\Ga$ has
color $k$ if $\Theta(r;q)=k$ for all $r\in \delta_{\rm
out}^{\ell_+}[c(\Ga)]$ and denote by ${\rm int}_{(h)}(\Ga)$ the
union of the maximal connected components $\La_i$ of int$(\Ga)$
where $\Theta(r;q)=h$ for all $r\in \delta_{\rm
in}^{\ell_+}[\La_i]$.

Given a color $k$ and a bounded, simply connected
$D^{(\ell_+)}$-measurable region $\La$, we denote by $\mathcal
B^{k}_\La$  the collection of all sequences $\und
\Ga=(\Ga(1),..,\Ga(n))$ of contours of color $k$ with spatial
support in $\La\setminus \delta_{\rm in}^{\ell_+}[\La]$ and such
that the  spatial supports  are mutually  disconnected.

\vskip2cm

\section{{\bf Proof of Theorem \ref{thme1.1}}}
    \label{sec:z31}

\subsection{The main technical result}
    \label{sec:z3}

From a technical point the main results in this paper are Theorem
\ref{thmz3.1} and
 Theorem \ref{thmz3.2} below. Their statements involve
 the notion of $k$-boundary conditions, diluted Gibbs measure and diluted partition functions.

 \vskip.5cm

$\bullet$\; {\em $k$-boundary conditions.}

Let  $k\in\{1,\dots S+1\}$ and $\La$  a bounded  $\mathcal D^{(\ell_+)}$-measurable
region.  A configuration $\bar q$ is
a $k$-boundary condition relative to
$\La$ if there is a
configuration $q^{(k)}\in \mathcal X^{(k)}$  (see \eqref{a3.4bis})
which is equal to  $\bar q$  in the
region $\{r\in \La^c: {\rm dist}\,(r,\La)\le 2\ga^{-1}\}$.

 \vskip.5cm

$\bullet$\; {\em  Diluted Gibbs measure and partition function.}

Let $\La$, $\bar q$  and  $q^{(k)}$ as above, then
the $k$-diluted Gibbs measure in $\La$ with b.c.\
$\bar q$ is
   \begin{equation}
    \label{z3.1}
G_{\la,\La,\bar q}^{(k)}(dq_\La):=\frac {e^{-\beta H_{\La,\la}
(q_\La| \bar q_{\La^c}) }}{  Z_{\la,\La,k}(\bar q)}
 \text{\bf
1}_{\{\Theta(q_\La\cup q^{(k)}_{\La^c};r)=k,\;\text{$r\in \delta_{\text{\rm
in}}^{\ell_+} [\La]$}\}}\,\, \nu(dq_\La)
     \end{equation}
where
     \begin{equation}
    \label{z3.1a}
 Z_{\la,\La,k}(\bar q)= \int_{\{\Theta(q_\La\cup
q^{(k)}_{\La^c};r)=k,\;\text{$r\in \delta_{\text{\rm in}}^{\ell_+}
[\La]$}\}} e^{-\beta H_{\La,\la} (q_\La| \bar q_{\La^c}) }\,
\nu(dq_\La)
      \end{equation}
is the diluted partition function.

\vskip1cm

    \begin{thm}
        \label{thmz3.1}
For any $\beta $ there are $c^*$, $\ga_\beta>0$ and for any $\ga\le
\ga_\beta$ there is $\la_{\beta,\ga}$ such that  for any bounded,
simply connected, $\mathcal D^{(\ell_{+})}$ measurable region $\La$,
any $k$- boundary conditions $\bar q$ and any $r\in \La$,
    \begin{equation}
    \label{z3.2}
G_{\la_{\beta,\ga},\La,\bar q}^{(k)}(\{\Theta^{(\zeta,\ell_-,\ell_+)}(q;r)=k\}) \ge
\; 1\;-\; \exp\Big\{ -\beta \frac{c^*}{4} \, (\zeta^2
\ell_{-}^d)\Big\}
     \end{equation}
\end{thm}

\vskip1cm

The proof of  Theorem \ref{thmz3.1} is a corollary of Theorem
\ref{thmz3.2} below,  which involves the fundamental notion of
contour weights:

 \vskip.5cm

$\bullet$\; {\em The true weight of a contour.} \nopagebreak

 Given a $k$-colored contour $\Ga$ and a $k$-boundary condition
 $q^{(k)}$ relative to $c(\Ga)$,
 we define the ''true'' weight $W^{k,{\rm true}}(\Ga|q^{(k)})$ as
         \begin{equation}
        \label{z3.3}
W^{k, {\rm true}}(\Ga|q^{(k)})= \frac{\dis{ \int_{\Upsilon_{{\rm
sp}(\Ga)}(\eta_\Ga)} e^{-\beta H_{{\rm sp}(\Ga),\la}(q_{{\rm
sp}(\Ga)}|q_{c(\Ga)^c})}
 \prod_{j=1}^pZ^{(k_j)}_{{\rm
int}_{j}(\Ga),\la}(q_{{\rm sp}(\Ga)}) \nu(dq_{{\rm
sp}(\Ga)})}}{\dis{ \int_{\mathcal X^{(k)}_{{\rm sp}(\Ga)}} e^{-\beta
H_{{\rm sp}(\Ga),\la}(q_{{\rm sp}(\Ga)}|q_{c(\Ga)^c})}
 \prod_{j=1}^pZ^{(k)}_{{\rm
int}_{j}(\Ga),\la}(q_{{\rm sp}(\Ga)})  \nu(dq_{{\rm sp}(\Ga)})}}
    \end{equation}
where
    $
 \Upsilon_{{\rm
sp}(\Ga)}(\eta_\Ga):=\{q_{{\rm sp}(\Ga)} : \eta(r;q_{{\rm
sp}(\Ga)})=\eta_\Ga(r),  r\in {\rm sp}(\Ga)\}
    $;
 ${\rm int}( \Ga)$ decomposes into $p$ maximal connected components
 ${\rm int }_{j}(\Ga), j=1,..,p$;
 $k_j$ denotes the
value of $\Theta$ on $\delta_{\rm in}^{\ell_+}[{\rm int
}_{j}(\Ga)]$.

 \vskip.5cm

The above are called true weights to distinguish them from fictitious weights
introduced in the proof of Theorem \ref{thmz3.2}.

 \vskip.5cm

    \begin{thm}
        \label{thmz3.2}
        In the same context of Theorem \ref{thmz3.1}, for all $\ga$ small
        enough and recalling definition \eqref{ngamma},
            \begin{equation}
    \label{z3.4}
W^{k, {\rm true}}(\Ga|q^{(k)}) \le
  \exp\Big\{ -\beta \frac{c^*}{2} \, \zeta^2
 \ell_{-} ^d \,N_\Ga\Big\}
     \end{equation}
        \end{thm}

 \vskip.5cm

As already pointed out Theorem \ref{thmz3.2} is the main technical
result in this paper, its proof follows the Pirogov-Sinai strategy
and it is reported in Part II.  We will next show that  Theorem
\ref{thmz3.1} follows from  Theorem \ref{thmz3.2}.

 \vskip.5cm

 {\bf Proof of Theorem \ref{thmz3.1}} (using Theorem \ref{thmz3.2}).
 By  definition
the $k$-diluted Gibbs measures in $\La$
have support on configurations where $\Theta=k$ on
$\delta_{\rm in}^{\ell_+}[\La]$.
Therefore if $\Theta(r;q_\La)\ne k$ there must be a contour
$\Ga$ such that $r\in c(\Ga)$.
Thus $G_{\la_{\beta,\ga},\La,\bar q}^{(k)}
(\{\Theta^{(\zeta,\ell_-,\ell_+)}(q;r)\ne
k\}) $ is bounded by
            \begin{equation*}
\sum_{c(D)\ni r} (S+2)^{(\ell_+/\ell_-)^d N_D} e^{ -\beta (c^*/2) \, \zeta^2
\ell_{-}^d \,N_D}
     \end{equation*}
where $D$ ranges over all possible values of ${\rm sp}(\Ga)$ such that $c(\Ga)\ni r$; $N_D$
is the number of $\mathcal D^{(\ell_+)}$ cubes in $D$. $(S+2)$ is the number of possible values of $\eta(\cdot)$,
$(\ell_+/\ell_-)^d N_D$ the number of $\mathcal D^{(\ell_-)}$ cubes in $D$.  The above is bounded by
            \begin{equation*}
 e^{ -\beta (c^*/4) \, \zeta^2
\ell_{-}^d } \sum_{c(D)\ni r} (S+2)^{(\ell_+/\ell_-)^d N_D} e^{ -\beta (c^*/4) \, \zeta^2
\ell_{-}^d N_D}
     \end{equation*}
The sum vanishes as $\ga\to 0$, see for instance the proof of Theorem 9.2.8.1 in \cite{leipzig},
 such that for $\ga$ small enough the above is bounded by
$e^{ -\beta (c^*/4) \, \zeta^2
\ell_{-}^d }$.  \qed

\vskip.5cm

In the following sections we will see that the proof of Theorem
\ref{thme1.1} follows  from Theorem \ref{thmz3.2} and Theorem 3.1 of
\cite{DMPV2} using   the same  general arguments  as in
\cite{leipzig} for the analogous proof in the LMP model.  In Part II
we will prove Theorem \ref{thmz3.2}.

\vskip2cm

\subsection{{Existence of DLR measures}}
        \label{sec:z4}

A probability $\mu$ on $\mathcal Q$ is DLR at $(\beta,\la)$ if for
any bounded, measurable cylindrical function $f$ and any bounded
measurable set $\La\subset \mathbb R^d$,
        \begin{equation}
    \label{z4.0.1}
\mu(f)= \mu\big( G_{\la,\La,q}(f)\big):=\int_{\mathcal
Q}G_{\la,\La,q}(f) \mu(dq)
    \end{equation}
We fix $\beta$  and set by default $\ga<\ga_\beta$ and
$\la=\la_{\beta,\ga}$, see Theorem \ref{thmz3.1}.  $\beta$ and $\la$
  in the sequel will be often omitted from the notation.  We will start by proving:

\vskip.5cm

   \begin{thm}
   \label{thmz4.0.1}
The set of DLR measures at $(\beta, \la_{\beta,\ga})$ is a non empty, convex, weakly compact set.
  \end{thm}

\vskip.5cm

The proof is made  simpler by the assumption that the interaction is
non negative.  We follow closely Section 12.1 of \cite{leipzig}
where the analogous statement is proved for the LMP model and where
the reader may find more details.   The basic estimate is
\eqref{z4.0.3} below.  Let $C\in \mathcal D^{(1)}$, $G_{ C,\bar q}$
the Gibbs measure on $\mathcal Q_C$ at  $(\beta, \la_{\beta,\ga})$
with boundary conditions $\bar q$,
   \begin{equation}
    \label{z4.0.2}
A_{\delta,N,C}:=\Big\{q\in \mathcal Q_C: |q|\le N, q\cap \{r\in
C:{\rm dist}(r,C^c)<\delta\}=\emptyset\Big\}
     \end{equation}
     Then, using the non negativity of the interaction,
   \begin{equation}
    \label{z4.0.3}
G_{C,\bar q}( A^c_{\delta,N,C})\le \eps_{\delta,N}:
= \sum_{n=1}^N \frac{S^n ( 2d \delta
n)}{n!} + \sum_{n>N} \frac{ S^n}{n!}
     \end{equation}
and therefore there are $n^*$ and $\delta_n>0$ (decreasing with $n$)
such that for any $\bar q\in \mathcal Q$,
   \begin{equation}
    \label{z4.0.4}
G_{ C,\bar q}( A^c_{\delta_n,n,C})\le e^{-n},\quad \text{for all $n\ge n^*$}
     \end{equation}
By supposing (without loss of generality) $n^*$ large enough, there
exist configurations $q^{(k)}\in \mathcal X^{(k)}$, $k=1,..,S+1$,
such that
        \begin{equation}
    \label{z4.0.4bis}
q^{(k)}\in \bigcap_{i\in \mathbb Z^d} A_{\delta_{n^*},n^*,C_i}
     \end{equation}

Call  $M(\mathcal Q)$ the set of all probabilities on $\mathcal Q$ and
   \begin{equation}
    \label{z4.0.5}
M^0(\mathcal Q):=\Big\{\mu\in M(\mathcal Q): \mu( A^c_{\delta_{n+|i|},n+|i|,C_i})\le e^{-(n+|i|)},\;\; \text{for all $n\ge n^*$ and all
$i\in \mathbb Z^d$}\Big\}
     \end{equation}
Define also for any bounded $\mathcal D^{(1)}$-measurable set $\La\subset \mathbb R^d$,
   \begin{equation}
    \label{z4.0.6}
\mathcal G_\La:=\Big\{\mu\in M(\mathcal Q): \mu( f)= \mu
\big(G_{\La,q}(f)\Big\},\quad \mathcal G^0_\La = \mathcal G_\La \cap
M^0(\mathcal Q)
     \end{equation}
If $\La$ is $\mathcal D^{(\ell_+)}$-measurable and $q^{(k)}$ as in \eqref{z4.0.4bis}, then
by \eqref{z4.0.4}
 $G_{\La,q^{(k)}} \in  \mathcal G^0_\La$ which is therefore non empty.  A stronger
 statement actually holds:

 \vskip.5cm

 \begin{lemma}
 \label{lemmaz4.0.1}
 $\mathcal G^0_\La$ is a non empty, convex, weakly compact set and if $\Delta\subset \La$ then
 $\mathcal G^0_\La \subset  \mathcal G^0_\Delta $.

 \end{lemma}

 \vskip.5cm

 {\bf Proof.}  For any $n\ge n^*$ the set $\dis{\bigcap_{i\in \mathbb Z^d} A_{\delta_{n+|i|},n+|i|,C_i}}$ is compact
 and if $\mu \in M^0(\mathcal Q)$,
    \begin{equation}
    \label{z4.0.7}
\mu\Big( \bigcap_{i\in \mathbb Z^d}
A_{\delta_{n+|i|},n+|i|,C_i}\Big)\ge  1 - c e^{-n},\quad
c:= \sum_{i\in \mathbb Z^d}e^{-|i|}
     \end{equation}
Then, by the Prohorov theorem, the weak closure of  $\mathcal
G^0_\La$ is weakly compact.  Since $A^c_{\delta_n,n,C}$ is closed,
the inequalities $\mu( A^c_{\delta_n,n,C})\le e^{-n}$ are preserved
under weak limits such that   $\mathcal G^0_\La$ is weakly closed,
hence weakly compact.  Convexity and the inclusion
 $\mathcal G^0_\La \subset  \mathcal G^0_\Delta $
 are obvious and the lemma is proved.  \qed

 \vskip.5cm

 \begin{coro}
 \label{coroz4.0.1}
Let $\La_n$ be an increasing sequence of  $\mathcal
D^{(1)}$-measurable sets invading $\mathbb R^d$, then\\
$\dis{\mathcal G^0:= \bigcap_n \mathcal G^0_{\La_n}}$  is a non
empty, convex, weakly compact set independent of the sequence
$\La_n$.

 \end{coro}

 \vskip.5cm

 \begin{lemma}
 \label{lemmaz4.0.2}
Any measure   in $\mathcal G^0$ is  DLR and any DLR measure is in $\mathcal G^0$.

 \end{lemma}

 \vskip.5cm

 {\bf Proof.} Let $\Delta$ be a bounded, measurable (but not necessarily $\mathcal D^{(1)}$-measurable) set, and
 $\La\supset \Delta$ a bounded  $\mathcal D^{(1)}$-measurable set.  Then if $\mu\in  \mathcal G^0$,
 $\mu \in  \mathcal G^0_\La$ and since $G_{\La,\bar q}(f)=G_{\La,\bar q}\big(G_{\Delta, q}(f)\big)$,
 it then follows that $\mu(f)=\mu\big(G_{\Delta, q}(f)\big)$, hence that $\mu$ is DLR.  Viceversa if
  $\mu$ is DLR then  by \eqref{z4.0.4} and the DLR property, $\mu \in M^0(\mathcal X)$. By \eqref{z4.0.1}
$\mu \in  \mathcal G_\La$, hence $\mu \in \mathcal G^0_\La$ and by
the arbitrariness of $\La$
  in $\mathcal G^0$.   \qed

 \vskip.5cm

  Corollary \ref{coroz4.0.1} and Lemma \ref{lemmaz4.0.2} prove Theorem \ref{thmz4.0.1}.  Moreover

 \vskip.5cm

 \begin{thm}
 \label{thmz4.0.2}
Let $\La_n$ be an increasing sequence of $\mathcal
D^{(\ell_+)}$-measurable regions invading $\mathbb R^d$ and
$q^{(k)}$, $k=1,..,S+1$, configurations satisfying
\eqref{z4.0.4bis}.  Then $G^{(k)}_{\La_n,q^{(k)}}$ converges weakly
to a measure $\mu^{(k)}\in \mathcal G^0$ and (with $c^*$ as in
Theorem \ref{thmz3.1})

         \begin{equation}
        \label{z4.0.8}
\mu^{(k)} \Big(\{ \Theta(\cdot;r) = k\}\Big) \ge 1-e^{-\beta (c^*/4)\zeta^2\ell_-^d},\quad
\text{for any $r\in \mathbb R^d$}
    \end{equation}
 \end{thm}

 \vskip.5cm

 {\bf Proof.}  Call $\La'= \La_n \setminus \delta_{\rm in}^{\ell_+}[\La_n]$,
$\Delta_n= \La' \setminus \delta_{\rm in}^{\ell_+}[\La']$.  Then by  \eqref{z4.0.4} and \eqref{z4.0.4bis}
$G^{(k)}_{\La_n,q^{(k)}}\in \mathcal G^0_{\Delta_n}$.  Since $\Delta_n$ is increasing,
 by Lemma \ref{lemmaz4.0.1} for $n\ge m$,
$G^{(k)}_{\La_n,q^{(k)}}\in \mathcal G^0_{\Delta_m}$  which is weakly compact. Then
$G^{(k)}_{\La_n,q^{(k)}}$ converges weakly
by subsequences to an element $\mu^{(k)}$ of $ \mathcal G^0_{\Delta_m}$.  Thus $\dis{\mu^{(k)} \in \bigcap_m \mathcal G^0_{\Delta_m}}$
and by Corollary \ref{coroz4.0.1} $\mu^{(k)} \in \mathcal G^0$. \eqref{z4.0.8} follows because
it is satisfied
by $G^{(k)}_{\La_n,q^{(k)}}$, $G^{(k)}_{\La_n,q^{(k)}}$ converges weakly to $\mu^{(k)}$ by subsequences  and
 $\{ \Theta(\cdot;r) = k\}$ is closed.  \qed

\vskip2cm

\subsection{{Relativized uniqueness of  DLR measures}}
        \label{sec:z5}
The title means  that   $k$-boundary conditions, $k\in\{1,..,S+1\}$,
select a unique measure in the thermodynamic limit.
 The precise results are stated in Theorem \ref{thmz5.0.1}
and its corollary Theorem \ref{thmz5.0.3}.

 \vskip.5cm

   \begin{thm}
   \label{thmz5.0.1}
 There are $\om$ and  $c$  positive
such that for all $\ga$ small enough, for all $k\in \{1,..,S+1\}$,
for all bounded, $\mathcal D^{(\ell_+)}$-mea\-su\-ra\-ble, simply
connected regions $\La_1$, $\La_2$, for all $k$-boundary conditions
$q_1,q_2$, for all
 $\mathcal D^{(\ell_+)}$-measurable sets $\Delta$ in
$\La_1\cap \La_2$ and for all
bounded, measurable cylindrical functions $f$ in
$\Delta$,
    \begin{equation}
       \label{z5.0.1}
\big| G^{(k)}_{\La_1,q_1}(f)- G^{(k)}_{\La_2,q_2}(f)\big| \le c
\|f\|_{\infty} |\Delta|\, e^{-\om \ell_+^{-1} {\rm dist}(\Delta,
(\La_1 \cap \La_2)^c)}
    \end{equation}

    \end{thm}

\vskip.5cm

The proof will be obtained after rewriting the expectations
$G^{(k)}_{\La,q^{(k)}}(f)$ in
a way which allows to
exploit the couplings introduced in \cite{DMPV2}.

\vskip.5cm

 \centerline  {\it Notation.}
 \nopagebreak
\noindent We fix  $\Delta$ and $f$ as in Theorem \ref{thmz5.0.1}.
Let $ \La \supset \Delta$ be a  bounded, $\mathcal
D^{(\ell_+)}$-mea\-su\-ra\-ble set, $\und \Ga\in \mathcal B^{k}_\La$
and (recall \eqref{z3.3})
     \begin{equation}
       \label{z5.0.2}
c(\und \Ga)= \bigcup_{\Ga\in  \und \Ga} c(\Ga),\;\; {\rm ext}(\und
\Ga)=  \La \setminus c(\und \Ga),\;\;W^{k, {\rm true}}(\und
\Ga;q)=\prod_{\Ga\in \und \Ga} W^{k, {\rm true}}( \Ga;q)
    \end{equation}
Denote by $\mathcal B^{k,{\rm ext}}_\La$   the subset of $\mathcal
B^{k}_\La$ of collections $\und \Ga=(\Ga_1,..,\Ga_n)$ made
exclusively of external contours, namely such that all $c(\Ga_i)$
are mutually disconnected. Let  $\und \Ga\in \mathcal B^{k,{\rm
ext}}_\La$ and call (dependence on  $f$, $\La$ and $\Delta$ is not
made explicit):
     \begin{eqnarray}
       \label{z5.0.3}
&&\hskip-2cm \hat K(\und \Ga)=  \Big \{\Ga\in  \und
\Ga: c(\Ga)\cap \Delta\ne \emptyset\Big\}
\\&&\hskip-2cm
       \label{z5.0.4}
\hat D(\und \Ga)= \{\Delta \cap {\rm ext}(\und \Ga)\}
\cup\{
 \bigcup_{\Ga\in \hat K(\und \Ga)}
 \delta_{\rm out}^{\ell_{+}}[c(\Ga)]\}\\&&
\hskip-2cm F(q,\und \Ga)= \frac{N^{(k)}(\und \Ga;q;f)}{N^{(k)}( \und
\Ga;q;1)}, \;\;\; q:q_{{\rm ext}(\und\Ga)}\in \mathcal X^{(k)}_{{\rm
ext}(\und \Ga)}
       \label{z5.0.5}
    \end{eqnarray}
where, calling
$\mathcal X^0( \Ga)=\{q_{c( \Ga)}: \eta(q_{c(\Ga)};r)= \eta_{\Ga}(r), r\in {\rm
sp}(\Ga),\Theta(q_{c( \Ga)};r)=h,
r\in\delta_{\rm in}^{\ell_{+,\ga}}[{\rm int}^{h}(\Ga)]\}$ and
 $\dis{\mathcal X^0(\und \Ga)=\bigcap_{\Ga\in \und \Ga}\mathcal X^0(\Ga)}$,
     \begin{equation}
       \label{z5.0.6}
N^{(k)}( \und \Ga;q;f)= \int_{ \mathcal X^0(\und \Ga)} e^{-\beta
H_{c(\und \Ga)}(q'_{c(\und \Ga)} |q_{{\rm ext}(\und\Ga)})}
f(q'_{c(\und \Ga)},q_{{\rm ext}(\und\Ga)})d\nu_{c(\und
\Ga)}(q'_{c(\und \Ga)})
    \end{equation}

\vskip1cm

   \begin{thm}
   \label{thm14n.2.1-1}
With the above notation
    \begin{equation}
       \label{z5.0.7}
\|F\|_\infty \le \|f\|_\infty,\;\;\;F(q,\und \Ga)=
F(q_{\hat D(\und \Ga)},\hat K(\und \Ga))
    \end{equation}
    \begin{equation}
       \label{z5.0.8}
 G^{(k)}_{\La,q^{(k)}}(f)= \sum_{\und \Ga \in \mathcal B^{k}_\La}
 \int_{\mathcal X^{(k)}_\La} F(q_\La;\phi_{\rm ext}(\und \Ga))
dp^{(k)}_{\ga,\La,q^{(k)}}(q_\La,\und \Ga)
    \end{equation}
where $\phi_{\rm ext}(\und \Ga)$ is the subset of  external contours
in $\und \Ga$ (obtained by deleting from  $\und \Ga$ all $\Ga'$ with
$c(\Ga')\subset c(\Ga)$ for some other $\Ga\in \und \Ga$); and,
recalling the definition \eqref{z3.1a},
    \begin{equation}
       \label{z5.0.9}
dp^{(k)}_{\ga,\La,q^{(k)}}(q_\La,\und
\Ga)=(Z^{(k)}_{\La,q^{(k)}})^{-1}e^{-\beta
H_{\La}(q_\La|q^{(k)}_{\La^c})} W^{k, {\rm true}}(\und
\Ga;q_\La)\text{\bf 1}_{q_\La\in \mathcal X^{(k)}_\La}
d\nu_{\La}(q_\La)
    \end{equation}

    \end{thm}

 \vskip.5cm

The proof is completely analogous to that
of Theorem 12.5.1.1 in \cite{leipzig} and omitted.

\vskip1cm

{\bf Proof of Theorem \ref{thmz5.0.1}.} By Theorem \ref{thmz3.2} the
contour weights satisfy the assumptions in Theorem 3.1 of
\cite{DMPV2} which can then be applied. It then follows that there
is a coupling $d\mathcal P(q',\und \Ga',q'',\und \Ga'')$ of
$dp^{(k)}_{\La_1,  q_1}$ and $dp^{(k)}_{\La_2,  q_2}$ with the
following property.
 $\mathcal P(\mathcal A)\ge 1- |\Delta|\, e^{-\om \ell_+^{-1} {\rm
dist}(\Delta, (\La_1 \cap \La_2)^c)}$ where $\mathcal A$ is the set of all
$(q',\und \Ga',q'',\und \Ga'')$ for which there
exists a $\mathcal D^{(\ell_+)}$-measurable
region $\Delta'$ such that: if $\Ga \in \und \Ga'\cup \und \Ga''$ then $c(\Ga)\cap
\delta_{\rm in}^{\ell_+}[\Delta']=\emptyset$; the contours of $ \und \Ga'$ and
$ \und \Ga''$ with spatial support in $\Delta'$ are identical as well as
the restrictions to $\Delta'$ of
$q'$ and $q''$; finally
$\Delta \subset \Delta'\setminus \delta_{\rm in}^{\ell_+}[\Delta']$.
By \eqref{z5.0.8},
    \begin{equation}
       \label{z5.0.10}
  G^{(k)}_{\La_1,q_1}(f)- G^{(k)}_{\La_2,q_2}(f)=\int \Big(
  F(q_{\La_1};\phi_{\rm ext}(\und \Ga'))- F(q_{\La_2};\phi_{\rm ext}(\und \Ga''))
  \Big)d\mathcal P(q',\und \Ga',q'',\und \Ga'')
    \end{equation}
and by the definition of $F$, $F(q_{\La_1};\phi_{\rm ext}(\und \Ga'))= F(q_{\La_2};\phi_{\rm ext}(\und \Ga''))$ on $\mathcal A$, hence Theorem \ref{thmz5.0.1}.  \qed

\vskip1cm

As an immediate corollary of  Theorem \ref{thmz5.0.1} we have:

\vskip.5cm

   \begin{thm}
   \label{thmz5.0.3}
In the same context of  Theorem \ref{thmz5.0.1},
    \begin{equation}
       \label{z5.0.11}
\big| G^{(k)}_{\La_1,q_1}(f)- \mu^{(k)}(f)\big| \le c \|f\|_{\infty}
|\Delta|\, e^{-\om \ell_+^{-1} {\rm dist}(\Delta, \La_1 ^c)}
    \end{equation}
where $\mu^{(k)}$ is the DLR measure defined in Theorem \ref{thmz4.0.2}.
    \end{thm}

\vskip2cm

\subsection{{Tail field and extremality}}
        \label{sec:z6}

In this section we will prove that the DLR measures $\mu^{(k)}$ have
all trivial $\si$-algebra at infinity (also called the tail field)
and they are therefore extremal DLR measures.  The property follows
from the Peierls bounds, Theorem \ref{thmz3.2}, and the exponential
decay of correlations, Theorem \ref{thmz5.0.3}. The particular
structure of the model is at this point rather unimportant and
indeed we will be able to avoid many proofs by referring to their
analogues  \cite{leipzig}.

\vskip.5cm

   \begin{defin}
   \label{definz6.0.1}
Let $\{\Delta_k\}$ be
an arbitrary but fixed  increasing sequence of $\mathcal
D^{(\ell_{+})}$-measurable cubes of sides
$2^k\ell_{+}$ which invades the whole space and $\mathcal S$ the collection of sequences $\{\La_k\}$ of the
form $\{\tau_i\Delta_k\}$,   where $\tau_i,i\in a
\mathbb Z^d, a\in \{2^{-n}, n\in \mathbb N\}$, is the translation by $i$.

   \end{defin}

  \vskip.5cm

When proving in the  next subsections that the measures $\mu^{(k)}$
are translational invariant, we will need translates of the sequence
$\{\Delta_k\}$, hence the definition of $\mathcal S$. Observe that
sequences in $\mathcal S$ are not necessarily $\mathcal D^{(\ell_
+)}$-measurable.

\vskip.5cm

   \begin{defin}
   \label{definz6.0.2}
The $k$-tail field, $k\in \{1,..,S+1\}$, is defined as
     \begin{eqnarray}
           \label{z6.0.1}
&&\mathcal Q_{k,{\rm tail}}=\Big\{ q\in
\mathcal Q: \lim_{n\to \infty} G_{\La_n,q}( f ) =
\mu^{(k)}(f),\;\text{ for any $\{\La_k\}\in \mathcal S$}\nn\\&&
\hskip2cm \text{and
for any bounded, measurable cylindrical function $f$}\;\Big\}
    \end{eqnarray}

   \end{defin}

  \vskip.5cm

   \begin{thm}
   \label{thmz6.0.1}

For all $\ga$ small enough  $\mu^{(k)}\big(\mathcal Q_{k,{\rm tail}}\big)=1$,
 $k\in \{1,..,S+1\}$.

  \end{thm}

  \vskip.5cm

By taking countably many intersection we can reduce the proof
of Theorem \ref{thmz6.0.1} to
the proof that for any  sequence  $\{\La_n\}\in \mathcal S$ and
any $\mathcal D^{(\ell_{+})}$-measurable cube $\Delta$
     \begin{eqnarray}
           \label{z6.0.2}
&&
\mu^{(k)}\Big(\big\{q\in \mathcal X: \lim_{n\to \infty}
G_{\La_n,q} ( f ) = \mu^{(k)}(f),\;\; \text{for all bounded, measurable functions $f$}\nn\\&&
\hskip2cm \text{
cylindrical in $\Delta$}\big\}\Big)=1
       \end{eqnarray}
This would be direct consequence of Theorem
\ref{thmz5.0.3} if we had
$G^{(k)}_{\La_n,q}$ instead of $G_{\La_n,q}$
in \eqref{z6.0.2} and the whole point
will be to reduce to such a case.

\vskip.5cm

$\bullet$\; {\em Random sets.}

We call $\La_{n;0}$ the union of all $\mathcal
D^{(\ell_{+})}$  cubes
contained in $\La_n$ (recall $\La_n$ may not
be $\mathcal D^{(\ell_{+})}$-measurable)
and define the
random set $\hat N_{n,{\rm out}}$ as follows. $\hat N_{n,{\rm out}}(q)$ is
 the
union of $\La_{n;0}^c$ with all the maximal connected components $A$
of the set $\{r\in \La_{n;0}: \Theta(q;r)\ne k\}$ such that
$\delta_{\rm out}^{\ell_+}[A]\cap \La_{n;0}^c\ne \emptyset$. We call
$\hat N_{n,{\rm in}}$ the complement of $\hat N_{n,{\rm out}}$ and
observe that by construction, $ \Theta(r;q)=k$ for all $r\in
\delta_{\rm in}^{\ell_+}[ \hat N_{n,{\rm in}}]$.

\vskip.5cm

$\bullet$\; {\em The favorable case.}

Given $r$ call $A(q;r)$ the maximal connected component of
$\{\Theta(\cdot;q)\ne k\}$ which contains $r$
($A(q;r)$ may be empty).  Calling ${\rm diam }(A)$ the
diameter of the set $A$, we define
     \begin{equation}
           \label{z6.0.3}
\mathcal B_n:=\Big\{q:{\rm diam}( A(r;q))
<2^{k-1}\ell_{+},\;\,\text{for all $r \in \delta_{\rm
in}^{\ell_+}[\La_{n;0}]$} \Big\}
       \end{equation}
Notice that
     \begin{equation}
           \label{z6.0.4}
\hat N_{n,{\rm out}}(q) \cap \La_{n-2}
=\emptyset,\;\;\text{for all $q\in\mathcal B_n$}
       \end{equation}

\vskip1cm

   \begin{thm}
   \label{thmz6.0.2}
Let $\eps_n:= 2^{n(d-1)}  e^{ -\beta (c^*/4) \, (\zeta^2
\ell_-^d)2^{n-1}}$ and
    \begin{equation}
           \label{z6.0.5}
F_n= \Big\{q: G_{\La_n,q} \big(\mathcal B_{n}\big) \ge
1 -\sqrt{\eps_n}\Big\}
       \end{equation}
Then
    \begin{equation}
           \label{z6.0.6}
\mu^{(k)}\big(\,F_n\,\big)\ge 1 -\sqrt{\eps_n}
       \end{equation}
and for all $q\in
F_n$ and all measurable, bounded functions $f$ cylindrical
in $\Delta$,
    \begin{equation}
     \label{z6.0.7}
 \big|G_{\La_n,q} (f) - \mu^{(k)}(f)\big| \leq 2
\|f\|_\infty \sqrt{\eps_n}  +  c \|f\|_{\infty} |\Delta|\,
e^{-\om \ell_+^{-1}  {\rm dist}(\Delta, \La_{n-2}^c)}
       \end{equation}
with $c$ and $\om$ as in  \eqref{z5.0.1}.

   \end{thm}

 \vskip.5cm

The proof of Theorem \ref{thmz6.0.2} is completely analogous
to the
proof of Theorem 12.2.2.5 in \cite{leipzig} and it is omitted,
we just outline its main steps. To prove
\eqref{z6.0.6} we write
    \begin{equation}
         \label{z6.0.8}
   \mu^{(k)} \big(\mathcal B_{n}\big)
= \mu^{(k)} \Big( G_{\La_n,q} \big(\mathcal
B_{n}\big)\Big)  \leq \mu^{(k)}( F_n)+( 1
-\sqrt{\eps_n}) \big(1-\mu^{(k)} (F_n )\big)
       \end{equation}
and \eqref{z6.0.6} follows from \eqref{z6.0.8} and the inequality
$\mu^{(k)} \big(\mathcal B_{n}\big) \ge 1 - \eps_n$.
To prove the latter we observe that $\mathcal B_{n}$ is a cylindrical set
and therefore $\dis{\mu^{(k)}\big(\,\mathcal B_{n}\,\big)= \lim_{m\to \infty}
G^{(k)}_{\La_m,q^{(k)}}(\mathcal B_{n})}$.
We can then use  the Peierls bounds in
\eqref{z3.4} and after some standard
combinatorial arguments prove the desired inequality and hence  \eqref{z6.0.6}.

Let $q\in F_n$  then
     \begin{equation*}
\Big|G_{\La_n,q}(f) -\sum_{B\subset \La_{n-2}^c
}G_{ \La_n,q}\big( \text{\bf 1}_{ \hat N_{n,{\rm out}}=B
}
 G^{(k)}_{B^c,q}(f)\big)\Big|
 \le  \|f\|_\infty G_{\La_n,q}\big({ \mathcal B_{n}^c}\big)
\le
 \|f\|_\infty \sqrt{\eps_n}
       \end{equation*}
Hence
    \begin{equation*}
\big|G_{\La_n,q} (f) - \mu^{(k)} (f)\big| \leq 2 \|f\|_\infty
\sqrt{\eps_n} + \sup_{B \subset \La_{n-2}^c }\sup_{q^{(k)} \in
\mathcal Q^{(k)}}\big| G^{(k)}_{ B^c,q^{(k)}}(f)-\mu^{(k)} (f)\big|
       \end{equation*}
such that   \eqref{z6.0.7} follows from \eqref{z5.0.11}, the bound
being uniform in all $f$ cylindrical in $\Delta$. \qed

\vskip1cm

{\em Proof of Theorem \ref{thmz6.0.1}.}

We use   a Borel-Cantelli argument. Let
$F_n$ as above, then
    \begin{equation}
           \label{z6.0.9}
\mu^{(k)} (F)=1,\qquad F:= \bigcup_m \bigcap_{n\ge m}F_n
       \end{equation}
By \eqref{z6.0.7}
    \begin{equation}
           \label{z6.0.10}
 \lim_{n\to \infty} G_{\La_n,q} ( f
) = \mu^{(k)}(f), \text{ for all $q\in F$}
       \end{equation}
and Theorem \ref{thmz6.0.1} follows from \eqref{z6.0.9} and
\eqref{z6.0.10}.

\vskip2cm

\subsection{{Decomposition of translational invariant DLR measures}}
        \label{sec:z7}
In this section we will prove that any translational invariant DLR
measure can be written as a convex combination of the measures
$\mu^{(k)}$, this is not yet the   decomposition into ergodic DLR
measures because we do not know  that the $\mu^{(k)}$ are
translational invariant (a statement proved in the next section).
However it follows directly from Theorem \ref{thmz5.0.3} that any
$\mu^{(k)}$ is translational invariant under $\{\tau_i, i\in
\ell_+\mathbb Z^d\}$. Indeed  by  Theorem
 \ref{thmz5.0.3} for any $q\in \mathcal Q^{(k)}$,
 $\mu^{(k)}=\lim_n G^+_{\La_n,q}$ weakly,  then
$\tau_i\big(\mu^{(k)})=\lim_n \tau_i\big(G^+_{\La_n,q}\big)$ and the
latter is equal to $\lim_n \big(G^+_{\tau_i(\La_n),\tau_i(q)}\big)$
which by Theorem \ref{thmz5.0.3} is equal to $\mu^{(k)}$.  We have:

\vskip.5cm

   \begin{thm}
   \label{thmz7.0.1}

For all $\ga$ small enough the following holds. Let  $m$ be any
 DLR measure  invariant
under $\{\tau_i,
i\in \ell_{+}\mathbb Z^d\}$,
 then there is a
unique sequence $(u_0,..,u_S) $ of numbers in $[0,1]$ such that
    \begin{equation}
    \label{z7.0.1}
m  = u_0 \mu^{(0)} + \cdots + u_S \mu^{(S)}
       \end{equation}

       \end{thm}

\vskip.5cm

{\bf Proof.}  The proof is an adaptation of the classical argument
by Gallavotti and Miracle-Sole for the analogous property in the
Ising model at low temperatures.  Its extension to  Ising models
with Kac potentials has been carried out in \cite{BMP} and adapted
in \cite{leipzig} to the LMP model.  All these proofs are basically
the same as the original one and we think it useless to repeat it
once more here.  The argument shows (see for instance Section 12.3
in \cite{leipzig}) that there are  numbers $u_{k,L}\in [0,1]$,
$k=1,..,S+1$ with the following property; for any bounded
cylindrical function $f$ there is a function $\eps(L)$ vanishing as
$L\to \infty$ and satisfying, for any $\mathcal
D^{(\ell_{+})}$-measurable cube $\La$ of side $L\ell_{+}$:
   \begin{eqnarray}
 &&
\hskip-1cm \big|m (f)  -\sum_{k=0}^{S} u_{k,L} \mu^{(k)} ( f) \big| \le \eps(L)
     \label{z7.0.2}
      \end{eqnarray}
By compactness there is a sequence $L_n$ (independent of $f$) such
that $\dis{\lim_{n\to \infty}u_{k,L_n}=:u_k}$, for all $k$.  Then
   \begin{eqnarray}
 m (f)  =\sum_{k=0}^{S} u_{k} \mu^{(k)} ( f)
     \label{z7.0.3}
      \end{eqnarray}
and \eqref{z7.0.1} is proved since $m$ is determined by expectations
of bounded cylindrical functions $f$.  \qed

\vskip2cm

\subsection{{Ergodicity}}
        \label{sec:z8}

In this section we will complete the proof of Theorem \ref{thme1.1} by proving that the measures
$\mu^{(k)}$ are translational invariant, since their tail field is trivial they are then ergodic and
the decomposition \eqref{z7.0.1} becomes the decomposition into ergodic DLR measures.

 Let $\mathcal S$ be as  in  Definition \ref{definz6.0.1}
and   $i\in a
\mathbb Z^d$,    $a\in \{2^{-n}, n
\in \mathbb N\}$,   define
   \begin{eqnarray}
 && \hskip-1cm
\mathcal Q_{k,{\rm tail};\;i}=\big\{ q\in \mathcal Q:
\lim_{n\to \infty} G_{\La_n,q}( f ) =
[\tau_i(\mu^{(k)})](f),\; \text{for any $\{\La_n\}\in
\mathcal S$}\nn\\&&\text{ and any bounded, measurable
cylindrical function $f$} \big\}
     \label{z8.0.1}
      \end{eqnarray}
such that $\mathcal Q_{k,{\rm tail};\;0}=\mathcal Q_{k,{\rm tail}}$
the tail set of  Definition \ref{definz6.0.2}.  It then follows (see
the proof of the analogous Lemma 12.4.1.1 in \cite{leipzig} for
details of the proof) that:

\vskip.5cm

   \begin{lemma}
   \label{lemmaz8.0.1}
For any $i\in a\mathbb Z^d$,  $a\in \{2^{-n}, n
\in \mathbb N\}$,
   \begin{equation}
   \label{z8.0.2}
\mathcal Q_{k,{\rm tail};\;i}=\tau_{-i}(\mathcal Q_{k,{\rm tail}}),\quad [\tau_i(\mu^{(k)})]\big (\mathcal Q_{k,{\rm tail};\;i}\big)=1
      \end{equation}
Moreover
$\tau_i(\mu^{(k)})\ne \mu^{(h)}_\ga$ if and only if
$\mathcal Q_{k,{\rm tail};\;i}\cap \mathcal Q_{h,{\rm tail}}= \emptyset$.

  \end{lemma}

\vskip1cm

   \begin{lemma}
   \label{lemmaz8.0.2}

For all $\ga$ small enough the following holds: for any  $i\in a\mathbb Z^d$,  $a\in \{2^{-n}, n
\in \mathbb N\}$,
$\mu^{(h)}$ and $\tau_i(\mu^{(k)})$, $h\ne k$, are mutually singular
and $\mathcal Q_{k,{\rm tail};\;i} \cap\mathcal Q_{h,{\rm tail}} =\emptyset$.
  \end{lemma}

\vskip.5cm

\noindent
{\bf Proof.} By Lemma \ref{lemmaz8.0.1} it suffices to show
that $\mu^{(k)}\ne \tau_i(\mu^{(h)})$ which is proved by the same
argument used to
prove Lemma 12.4.1.3 of \cite{leipzig}, we just report the main steps.
Suppose (without loss of generality) that
$k\in \{1,..,S\}$, i.e.\ $\mu^{(k)}$ an ordered state.

Since $\tau_i(\mu^{(k)})$ is invariant by  translations of
$\ell_{+}\mathbb Z^d$, for any $i\in a\mathbb Z^d$,  $a\in \{2^{-n}, n
\in \mathbb N\}$, we may also restrict to  $\tau_i$ with $i\in
C^{(\ell_{+})}_0\cap a\mathbb Z^d$. Let $\La$ be  a $\mathcal
D^{(\ell_{+})}$-measurable cube,
$|q_\La(k)|$ the number of particles in   $q_\La$ with spin $s=k$, then
 $\dis{
\mu^{(h)}(\frac {|q_\La(k)|}{|\La|})= \mu^{(h)}(\frac
{|q_{C_0^{(\ell_{+})}}(k)|}{|C_0^{(\ell_{+})}|})}$ because
$\mu^{(h)}$ is   invariant
under translations in $\ell_{+}\mathbb Z^d$.
 $\dis{\mu^{(h)}(\frac {|q_\La(k)|}{|\La|})}$ is then bounded from
 above by
    \begin{eqnarray*}
 &&  (\rho^{(h)}_{k}+\zeta)
 +  \mu^{(h)}(\mathbf 1_
 {\Theta(\cdot;0)\ne h}
 \frac
 {|q_{C_0^{(\ell_{+})}(k)}|}{\ell_{+}^d}) \le
  \rho^{(h)}_{k}+\zeta+c [e^{-\beta (c^*/4) \zeta^{2}
 \ell_{-}^d}]^{1/2}
      \end{eqnarray*}
    \begin{eqnarray*}
 && c =  \ell_{+}^{-d}
 \mu^{(h)}(|q_{C_0^{(\ell_{+})}(k)}|^2)^{1/2} \le
 \ell_{+}^{-d/2} \Big(\sum_{C^{(1)}_i\subset
 C_0^{(\ell_{+})}}
 \mu^{(h)}(|q_{C^{(1)}_i}(k)|^2)\Big)^{1/2}
      \end{eqnarray*}
having used Cauchy-Schwartz. Since the energy is non negative,
   \begin{eqnarray*}
 &&
 \mu^{(h)}(|q_{C^{(1)}_i(k)}|^2) \le \sum_{n\ge 1}
 \frac{ n^2}{n!} <\infty
      \end{eqnarray*}
Thus in conclusion
    \begin{eqnarray}
 && \mu^{(h)}(\frac {|q_\La(k)|}{|\La|})\le \rho^{(h)}_{k}+(\zeta
 +   c' e^{-\beta (c^*/8) \zeta^{2}
 \ell_{-}^d})
     \label{z8.0.3}
      \end{eqnarray}
For any $j\in C^{(\ell_{+})}_0\cap a\mathbb Z^d$,
$\tau_j(\La) \supset \La_0$, $\La_0=\La\setminus
\delta_{\rm in}^{\ell_{+,\ga}}[\La]$. Let $N_\La$
and $N_{\La_0}$ the number of $\mathcal
D^{(\ell_{+})}$-cubes in $\La$ and $\La_0$,  then for any $i\in C^{(\ell_{+})}_0\cap a\mathbb Z^d$
   \begin{eqnarray*}
 &&
[\tau_i(\mu^{(k)})](\frac {|q_\La(k)|}{|\La|})\ge  \mu^{(k)}(\frac
 {|q_{\La_0}(k)|}{|\La|})=\frac{N_{\La_0}}{N_\La}\mu^{(k)}(\frac
 {|q_{C_0^{(\ell_{+})}}(k)|}{\ell_{+}^d})
      \end{eqnarray*}
We then write  $\dis{\mu^{(k)}(\frac
 {|q_{C_0^{(\ell_{+})}}|}{\ell_{+}^d})\ge \mu^{(k)}(\frac
 {|q_{C_0^{(\ell_{+})}}|}{\ell_{+}^d} \text{\bf
 1}_{\Theta(\cdot;0)=k})}$ and get
   \begin{eqnarray*}
 &&  \hskip-3cm[\tau_i(\mu^{(k)})](\frac {|q_\La(k)|}{|\La|})\ge
 \frac{N_{\La_0}}{N_\La}(\rho^{(k)}_{k}-\zeta)(1-e^{-\beta (c^*/4) \zeta^{2}
 \ell_{-}^d}) \\&& \hskip.2cm\ge \rho^{(k)}_{k}-\Big(\zeta +
 \rho^{(k)}_{k}\{e^{-\beta (c^*/4) \zeta^{2}
 \ell_{-}^d} + \frac{N_\La-N_{\La_0}}{N_\La}\}\Big)
      \end{eqnarray*}
which is strictly larger than the r.h.s.\ of
\eqref{z8.0.3} for $\La$ large  and $\ga$ small.  Thus
$\tau_i(\mu^{(k)}) \ne \mu^{(h)}$.
\qed

\vskip1cm

We will  prove translational invariance for special
values of the  mesh, $a\in \{2^{-n}, n\in \mathbb N\}$, the general case
follows by a
density argument completely analogous to the one used for the LMP
model, see Subsection 12.4.2 of \cite{leipzig}, which is therefore omitted.

\vskip.5cm

   \begin{thm}
   \label{thmz8.0.1}
For all $k\in \{0,..,S\}$ and all $\ga$ small enough the measures $\mu^{(k)}$ are
invariant under   translations by
$\tau_i$, for any $i\in a\mathbb Z^d, a\in \{2^{-n}, n \in \mathbb N\}$.

  \end{thm}

\vskip.5cm

\noindent
{\bf Proof.} Fix $a\in \{2^{-n}, n\in \mathbb N\}$.
Since $\tau_i(\mu^{(k)})$ is invariant under translations in $\ell_+\mathbb Z^d$,  the measure
   \begin{eqnarray}
 && \nu :=  \frac{a^d}{|C^{(\ell_{+})}_0|}
 \sum_{i\in C^{(\ell_{+})}_0\cap a \mathbb Z^d} \tau_i(\mu^{(k)})
      \label{z8.0.4}
      \end{eqnarray}
is invariant under the group of translations $\{\tau_i, \; i\in a \mathbb Z^d\}$. Then by
\eqref{z7.0.1}
    \begin{equation}
    \label{z8.0.5}
\nu =  u_0 \mu^{(0)} + \cdots + u_S \mu^{(S)}
       \end{equation}
By  Lemma \ref{lemmaz8.0.2} $\tau_i(\mu^{(k)})\big(\mathcal
Q_{h,{\rm tail}}\big)=0$ for any $h\ne k$, such that
 $\nu\big(\mathcal Q_{h,{\rm tail}}\big)=0$.  On the other hand $\mu^{(h)}\big(\mathcal Q_{h,{\rm tail}}\big)=1$,
 therefore in \eqref{z8.0.5} $u_h=0$ for all $h\ne k$ and hence $\nu=\mu^{(k)}$.
If there is $i\in C^{(\ell_{+})}_0\cap a \mathbb Z^d$ such that
$\tau_i(\mu^{(k)} )\ne \mu^{(k)} $, then again by Lemma
\ref{lemmaz8.0.2}, $\mu^{(k)}(\mathcal Q_{k,{\rm tail};i})=0$ and
$[\tau_i(\mu^{(k)} )](\mathcal Q_{k,{\rm tail};i})=1$ which
contradicts \eqref{z8.0.4} (as we have proved that $\nu=\mu^{(k)}$).
\qed

\vskip2cm

\part{ }
\label{sec:z9}

In this part we prove Theorem \ref{thmz3.2}, thus we show that there
is a constant $c^*$ such that the Peierls bounds are satisfied with
constant $c=c^*/2$ where we say that the Peierls bound holds with
constant $c$ if
        \begin{equation}
        \label{pb}
W^{k,\text{true}}(\Ga|q^{(k)})\le \exp\{-\beta c\zeta^2\ell_-^d
N_\Ga\},\qquad  N_\Ga=\frac{|{\rm sp}(\Ga)|}{\ell_+^d}
    \end{equation}
for all $k$, for all bounded contours of color $k$ and for all
$k$-boundary conditions $q^{(k)}$.

\vskip1cm

\section{{\bf Cut-off weights of contours}}
    \label{sec:3.3}

 As explained in Subsections 11.4 and 11.5 of
\cite{leipzig}, following the approach of
Zahradnik,\cite{cluster-expansion} we introduce the cutoff contours
weights.

Given any $k$ and any $k$-colored contour $\Ga$ we are going to
 define the weight $W^{(k)}(\Ga|q)$ for any
configuration $q$ which is a $k$-boundary condition for $c(\Ga)$ in
\eqref{I.4.3} below, the definition will imply that $W^{(k)}(\Ga|q)$
depends only on the restriction of $q$ to $\{r\in c(\Ga)^c :
\text{dist}(r; c(\Ga))\le 2\ga^{-1}\}$. For $\und
\Ga=(\Ga(1),..,\Ga(n))\in \mathcal B^{k}_\La$ we  call
          \begin{equation}
         \label{I.4.3a}
W_\la^{(k)}(\und\Ga|q_\La)= \prod_{i=1}^n
W_\la^{(k)}(\Ga(i)|q_{\La}),\quad
\mathbf{X}^{(k)}_{\La,\la}(q_\La)=\sum_{\und \Ga\in \mathcal
B^{k}_\La}W_\la^{(k)}(\und\Ga|q_\La)
    \end{equation}
and for any bounded, simply connected $\mathcal
D^{(\ell_+)}$-measurable region $\La$ and any $k$-boundary condition
$\bar q_{\La^c}\in \mathcal X_{\La^c}^{(k)}$ we introduce the
$k$-cutoff partition function in a region $\La$ with b.c. $\bar
q_{\La^c}$
 as
    \begin{equation}
      \label{I.4.2}
Z^{(k)}_{\La,\la} (\bar q_{\La^c})= 
\int_{\mathcal X_\La^{(k)}} e^{-\beta H_{\La,\la}(q_{\La}|\bar
q_{\La^c})} \mathbf{X}^{(k)}_{\La,\la}(q_\La) \nu(dq_\La)
    \end{equation}
With same notation as in \eqref{z3.3} we then define
   \begin{equation}
\mathcal N_\la^{(k)}(\Ga,q_{{c(\Ga)}^c})=  \int_{\Upsilon_{{\rm
sp}(\Ga)}(\eta_\Ga)} e^{-\beta H_{{\rm sp}(\Ga),\la}(q_{{\rm
sp}(\Ga)}|q_{c(\Ga)^c})}
 \prod_{j=1}^pZ^{(k_j)}_{{\rm
int}_{j}(\Ga),\la}(q_{D_j})\nu(dq_D)  \nu(dq_{{\rm sp}(\Ga)})
    \label{I.4.1}
     \end{equation}
   \begin{equation}
\mathcal D_\la^{(k)}(\Ga,q_{{c(\Ga)}^c})=  \int_{\mathcal
X^{(k)}_{{\rm sp}(\Ga)}} e^{-\beta H_{{\rm sp}(\Ga),\la}(q_{{\rm
sp}(\Ga)}|q_{c(\Ga)^c})}
 \prod_{j=1}^pZ^{(k)}_{{\rm
int}_{j}(\Ga),\la}(q_{D_j})\nu(dq_D)  \nu(dq_{{\rm sp}(\Ga)})
    \label{I.4.222}
     \end{equation}

All the above quantities depend on the weights
$\{W_\la^{(k)}(\Ga|q_{c(\Ga)^c})\}$ which we define (implicitly) by
introducing first  a constant $c_w>0$ and then  setting
   \begin{equation}
      \label{I.4.3}
W_\la^{(k)}(\Ga|q_{c(\Ga)^c})= \min\Big\{\frac{\mathcal
N_\la^{(k)}(\Ga,q_{c(\Ga)^c})}{\mathcal
D_\la^{(k)}(\Ga,q_{c(\Ga)^c})}, e^{-\beta c_w \zeta^2 \ell_{-}^d
N_\Ga}\Big\}
     \end{equation}
\noindent \eqref{I.4.3} is not a closed formula because the r.h.s.\
still depends on the weights, however the contours on the r.h.s.\
are ``smaller'' and, by means of  an inductive procedure, it is
possible to prove there is a unique choice of
$W_\la^{(k)}(\Ga|q_{c(\Ga)^c})$ such that \eqref{I.4.3} holds for
all $k$, all $\Ga$ and all $q_{c(\Ga)^c}$,
  see Theorem 10.5.1.2 in \cite{leipzig}.

\bigskip
The important point of these definitions is that if the estimate
\eqref{3.9} below holds, then the cut-off weights are equal to the
true ones. This is the content of the next Theorem whose proof is
omitted being completely analogous to Theorem 10.5.2.1 in
\cite{leipzig}.

\bigskip

    \begin{prop}
    \label{thm3.2}
Suppose that for any $k$, any contour $\Ga$ of color $k$ and any
$k$-boundary conditions $q_{c(\Ga)^c}$ for $c(\Ga)$,
    \begin{equation}
    \label{3.9}
W_\la^{(k)}(\Ga|q_{c(\Ga)^c})<e^{-\beta c_w \zeta^2 \ell_{-}^d
N_\Ga}
    \end{equation}
then
        \begin{equation}
    \label{3.10}
W_\la^{(k)}(\Ga|q_{c(\Ga)^c})=W^{k,\text{true}}(\Ga|q^{(k)})
    \end{equation}
    \end{prop}

\bigskip

We will prove that if $c_w>0$ is small enough then \eqref{3.9} holds
for all $\ga$ correspondingly small. The main ingredient in the
proof of  \eqref{3.9} is the exponential decay in restricted
ensembles proved in \cite{DMPV2}.

\vskip2cm

\section{{\bf Proof of the Peierls bound}}
        \label{sec:3}

The proof of \eqref{3.9} is based on an extension of the classical
Pirogov-Sinai strategy, we refer to Chapter 10 of \cite{leipzig} for
general comments and proceed with the main steps of the proof.  Most
of it follows from Chapter 11 of  \cite{leipzig} and Theorem 3.1 of
\cite{DMPV2}.  Precise quotations will be given in complementary
sections where we will also add proofs to fill in parts not covered
by the above references.

\vskip1cm

\subsection{Energy estimate}
    \label{sec:3.4}

The first step is the following Theorem.

\vskip.5cm

    \begin{thm} [Energy estimate]
        \label{thm:2.9}
There is  $c_1>0$ such that the following holds. Given any $c'>0$
there is $c$ such that for all $\la$: $|\la-\la_\beta|\le
c'\ga^{1/2}$, for all $k$, for all $k$-contour $\Ga$, for all $k$
boundary conditions $q_{c(\Ga)^c}$ and for any $c_w>0$, the
 following estimate holds for all $\ga$ small enough:
        \begin{equation}
      \label{I.5.1}
\frac{\mathcal N_\la^{(k)}(\Ga,q_{c(\Ga)^c})}{\mathcal
D_\la^{(k)}(\Ga,q_{c(\Ga)^c})} \le e^{-\beta [c_1 \zeta^2
-c\ga^{1/2}-(\alpha_+-\alpha_-)d]\ell_-^d N_\Ga} \prod_{j=1}^p
\frac{e^{\beta \mathbf{I}_{k_j}({\rm int}_{j}(\Ga))} Z^{(k_j)}_{{\rm
int}_{j}(\Ga)),\la}(\chi^{(k_j)}_{{\rm sp}(\Ga)})} {e^{\beta
\mathbf{I}_{k}({\rm int}_{j}(\Ga)))} Z^{(k)}_{{\rm int}_{j}(\Ga)),
\la}(\chi^{(k)}_{{\rm sp}(\Ga)})}
    \end{equation}
    where for any bounded $\mathcal D^{(\ell_+)}$-measurable set $\Om$,
         \begin{equation}
         \label{I.1.3}
\mathbf{I}_k(\Omega) = \int_{\Om^c} \big[e_\la^{\rm mf}(\rho^{(k)})-
e_\la^{\rm mf}(
J_\ga\star\chi^{(k)}_{\Omega^c})\big]-\int_{\Om}e_\la^{\rm mf}(
J_\ga\star\chi^{(k)}_{\Omega^c})
     \end{equation}
    \begin{equation}
      \label{2.38}
\chi^{(k)}_{\Omega^c}(x,s)=\rho^{(k)}_s\text{\bf
1}_{x\in\Omega\cap\ga^{-1/2}\mathbb{Z}^d}
     \end{equation}
and  $\rho^{(k)}$ a minimizer  of $F^{\rm mf}_{\beta,\la_\beta}$,
see \eqref{z1.1}).
    \end{thm}

\bigskip

In classical Pirogov-Sinai models with nearest neighbor interactions
 the analogue of  Theorem \ref{thm:2.9} follows directly
from the extra energy due to presence of the contour, here contours
have a non trivial spatial structure which leads, after a coarse
graining a la Lebowitz-Penrose, to a delicate variational problem.

Theorem \ref{thm:2.9} will be proved in Section \ref{sec:4}.

 \vskip1cm

\subsection{Surface corrections to the pressure}
    \label{sec:3.6}

We  exploit the arbitrariness of $c_w$ in Theorem
\ref{thm:2.9} and fix
          \begin{equation}
      \label{I.5.1bis}
c_w = \frac{c_1}{100}
    \end{equation}
such that the first factor on the r.h.s.\ of \eqref{I.5.1} is
consistent with \eqref{3.9} but we need a good control of the ratio
of partition functions in  \eqref{I.5.1}.  As typical in the
Pirogov-Sinai theory a necessary requirement comes from demanding
that the pressures in the restricted ensembles are equal to each
other, a request which will fix the choice of the chemical
potential.

\vskip1cm

 \begin{thm} [Equality of pressures]
    \label{plabeta}
For any chemical potential $\la\in [\la_\beta-1,\la_\beta+1]$ and
any van Hove sequence $\La_n\to \mathbb R^d$ of $\mathcal
D^{(\ell_+)}$-measurable regions the following limits exist
      \begin{eqnarray}
        \nn
&&\lim_{n\to \infty} \frac{1}{\beta|\La_n|}
\;\log{Z^{(k)}_{\La_n,\la} (\bar q_{\La^c})} =:P_\la^{(\rm
ord)},\quad  \text{for all $k\in \{1,\dots ,S\}$}
\\&& \lim_{n\to \infty} \frac{1}{\beta|\La_n|}
\;\log{Z^{(S+1)}_{\La_n,\la} (\bar q_{\La^c})} =:P_\la^{(\rm
disord)}
     \label{p9.6e.0.3}
    \end{eqnarray}
and they are independent of the sequence $\La_n$ and of the $k$
boundary condition $\bar q_{\La^c}$. Moreover there is $c_0>0$ and,
for all $\ga>0$ small enough, there is $\la_{\beta,\ga}$,
$|\la_{\beta,\ga}-\la_\beta|\le c_0\ga^{1/2}$  such that
        \begin{equation}
    \label{p2.30}
P_{\la_{\beta,\ga}}^{(\rm ord)}=P_{\la_{\beta,\ga}}^{(\rm disord)}=:
P_{\la_{\beta,\ga}}
     \end{equation}

     \end{thm}

\vskip1cm

Theorem \ref{plabeta} is proved in Appendix  \ref{sec:a2.10}. The
existence of the thermodynamic limits, \eqref{p9.6e.0.3}, is not
completely standard because there is the additional term given by
the weights of the contours. However the Peierls bounds
(automatically satisfied by the cut-off weights) imply that contours
are rare and small and can then be controlled.  The equality
\eqref{p2.30}  is more subtle, it is proved by showing that (i)
$P_\la^{(\rm ord)}$ and $P_\la^{(\rm disord)}$ depend continuously
on $\la$; (ii) by a Lebowitz-Penrose argument they are close as
$\ga\to 0$ to the mean field values; (iii) the difference of the
mean field pressures changes sign as $\la$ crosses $\la_\beta$.

By Theorem \ref{plabeta} at $\la=\la_{\beta,\ga}$   the volume
dependence in the ratio \eqref{I.5.1} disappears and to conclude the
estimate we need to prove that the next surface term ``is small''.
This is the hardest part of the whole analysis where Theorem 3.1 of
\cite{DMPV2} enters crucially.

\vskip.5cm
            \begin{thm} [Surface corrections to the
            pressure]
        \label{thm:2.14}
With $c_w>0$ as in \eqref{I.5.1bis} there are $\ga_0$ and $c$ such
$c$ such that for any $\ga\le \ga_0$, all bounded $\mathcal
D^{(\ell_+)}$-measurable $\La\subset\mathbb{R}^d$ and all $k\in
\{1,\dots,S+1\}$
            \begin{equation}
     \label{aa2.57}
\Big|\log\{e^{\beta\mathbf{ I} _{k}(\La)}
Z^{(k)}_{\La,\la_{\beta,\ga}}(\chi^{(k)}_{\La^c})\}-\beta
|\La|P_{\la_{\beta,\ga}}\Big|\le c\ga^{1/4} |\delta_{\rm
in}^{\ell_+}(\La)|
    \end{equation}
with $\la_{\beta,\ga}$ and $P_{\la_{\beta,\ga}}$
 as in Theorem \ref{plabeta}.

        \end{thm}

\bigskip

 Theorem \ref{thm:2.14} is proved in Section \ref{sec:5}.

\vskip.5cm

\subsection{Conclusions}
    \label{sec:3.7}

By Theorem \ref{thm:2.14} with $\La={\rm
int}_{j}(\Ga)$,  
    \begin{equation}
\Big|\log\{e^{\beta \mathbf{ I}_{k_j}({\rm int}_{j}(\Ga))}
Z^{(k_j)}_{{\rm int}_{j}(\Ga),\la_{\beta,\ga}}(\chi^{(k_j)}_{{\rm
sp}(\Ga)})\}-\beta |{\rm int}_{j}(\Ga)|P_{\la_{\beta,\ga}}\Big|\le
c'\ga^{1/4}|\delta_{\rm out}^{\ell_+}[{\rm int}_{j}(\Ga)]|
   \label{2.29}
    \end{equation}
        \begin{equation}
\Big|\log\{e^{\beta \mathbf{ I}_{k}({\rm int}_{j}(\Ga))}
Z^{(k)}_{{\rm int}_{j}(\Ga),\la_{\beta,\ga}}(\chi^{(k)}_{{\rm
sp}(\Ga)})\}-\beta |{\rm int}_{j}(\Ga)|P_{\la_{\beta,\ga}}\Big|\le
c'\ga^{1/4}|\delta_{\rm out}^{\ell_+}[{\rm int}_{j}(\Ga)]|
   \label{2.29a}
    \end{equation}

Hence by \eqref{I.5.1}
        \begin{equation}
\frac{\mathcal N_\la^{(k)}(\Ga,q_{c(\Ga)^c})}{\mathcal
D_\la^{(k)}(\Ga,q_{c(\Ga)^c})} \le e^{-\beta c_1 \zeta^2 \ell_-^d
N_\Ga} e^{2c' \ga^{1/4} \ell_+^d N_\Ga}\le e^{-\beta (c_1/2) \zeta^2
\ell_-^d N_\Ga}
    \end{equation}
the last holding for all $\ga$ small enough. By \eqref{I.5.1bis} we
have then proved \eqref{3.9}  and \eqref{3.10} yields
        \begin{equation}
    \label{3.10bis}
W^{k,\text{true}}(\Ga|q^{(k)}) \le e^{-\beta (c_1/2) \zeta^2
\ell_-^d N_\Ga} ,\quad \text{for $\ga$ small enough}
    \end{equation}

\vskip1cm

\section{{\bf  Energy estimate }}
        \label{sec:4}

The proof of the energy estimate \eqref{I.5.1}, is divided into two
steps. The first step (Theorem \ref{thm2.2-11} below) is the proof
that it is possible to reduce to ``perfect boundary conditions'',
namely to reduce the analysis to partition functions with ``perfect
boundary conditions'', i.e. with boundary conditions $\rho^{(k)}$,
one of the minimizers of the mean field free energy functional. This
implies that we can factorize with a negligible error the estimate
in int$(\Ga)$ from the one in sp$(\Ga)$. The second step in the
proof of \eqref{I.5.1} involves a bound on the constrained partition
function in ${\rm sp}(\Ga)$ which yields the gain factor $e^{-\beta
c_1 \zeta^2 \ell_-^d N_\Ga}$.

\vskip1cm

\subsection{Reduction to perfect boundary conditions}
    \label{subsec:4.1}

\vskip.5cm
 Without loss of generality  we fix
$k$ and restrict to  contours $\Ga=($sp$(\Ga),\eta_\Ga)$ with color
$k$ and define regions in $c(\Ga)$  as follows. The construction is
the same as that used in Chapter 11, Subsection 11.2.1 of
\cite{leipzig} to which we refer.

\bigskip

We denote by  $\rm{int}_{h} (\Ga)$  the maximal connected components
$\La\subset$ int$(\Ga)$ such that $\Theta(r;q)=h$ for all $r\in
\delta_{\rm{in}}^{\ell_+}[\La]$. We call
 $$
\Delta_1:=\delta_{\rm out}^\ell[c(\Ga)^c],\;\;
 \Delta_2:=\delta_{\rm out}^\ell[c(\Ga)^c\cup \Delta_1],\;\;
 \Delta_3:=\delta_{\rm out}^\ell[c(\Ga)^c\cup
  \Delta_1\cup \Delta_2],\qquad \ell:=\ell_+/8
   $$
These are successive corridors that we meet when we move from
$c(\Ga)^c$ into ${\rm sp}(\Ga)$.  In all of them $\eta=k$ and the
region where $\eta\ne k$ is far away, by $\dis{\ell_+-\frac
38\ell_+}$. When approaching sp$(\Ga)$ from int$_h(\Ga)$ we see:
  $$
\Delta^{(h)}_4:=\delta_{\rm out}^\ell[{\rm int}_{h}(\Ga)]\cup
\delta_{\rm in}^\ell[{\rm int}_{h}(\Ga)],\;\;
\Delta^{(h)}_5:=\delta_{\rm out}^\ell[\Delta^{(h)}_4],\qquad h
\text { may also be }=k
  $$
By the definition of contours the above $h$ corridors are in a
region where $\eta_\Ga=h$,
and the distance from where $\eta_\Ga\ne h$,  is
$\dis{\ell_+-\frac 28\ell_+}$.
 We then call
        \begin{equation}
        \label{4.0}
 B= B^{(k)}\bigcup_{h\ne k} B^{(h)},\qquad
B_0= B_0^{(k)}\bigcup_{h\ne k} B^{(h)}
    \end{equation}
where
     \begin{equation}
         \label{4.00}
 B^{(h)} = \Delta_4^{(h)}\cup \Delta_5^{(h)},\;\;  B^{(k)} = \Delta_4^{(k)}\cup \Delta_5^{(k)}\cup \Delta_1\cup \Delta_2\cup \Delta_3
     \end{equation}
           \begin{eqnarray}
&&  B_0^{(h)} = \Delta_4^{(h)}\cap {\rm sp}(\Ga), \;\;
 B^{(k)} _0=\Delta_2   \cup  (\Delta_4^{(k)}\cap {\rm sp}(\Ga))
    \label{4.01}
    \end{eqnarray}
and finally, letting $\Delta_4=\bigcup_h\Delta_4^{(h)}$,
  $\Delta_5=\bigcup_h\Delta_5^{(h)}$, we define
      \begin{equation}
          \label{a4.0.0}
 \La = {\rm sp}(\Ga)\setminus \big(\Delta_1\cup
 \Delta_2\cup(\Delta_4\cap {\rm sp}(\Ga))\big),\qquad
 \La' =\La \setminus \big(\Delta_3\cup ( \Delta_5\cap {\rm sp}(\Ga))\big)\
     \end{equation}
Observe that the points in $\La$ interact only with those in
$B_0\subset\La^c$.

\medskip

After a Lebowitz-Penrose coarse graining in $B$ we will reduce to a
variational problem with the free energy functional $ F^*$ defined
in \eqref{a3.11} below. We will prove existence and uniqueness of
minimizers and their stability properties concluding that with
``negligible error''  we can  ``eliminate" the corridor $B_0$ in
$\delta_{\rm in}^{\ell_+}[$ sp$(\Ga)]$ which separates int$(\Ga)$
from the rest of sp$(\Ga)$, their interaction with $B_0$ being
replaced by an interaction with perfect boundary conditions.

  \vskip.5cm

Given any set $\Omega\subset \mathbb{R}^d$,  any b.c. $\bar
q_{\Omega^c}$ and any measurable set $A\subset  \mathcal{Q}_\Omega$,
we call
    \begin{equation}
      \label{2.41nn}
\hat Z_{\Omega,\la} (A\big|\bar q_{\Omega^c})= \int_{A} e^{-\beta
H_{\Omega,\la}(q_{\Omega}|\bar q_{\Omega^c})}\nu(dq_\Omega)
    \end{equation}

 \vskip1cm

   \begin{thm}  [Reduction to perfect boundary conditions]
   \label{thm2.2-11}
 There are $c,c'>0$ such that for all  $\la$:
$|\la-\la_\beta|\le c'\ga^{1/2}$, for all $\ga$ small enough for any
$k\in\{1,\dots S+1\}$, any contour $\Ga=({\rm sp}(\Ga),\eta_\Ga)$ of
color $k$ and any boundary condition $q_{c(\Ga)^c}\in
\mathcal{X}^{(k)}_{c(\Ga)^c}$ the following holds.

Recalling  \eqref{2.38}, we call
    \begin{equation}
      \label{2.38nn}
\chi_{B_0}=\sum_{h\ne
k}\chi^{(h)}_{B^{(h)}_0}+\chi^{(k)}_{B^{(k)}_0}
     \end{equation}
Then, with $ F^*$ defined in \eqref{a3.11} below, and $\La$ as in
\eqref{a4.0.0}
      \begin{eqnarray}
         \nn
&& \hskip-1.5cm  {\mathcal N}^{(k)}_{\la}(\Ga,q_{c(\Ga)^c})\le
\;e^{-\beta F^*_{B_0,\la_\beta}(\chi_{B_0})+c\ga^{1/2}|B|} \prod_{h}
{ Z}^{(h)}_{{\rm int}_{(h)}(\Ga),\la} (\chi^{(h)}_{B^{(h)}_0})
\\&&\nn\hskip1.5cm \times \hat Z_{\Delta_1,\la}
\Big(\mathcal{X}^{(k)}_{\Delta_1}\big|\chi^{(k)}_{B^{(k)}_0}\cup
q_{c(\Ga)^c}\Big)\\&&\hskip1.5cm\times   \hat Z_{\La, \la}
\Big(\{\eta(q_\La ;\cdot)=\eta_\Ga(\cdot)\}\big|\chi_{B_0}\Big)
     \label{2.45n}
     \end{eqnarray}
   \end{thm}

\vskip.5cm

We first rewrite ${\mathcal N}^{(k)}_{\la}(\Ga,q_{c(\Ga)^c})$ as follows.

   \begin{lemma}
   \label{lemma13n.2.1-1}

There is a non negative, bounded function $\phi(q_{c(\Ga)\setminus
B})$ whose explicit expression is given in \eqref{4.5n} below, which
vanishes unless the phase indicator $\eta(r;q_{c(\Ga)\setminus B})$
verifies
    \begin{equation}
      \label{4.2n}
\eta(r;q_{c(\Ga)\setminus B}) = \begin{cases} k &\text{if
dist$(r,B^{(k)})\le 2\ga^{-1}$},
\\h&\text{if
dist $(r,B^{(h)})\le 2\ga^{-1}$}\end{cases}
     \end{equation}
This function $\phi$ is such that
       \begin{eqnarray}
         \nn
&& \hskip-2.3cm {\mathcal{ N}}^{(k)}_{\ga,\la}(\Ga,\bar
q_{c(\Ga)^c})= \int \phi(q_{c(\Ga)\setminus B}) \hat Z_{B^{(k)},\la}
\Big(\mathcal{X}^{(k)}_{ B^{(k)}}\big|q_{c(\Ga)\setminus B}\cup \bar
q_{c(\Ga)^c}\Big)
\\&& \hskip1cm\prod_{h\ne k}\hat
Z_{B^{(h)},\la} \Big(\mathcal{X}^{(h)}_{B^{(h)}}\big|
q_{c(\Ga)\setminus B}\Big)d\nu(q_{c(\Ga)\setminus B})
         \label{4.3n}
     \end{eqnarray}

   \end{lemma}

  \vskip.5cm

  {\bf Proof.} The argument is the same as Lemma 11.2.2.3 of \cite{leipzig},
  for the reader convenience we report it.
  We drop the dependence on
  $\la$   from the notation.

  We call
    $$\La_h:={\rm int}_{(h)}(\Ga)\cap B^c$$
and we
  define
         \begin{equation*}
 \phi_h(q_{\La_h})= \text {\bf 1}_{\{\eta(r;q_{\La_h})=h, \forall
r\in \La_h\}} e^{-\beta H_{\La_h}(q_{\La_h})} \mathbf{X}^{(h)}_{{\rm
int}_{(h)}(\Ga)} ( q_{\La_h})
     \end{equation*}
Recall  that $\mathbf{X}^{(h)}_{A,\la} ( q_{A})$ depends only on the
restriction of $q_A$ to $A'$ where $A'$ is the set of all $r$ at
distance $\le \ga^{-1}$ from $A^0=A\setminus \delta_{\rm
in}^{\ell_{+}}[A]$. In fact  all contours $\Ga$ which contribute to
$\mathbf{X}^{(h)}_{A,\la}(q_A) $ have spatial support in $A^0$. For
this reason we can change arbitrarily $q_A$ in the complement of
$\{r : {\rm dist}(r,A^0)> 2\ga^{-1}\}$ leaving unchanged
$\mathbf{X}^{(h)}_{A,\la}(q_A) $. Thus
    \begin{eqnarray}
     \label{4.4n}
\mathbf{X}^{(h)}_{{\rm int}_{(h)}(\Ga)} ( q_{{\rm int}_{(h)}(\Ga)})=
\mathbf{X}^{(h)}_{{\rm int}_{(h)}(\Ga)} ( q_{\La_h})
     \end{eqnarray}
because $\La_h$ contains all $r:{\rm dist}(r,{\rm
int}_{(h)}(\Ga)\setminus \delta_{\rm in}^{\ell_+}[{\rm
int}_{(h)}(\Ga)])\le \ga^{-1}$.

Analogously, calling
    $$
\La_0={\rm sp}(\Ga)\setminus B
        $$
we define
    $$
\phi_0(q_{\La_0})= \text {\bf 1}_{\{\eta(r;q_{\La_0})=\eta_\Ga(r),
r\in \La_0 \}} e^{-\beta H_{\La_0}(q_{\La_0})}
    $$
Since
    $$q_{c(\Ga)\setminus B}=q_{\La_0}\cup q_{\La_k}\bigcup_{h\ne
k}q_{\La_h}
    $$
by setting
         \begin{eqnarray}
         \label{4.5n}
&& \phi(q_{c(\Ga)\setminus B})=\phi_0(q_{\La_0})
\phi_k(q_{\La_k})\,\prod_{h\ne k} \phi_h(q_{\La_h})
     \end{eqnarray}
and recalling \eqref{4.4n}, \eqref{4.3n} becomes an identity.
\qed

\vskip1cm
We postpone the proof of Theorem \ref{thm2.2-11}
since it uses a coarse graining in $B$ that reduce
to a minimization problem for the free energy functional
that we define in the next subsection.

\vskip1cm

\subsection{Free Energy Functional}
    \label{sec:4.2}

\vskip.5cm

The LP (Lebowitz Penrose) free energy functional $F^*$, defined in
\eqref{a3.11} below,   is a $\ga^{-1/2}$-discretization of
\eqref{z1.8}.

We start by recalling properties of the mean field free energy
proved in \cite{DMPV2}.

Recalling \eqref{z1.1}, we call $F^{\rm mf}_{\la_\beta}=F^{\rm
mf}_{\beta,\la_\beta}$ and we observe that the  minimizers
$\rho^{(k)}$, $k=1,\dots, S+1$ are critical points  of $F^{\rm
mf}_{\la_\beta }$, namely they satisfy
        \begin{equation}
      \label{2.3}
\rho^{(k)}_s=  \exp\Big\{-\beta\{\sum_{s'\ne s} \rho^{(k)}_{s'} -
\la_\beta\} \Big\},\qquad \text {for all } s=\{1,\dots,S\}
     \end{equation}
We will denote the common minimum value of $F^{\rm mf}_{\la_\beta }$
by
    \begin{equation}
    \label{phi}
\phi=F^{\rm mf}_{\la_\beta }(\rho^{(k)})
    \end{equation}
Furthermore $F^{\rm mf}_{\la_\beta}$ is strictly convex,
 namely there is a constant $\kappa^*>0$ such that
   \begin{equation}
      \label{2.4}
\langle v,L^{(k)} v \rangle=   \sum_{s,s'} L^{(k)}(s,s') v(s) v(s')
\ge \kappa^*\langle v,v\rangle,\qquad L^{(k)}(s,s') =
\frac{1}{\beta\rho^{(k)}_s}
 \text{\bf 1}_{s=s'} + \text{\bf 1}_{s\ne s'}
     \end{equation}

\bigskip

 Let $\La$ be a $D^{(\ga^{-1/2})}$- measurable bounded region of
$\mathbb R^d$. Given two non negative functions, $\rho$ and $\bar
\rho$ defined in $\mathbf{R}^d$ and $D^{(\ga^{-1/2})}$- measurable,
we call $\rho_\La$ the function equal to $\rho$ in $\La$ and equal
to 0 in $\La^c$. Analogous definition for $\bar \rho_{\La^c}$.

We define
   \begin{equation}
      \label{a3.11}
F^*_{\La,\la}(\rho_\La|\bar \rho_{\La^c}) =F^*_{\La,\la}(\rho_\La)+
\ga^{-d/2}(\rho_\La, V^*  \bar \rho_{\La^c}),\qquad
F^*_{\La,\la}(\rho_\La)=\ga^{-d/2}\Big\{\frac 12 (\rho_\La, V^*
\rho_\La)- \frac 1\beta (1_\La,\mathcal{I}_\la(\rho_\La))\Big\}
     \end{equation}
where, setting  $\mathfrak{L}_\ga=\ga^{-1/2}\mathbb{Z}^d$,
   \begin{equation}
      \label{a3.12}
(f,g)= \sum_{x\in\mathfrak{L}_\ga} \sum_{s=1}^{S}f(x,s)
g(x,s)
     \end{equation}
and the matrix $V^*=\big(V^*(x,s,x',s'), x,x'\in
\mathfrak{L}_\ga\cap\La, s,s'\in\{1,\dots,S\}\big)$ is given by
   \begin{equation}
      \label{a3.13}
 V^*(x,s,x',s')=\text{\bf 1}_{s'\ne s} \hat V_\ga(x,x')
     \end{equation}
     \begin{equation}
    \label{2.9}
\hat V_{\ga}(x,x')=\ga^{-d/2} \sum_{y\in \mathfrak{L}_\ga} \hat
J_\ga(x,y)\ga^{-d/2}\hat J_\ga(y,x')
        \end{equation}
         \begin{equation}
    \label{2.8}
\hat J_\ga(x,y)=\mintone{C^{(\ga^{-1/2})}_x}
\mintone{C^{(\ga^{-1/2})}_y} J_{\ga}(r,r')drdr'
        \end{equation}
 Finally $\mathcal{I}_\la$ in \eqref{a3.11} is
        \begin{equation}
      \label{a3.15}
\mathcal{I}_\la(\rho_\La) (x,s)=\mathcal{S}(\rho(x,s))-
\beta\la,\quad \mathcal{S}(\rho(x,s)=- \rho_\La(x,s)[\log
\rho_\La(x,s)-1]
     \end{equation}

     \vskip2cm

The relation of this functional with the model is given in the
Theorem \ref{propA.2} below.  Let $\La$ be a $\mathcal
D^{(\ell_-)}$-measurable bounded region of $\mathbb R^d$. Recalling
the definition \eqref{a3.1b} and \eqref{a3.1}, we shorthand
$n(x,s;q)=n^{(\ga^{-1/2)}}(x,s;q)$, and
$\rho(x,s;q)=\rho^{(\ga^{-1/2)}}(x,s;q)$ $s\in\{1,\dots,S\}$, $x\in
\mathfrak{L}_\ga$.

\begin{defin} (Space of densities).
    \label{dens}

We call $\tilde {\mathcal M}_{\La,\ell}$ the space of non negative,
$\mathcal D^{(\ell)}$-measurable functions defined in $\La\subset
\mathbb{R}^d$. Thus the elements $\rho_\La\in \tilde {\mathcal
M}_{\La,\ell}$ are actually functions of finitely many variables,
i.e.\ $\{\rho_\La(x), x\in \ell \mathbb Z^d \cap \La\}$.

Given any  $u>0$  we denote by
    \begin{equation}
    \label{a4.16n}
\mathcal{B}^*(u)=\{\rho_\La\in\tilde {\mathcal M}_{\La,\ga^{-1/2}}
:\mintone{C_r^{(\ell_-)}}\rho_\La\le u\}
    \end{equation}
and
 we call
    \begin{equation}
    \label{a4.16nn}
\mathcal{B}(u):=\big\{q_\La:
\rho(\cdot;q_\La)\in\mathcal{B}^*(u)\big\}
    \end{equation}
Analogously, given any $A^*\subset \mathcal{B}^*(u)$, we call
$A=\big\{q_\La: \rho(\cdot;q_\La)\in A^*\big\}$.
 \end{defin}
 Given  a
configuration $\bar q_{\La^c}$  and recalling the definition
of the constrained partition
function given in \eqref{2.41nn}, we have:

            \begin{thm}
 \label{propA.2}
There is  $\bar c>0$ such that the following holds. For any
$\rho^*>$$\dis{\max_h\max_s\rho^{(h)}_s}$ and for any $A^*\subset
\mathcal{B}^*(\rho^*)$
            \begin{equation}
      \label{4.12}
\Big| \log
 \hat  Z_{\La,\la}\big(A|\bar q_{\La^c}\big)
 +\beta\inf_{\rho_\La\in A^*} F^*_{\La,\la}(\rho_{\La}|\bar
q_{\La^c})\Big|\le \bar c \ga^{1/2} |\La|
     \end{equation}
    \end{thm}

The proof of Theorem \ref{propA.2} is
the same as the proof of Theorem 11.1.3.3 of \cite{leipzig} and it is omitted.

\bigskip

In Theorem below we state what we need in order to prove
\eqref{2.45n}, its proof is postponed to Section \ref{sec:I.2}.
        \begin{thm}
        \label{thm:4.3}
Let  $\Ga$ be a contour of color $k$ and $\bar q_{c(\Ga)^c}\in
\mathcal{X}^{(k)}_{c(\Ga)^c}$. Let $B$ be the set defined in
\eqref{4.0}. Given any $q_{c(\Ga)\setminus B}$ such that $\eta=h$ in
$\delta^{\ga^{-1}}_{\rm out}[B^{(h)}]$, for all $h$, let $\bar
q:=q_{c(\Ga)\setminus B}\cup \bar q_{c(\Ga)^c}$.

There are positive constants $c$ and $\om$ and for all $h$ there are
positive functions $\rho^*_{B^{(h)}}\in \tilde{\mathcal
M}_{B^{(h)},\ga^{-1/2}}\cap \mathcal{X}^{(h)}_{B^{(h)}}$, such that
 \begin{equation}
          \label{4.9n}
\rho_{B^{(h)}}^*(r)=\rho^{(h)},\quad \forall r\in B_0^{(h)}
    \end{equation}
and furthermore for all $h$ and all $\rho_{B^{(h)}}\in
\tilde{\mathcal M}_{B^{(h)},\ga^{-1/2}}\cap
\mathcal{X}^{(h)}_{B^{(h)}}$,
    \begin{equation}
          \label{4.11n}
F_{B^{(h)},
\la_{\beta}}^*(\rho_{B^{(h)}}|\bar q_{c(\Ga)\setminus B})\ge
F_{B^{(h)}, \la_{\beta}}^*(\rho^*_{B^{(h)}}|\bar q_{c(\Ga)\setminus
B})- c e^{-\om \ga\ell_+}
    \end{equation}

        \end{thm}
 \vskip1cm

\subsection{Reduction to perfect boundary conditions, conclusion}
    \label{sec:4.3}

 Theorem \ref{thm:4.3} is the main model-dependent estimate needed for the proof of
 Theorem \ref{thm2.2-11} the others arguments are
 the same as in Subsection 11.2.2 of \cite{leipzig}.
We thus only sketch them.

\medskip
Going back to \eqref{4.3n}, we first  observe that if
$\rho_{B^{(h)}}\in \mathcal{X}^{(h)}_{B^{(h)}}$ then
$\rho_{B^{(h)}}\in \mathcal{B}^*(u)$ (see \eqref{a4.16n}) with
$u=\rho^{(h)}+\zeta$. Then by Theorem \ref{propA.2}
    \begin{equation}
          \label{4.7n}
\log \hat Z_{B^{(h)},\la} \Big(\mathcal{X}^{(h)}_{B^{(h)}}\big|\bar
q_h\Big)\le \bar c\ga^{1/2}|B^{(h)}|-\beta
\inf_{M_h}F_{B^{(h)}, \la}^*(\rho_{B^{(h)}}|\bar q_h)
    \end{equation}
where we have set $M_h=\tilde{\mathcal M}_{B^{(h)},\ga^{-1/2}}\cap
\mathcal{X}^{(h)}_{B^{(h)}}$,  $\bar q_k=q_{c(\Ga)\setminus B}\cup
\bar q_{c(\Ga)^c}$ and $\bar q_h=q_{c(\Ga)\setminus B}$ if $h\ne k$.
 Since  the dependence on $\la$ in
$F^*_{B^{(h)},\la}$ is given by the term
$\dis{-\la\int_{B^{(h)}}\rho_{B^{(h)}}(r)dr}$, for all $\la$:
$|\la-\la_\beta|\le c'\ga^{1/2}$,
 we have
      \begin{equation}
          \label{4.8n}
 \hskip-2cm | F^*_{B^{(h)},\la}(\rho_{B^{(h)}}| q_h)-
F^*_{B^{(h)},\la_\beta}(\rho_{B^{(h)}}| q_h)| \le c \ga^{1/2}|B^{(h
)}|
     \end{equation}
Thus from \eqref{4.7n}, \eqref{4.8n} and Theorem \ref{thm:4.3} we
get
    \begin{equation}
          \label{4.21}
\log \hat Z_{B^{(h)},\la} \Big(\mathcal{X}^{(h)}_{B^{(h)}}\big|\bar
q_h\Big)\le c\ga^{1/2}|B^{(h)}|-\beta F_{B^{(h)},
\la_\beta}^*(\rho^*_{B^{(h)}}|\bar q_h)
    \end{equation}

\vskip.5cm

Using the formula
                 \begin{eqnarray}
&&\hskip-1cm    F^*_{{A\cup B},\la} (\rho_{A\cup B}|
\rho_{(A\cup B)^c})=   F^*_{ A,\la } (\rho_{A}|
\rho_{(A\cup B)^c})+ F^*_{ B,\la} (\rho_{ B}| \rho_{(A\cup
B)^c}+\rho_A)
    \label{4.12n}
     \end{eqnarray}
 setting
                 \begin{eqnarray}
   \label{4.13n}
&& \hskip-1cm B^{(h)}_1:=(B^{(h)}\setminus B_0^{(h)})\cap {\rm sp
}(\Ga),\;\; B^{(h)}_2:=B^{(h)}\cap {\rm int }_{(h)}(\Ga),\quad h\ne
k
    \end{eqnarray}
and using \eqref{4.9n}, we have
       \begin{eqnarray}
          \label{4.14n}
&& \hskip-2cm
  F_{B^{(h)}, \la_{\beta}}^*(\rho^*_{B^{(h)}}| q_h)=   F_{B_0^{(h)},
  \la_{\beta}}^*(\chi^{(h)}_{B_0^{(h)}}) +
F_{B_1^{(h)}, \la_{\beta}}^*(\rho^*_{B^{(h)}_1}|
\chi^{(h)}_{B^{(h)}_0}+ q_{{\rm sp}(\Ga)\cap B^c})\nn\\&&\hskip1cm +
F_{B_2^{(h)}, \la_{\beta}}^*(\rho^*_{B^{(h)}_2}|
\chi^{(h)}_{B^{(h)}_0}+ q_{{\rm int}_h(\Ga)\cap B^c})
     \end{eqnarray}
By using that $|\la-\la_\beta|\le c'\ga^{1/2}$ we replace
$\la_\beta$ by $\la$ in the last two terms with an error bounded by
$c\ga^{1/2}|B|$ and then use Theorem \ref{propA.2} ``backwards'' to
reconstruct partition functions. We then have
       \begin{eqnarray}
          \nn
&& \hskip-1.4cm
 -\beta  F_{B^{(h)}, \la_{\beta}}^*(\rho^*_{B^{(h)}}|\bar q_h)
 \le \bar c\ga^{1/2}|B^{(h)}|-\beta F_{B_0^{(h)},
\la_{\beta}}^*(\chi^{(h)}_{B_0^{(h)}})
            \nn\\&&
\hskip1cm+ \log\hat Z_{B_1^{(h)},\la}
\Big(\mathcal{X}^{(h)}_{B_1^{(h)}}\big |\chi^{(h)}_{B^{(h)}_0}\cup
q_{{\rm sp}(\Ga)\setminus B}\Big) \log\hat Z_{B_2^{(h)},\la} \Big
(\mathcal{X}^{(h)}_{B_2^{(h)}}\big|\chi^{(h)}_{B^{(h)}_0}\cup
q_{{\rm int}_{(h)}(\Ga)\cap B^c}\Big)
    \nn\\&&
=\bar c\ga^{1/2}|B^{(h)}|-\beta F_{B_0^{(h)},
\la_{\beta}}^*(\chi^{(h)}_{B_0^{(h)}})+ \log\hat Z_{\rm{int}_h,\la}
\Big(\chi^{(h)}_{B^{(h)}_0}\Big)
    \label{4.15n}
     \end{eqnarray}
Since an analogous bound holds for $-\beta F_{B^{(k)},
\la_{\beta}}^*(\rho^*_{B^{(k)}}| q_k)$, we then get \eqref{2.45n}
from \eqref{4.3n} and from \eqref{4.5n} thus proving  Theorem
\ref{thm2.2-11}.

\vskip2cm

\subsection{Reduction to uniformly bounded densities}
        \label{sec:I}
The main step to conclude the proof of the energy estimate,  is to
bound the last term in \eqref{2.45n}, namely $\hat
Z_{\La,\la}\big(\{\eta(q_\La ;\cdot)=\eta_\Ga(\cdot)\}|\chi_{B_0}
\big)$ with $\La$ as in \eqref{a4.0.0}. Also here we use coarse
graining, but since the number of particles is not bounded we cannot
apply directly  Theorem \ref{propA.2}.  The same difficulty is
present in the LMP model and the arguments used there can be
straightforwardly adapted to the present contest. The outcome is
Theorem \ref{thm:4.5} below which is the same as Proposition
11.3.0.1 of \cite{leipzig} to which we refer for proofs.

\medskip

We need some definitions.  Let $\rho_{max}$ be a positive number
such that
    \begin{equation}
    \label{2.6.0}
\log\rho_{max} >\beta[1+\la_\beta],
\qquad \rho_{\max}>\max_h\max_s\rho^{(h)}_s+\zeta,\qquad \sum_{n\ge
\rho_{max}\ell_-^d}\frac {(S\ell_-^de^{\beta\la})^n}{n!}\le
e^{-4S\rho_{max}\ell_-^d}
     \end{equation}
     \begin{equation}
     \label{D.10}
 \frac{1 }{32\beta \rho_{\max}}  \le \frac{\kappa^*}{16},\qquad
 \text{$\kappa^*$ as in \eqref{2.4}}
     \end{equation}
\medskip

For any $\rho_\La\in \tilde{\mathcal{M}}_{\La,\ga^{-1/2}}$,
recalling \eqref{a4.0.0} we call
        \begin{eqnarray*}
       \nn
\mathcal{C}_0(\rho_\La,\zeta/2, \La')=\Big\{ C^{(\ell_-)}_x \subset
\La': \eta^{(\zeta/2,\ell_-)}(\rho_{\La};x)=0\Big\},\qquad
N_0(\rho_\La,\zeta/2,\La')=|\mathcal{C}_0(\rho_\La,\zeta/2, \La')|
     \end{eqnarray*}
We also call $\mathcal{C}_{\ne}(\rho_\La,\zeta/2, \La')$ the
collection of all pairs of cubes $(C_{x_1},C_{x_2})$ both in
$\mathcal{D}^{\ell_-}$ and such that $C_{x_i}\subset \La'$, $i=1,2$,
$C_{x_1}\cap C_{x_2}\ne \emptyset$,
$\eta^{(\zeta/2,\ell_-)}(\rho_{\La};x_i)=k_i$, $i=1,2$, $k_1\ne
k_2$. We require that $\mathcal{C}_{\ne}(\rho_\La,\zeta/2, \La')$
must be maximal, namely any pair $(C_{x'_1},C_{x'_2})$ that verify
the same property must have at least one among $C_{x'_1}$ and
$C_{x'_2}$ appearing in $\mathcal{C}_{\ne}(\rho_\La,\zeta/2, \La')$.
We denote by $N_{\ne}(\rho_\La,\zeta/2,\La')$ the number of cubes
(cubes not pairs of cubes!) appearing in
$\mathcal{C}_{\ne}(\rho_\La,\zeta/2, \La')$.

\vskip.5cm

   \begin{thm}
   \label{thm:4.5}
 Let $\rho_{\max}$ be as in
\eqref{2.6.0} and \eqref{D.10}.   There is $c>0$ such that for all
$\ga$ small enough and for all $\la$ such that $|\la-\la_\beta|\le
c\ga^{1/2}$.
                 \begin{equation}
\hskip-1cm \log \hat Z_{\La,\la}\big(\{\eta(q_\La
;\cdot)=\eta_\Ga(\cdot)\}|\chi_{B_0} \big) \le
 - \min_{m\ge 0} \Big\{ m
\ell_-^d \frac{\zeta}{2}+ \beta \inf_{\rho_\La\in \mathcal
G_{2m}(\rho_{\max})}  F^*_{\La,
\la_\beta}(\rho_\La|\chi_{B_0})\Big\}+c\ga^{1/2}|\La|
     \label{4.27}
     \end{equation}
where, recalling \eqref{ngamma},
   \begin{eqnarray}
   \nn
&&\hskip-2cm \mathcal G_m(\rho_{\max}):= \Big\{\rho_\La\in \tilde{
\mathcal M}_{\La,\ga^{-1/2}}: \rho_\La(x,s)\le \rho_{\max},\,\forall
x\in\La\cap \ga^{-1/2}\mathbb{Z}^d, \forall s\text{ and }\\&&
N'_0(\rho_\La,\zeta/2,\La')+N'_{\ne}(\rho_\La,\zeta/2,\La')
 \ge \frac{3^{-d}N_\Ga-m}2
 \Big\}
    \label{4.24n}
     \end{eqnarray}
     \end{thm}

 \vskip1cm

\subsection{Free Energy cost of contours with perfect boundary
conditions}
    \label{sec:5n}

\bigskip

The main result in this Section is Theorem \ref{thm3.1} below which
gives a bound of the $\inf$ on the right hand side of \eqref{4.27}.
This estimate is the same as the one given in Theorem 11.3.2.1 of
\cite{leipzig} to which we refer for proofs.

\vskip.5cm

In the sequel $\Omega$ denotes a bounded, connected $\mathcal
D^{(\ga^{-1})}$-measurable region such that its complement is the
union of $p\ge 1$ (maximally connected) components
$\Omega_1$,...,$\Omega_p$ at mutual distance $>\ga^{-1}$ (such that
they do not interact): $\dis{\Omega^c=\bigcup_{i=1}^p \Omega_i}$.
The boundary conditions are chosen by fixing arbitrarily $k_i\in
\{1,..,S+1\}$, $i=1,..,p$, and setting $\rho^{(k_i)}$ on $\Omega_i$,
Analogously to \eqref{2.38} we will denote for any
$s\in\{1,\dots,S\}$,
     \begin{equation}
     \label{3.1}
\chi^{(\und k)}_{\Omega^c}(r,s)=\sum_{j=1}^p\chi^{(\und
k_j)}_{\Omega_j}(r,s),\qquad \chi^{(\und
k_j)}_{\Omega_j}(r,s)=\rho^{(k_j)}_s\text{\bf 1}_{r\in
\Omega_j},\qquad\und k=(k_1,\dots,k_p)
     \end{equation}

\medskip

The reason why the region $\Om$ is $D^{(\ga^{-1})}$-measurable is
that we are going to use Theorem \ref{thm3.1} with $\Omega=\La$,
$\La$ as in \eqref{a4.0.0}.

\bigskip

 Recalling \eqref{z1.5} we set for any given $\psi$,
    \begin{equation}
            \label{aa3.4b}
f^{\rm mf}(\psi)(x,s)=\dis{ e_{\la_\beta}^{\rm mf}(\psi)(x,s)- \frac
1\beta\psi(x,s)[\log\psi(x,s)-1]}
        \end{equation}
so that recalling \eqref{z1.1},
    \begin{equation}
        \label{6.48y}
\big (f^{\rm mf}( \psi),\text{\bf 1}_\Omega\big)=
\sum_{x\in\Omega\cap \mathfrak L_\ga}F_{\beta,\la_\beta }^{\rm mf}(
\mathcal R(x,1),..,\mathcal R(x,S))
    \end{equation}

\medskip
Given a kernel $K(x,y)$, $x,y\in\ga^{-1/2}\mathbb{Z}^d$, we denote
by
     \begin{equation}
        \label{6.40y}
K\star\psi(x,s):=\ga^{-d/2}\sum_{y\in\ga^{-1/2}\mathbb{Z}^d}
K(x,y)\psi(y,s)
      \end{equation}
as a consequence
     \begin{equation}
        \label{6.41y}
\hat J_\ga\star \mathbf{1}=1
      \end{equation}

\medskip

With this notation, we  the following holds.

    \begin{lemma}
    \label{lemma4.16}
Let $\Omega\subset \mathbf{R}^d$ and $\chi^{(\und k)}_{\Omega^c}$ as
above. Let $\rho_\Om$  be any $\mathcal D^{(\ga^{-1/2})}$-measurable
function. Letting
     \begin{equation}
     \label{I.1.9.2}
 \mathcal R=\hat J_\ga*(\rho_\Om+\chi^{(\und k)}_{\Omega^c})
     \end{equation}
we get
\begin{equation}
      \label{I.1.9.4}
F^*_{\Omega,\la_\beta}(\rho_\Om|\chi^{(\und k)}_{\Omega^c}) =
\ga^{-d/2}\{F_1+F_2+F_3\}
     \end{equation}
where (recall \eqref{a3.12}),
     \begin{eqnarray}
     \nn
&&  F_1=\big (f^{\rm mf}( \mathcal R),\text{\bf
1}_\Omega\big),\qquad
 F_2 = \frac {1} \beta\big(\mathcal S(
\mathcal R) - \hat J_\ga *\mathcal S(\rho_\Om+\chi_{\Omega^c}),
\text{\bf 1}_\Omega\big)
\\ &&  \nn F_3=  \big( f^{\rm mf}(\mathcal R)-
f^{\rm mf}(\hat J_\ga* \chi_{\Omega^c} ), \text{\bf 1}_{\Omega^c}
\big)+ \frac {1} \beta \big(\mathcal{S}( \mathcal R )- \hat J_\ga*
\mathcal S(\rho_\Om+\chi_{\Omega^c}),\text{\bf 1}_{\Omega^c}\big)
\\&&\hskip1cm - \big(f^{\rm mf}(\hat J_\ga*
\chi_{\Omega^c}),\text{\bf 1}_{\Omega}\big) -
 \frac {1} \beta\big(  \mathcal S(\hat J_\ga*\chi_{\Omega^c}) - \hat
 J_\ga*
 \mathcal S(\chi_{\Omega^c}),1\big)
    \label{D.8b}
     \end{eqnarray}
Finally
    \begin{equation}
     \label{D.5c}
F_{3}  \geq  \sum_{i=1}^p \mathbf{I}_{k_i}(\Omega,\Omega_i)
      \end{equation}
        \end{lemma}

\vskip.5cm

{\bf Proof.} Recalling \eqref{a3.15} and \eqref{z1.5} we rewrite
\eqref{a3.11}  as follows:
   \begin{equation}
      \label{I.1.9.1}
F^*_{\Omega,\la_\beta}(\rho_\Om|\chi^{(\und k)}_{\Omega^c}) =
\ga^{-d/2}\Big\{ \big( e_{\la_\beta}^{\rm mf}( \mathcal
R)-e_{\la_\beta}^{\rm mf}(\hat J_\ga*\chi^{(\und
k)}_{\Omega^c}),\mathbf{1} \big) - \frac 1\beta
(\mathcal{S}(\rho_\Om), \mathbf{1})\Big\}
     \end{equation}
By adding and subtracting $-\beta^{-1}(\mathcal S( \mathcal
R),\mathbf{1})$, we then have
    \begin{equation}
F^*_{\Omega,\la_\beta}(\rho_\Om|\chi^{(\und
k)}_{\Omega^c})=\ga^{-d/2}\Big\{ \big(f^{\rm mf}( \mathcal
R),\text{\bf 1}\big)-\big(e_{\la_\beta}^{\rm mf}(\hat
J_\ga*\chi^{(\und k)}_{\Omega^c}),\mathbf{1} \big)+\frac\1\beta\big(
\mathcal S( \mathcal R)-\mathcal S( \rho_\Om),\text{\bf
1}\big)\Big\}
    \label{a4.33n}
    \end{equation}
We add and subtract $-\beta^{-1}\mathcal S(\hat J_\ga *\chi^{(\und
k)}_{\Omega^c})\}$ to the second term in \eqref{a4.33n} and we add
and subtract $\big(\mathcal S( \chi^{(\und k)}_{\Om^c}),\text{\bf
1}\big)$ to the last term. We also use that by \eqref{6.41y}, for
any $\psi$, $(\hat J_\ga* \mathcal S(\psi),\text{\bf 1}\big)=\big(
\mathcal S(\psi),\text{\bf 1}\big)$. We thus get
\eqref{I.1.9.4}.

$\dis{F_3 =\sum_{i=1}^p F_3^{(i)}}$, where $F_3^{(i)}$ is
    \begin{eqnarray*}
&&
 F_3^{(i)}=\big( f^{\rm mf}(\mathcal R)-
f^{\rm mf}(\hat J_\ga* \chi_{\Omega_i} ), \text{\bf 1}_{\Omega_i}
\big)+ \frac {1} \beta \big(\mathcal S( \mathcal R )- \hat J_\ga*
\mathcal S(\rho_\Om+\chi_{\Omega_i}),\text{\bf 1}_{\Omega_i}\big)
 \\&&\hskip2cm- \big(f^{\rm mf}(\hat J_\ga*
\chi_{\Omega_i}),\text{\bf 1}_{\Omega}\big) -
 \frac {1} \beta\big(  \mathcal S(\hat J_\ga*\chi_{\Omega_i}) - \hat
 J_\ga*
 \mathcal S(\chi_{\Omega_i}),\mathbf{1}\big)
    \end{eqnarray*}
where $\chi_{\Omega_i}=\chi^{(k_i)}_{\Omega^c_i}$, see \eqref{3.1}.
By convexity the second sum in the definition of $F_3^{(i)}$ is non
negative.  Also the first term is bounded from below by replacing
replacing $\mathcal R$ by $\rho^{(k_i)}$, such that
    \begin{eqnarray*}
&&  F_3^{(i)}\ge \big( f^{\rm mf}(\rho^{(k_i)})- f^{\rm mf}(\hat
J_\ga* \chi^{(k_i)}_{\Omega_i} ), \text{\bf 1}_{\Omega^c} \big)-
\big(f^{\rm mf}(\hat J_\ga*
\chi^{(k_i)}_{\Omega_i}),\text{\bf 1}_{\Omega}\big)
            \\&&\hskip1cm  -
 \frac {1} \beta\big(  \mathcal S(\hat J_\ga*\chi^{(k_i)}_{\Omega_i}) - \hat
 J_\ga*
 \mathcal S(\chi_{\Omega_i}),1\big)
     \end{eqnarray*}
All the entropy terms cancel with each other, so we get
\eqref{D.5c}.

\qed

\vskip.5cm

Given a function $\rho$ defined in $\Omega\times \{1,\dots,S\}$ and
$\mathcal D^{(\ga^{-1})}$-measurable, let $N_0(\rho,\zeta/2,\Omega)$
and $N_{\ne}(\rho,\zeta/2,\Omega)$ be the numbers defined in
Subsection \ref{sec:I}.

 \vskip1cm

        \begin{thm}
    \label{thm3.1}
Let $c_1$ be such that $4c_13^d= \{32\beta\rho_{\max}\}^{-1}$ with
$\rho_{\max}$ as in Theorem \ref{thm:4.5}. For any $\Omega\subset
\mathbf{R}^d$, $\chi^{(\und k)}_{\Omega^c}$ and $\rho_\Om$  as in
Lemma \ref{lemma4.16}, if $\dis{\sup_{x,s}\rho_\Om(x,s) \le
\rho_{\max}}$,  then
   \begin{equation}
      \label{I.1.9}
F^*_{\Omega,\la_\beta}\big(\rho_\Om|\chi^{(\und k)}_{\Omega^c}\big)
\ge \phi|\Omega| + \ga^{-d/2}\sum_{i=1}^p I_{k_i}(\Omega,\Omega_i) +
2\cdot 3^dc_1 [N_0(\rho_\Om,\zeta/2,\Omega)
+N_{\ne}(\rho_\Om,\zeta/2,\Omega)] \zeta^2\ell_-^d
     \end{equation}
where  $\phi$  is defined in \eqref{phi} and
\begin{equation}
    \label{4.30n}
\mathbf{I}_{k_i}(\Omega,\Omega_i):=\Big( [e_{\la_\beta}^{\rm
mf}(\rho^{(k_i)})- e_{\la_\beta}^{\rm mf}(\hat J_\ga*
\chi^{(k_i)}_{\Omega_i} )], \text{\bf 1}_{\Omega_i} \Big) -
\big(e_{\la_\beta}^{\rm mf}(\hat J_\ga*
\chi^{(k_i)}_{\Omega_i}),\text{\bf 1}_{\Omega}\big)
    \end{equation}

\end{thm}

\vskip.5cm

 The proof is the same as the one of  Theorem
11.3.2.1 of \cite{leipzig} and thus is omitted.

\vskip1cm

\subsection{Energy estimate, conclusion}
    \label{sec:4.4n}
In this Subsection we conclude the proof of Theorem \ref{thm:2.9},
the arguments we use are taken from Subsection 11.3.3 of
\cite{leipzig}.

 By \eqref{4.27}   and \eqref{I.1.9}
 there is a  constant $c$ such that
        \begin{eqnarray*}
&& \hskip-3cm \log \hat Z_{\La,\la_{\beta,\ga}}\Big(\{\eta(q_\La
;\cdot)=\eta_\Ga(\cdot)\}\big|\chi_{B_0}\Big)\le -\beta\Big(
\phi|\La| +\sum_{h} I_{h}(\La; B_0^{(h)}) \nn\\&&+  \min_{ m\ge
0}\{m \ell_-^d \frac{\zeta}{2}+2\cdot 3^d c_1 \frac{3^{-d}N_\Ga
-2m}{2} \zeta^2\ell_-^d\}\Big) + c\ga^{1/2}|{\rm sp}(\Ga)|\nn
     \end{eqnarray*}
For all $\ga$ small enough (recall that $\zeta=\ga^a$), the min is
achieved at $m=0$ such that ($3^d>1$),
        \begin{eqnarray}
        && \hskip-3.7cm
\log \hat Z_{\La,\la_{\beta,\ga}}\Big(\{\eta(q_\La
;\cdot)=\eta_\Ga(\cdot)\}\big|\chi_{B_0}\Big)\le -\beta\Big(
\phi|\La| +\sum_{h} \mathbf{I}_{h}(\La; B_0^{(h)})
        \nn\\&& \hskip2cm
+ c_1 N_\Ga \zeta^2\ell_-^d \Big) + c\ga^{1/2}|{\rm sp}(\Ga)|
   \label{4.49}
     \end{eqnarray}
which inserted in \eqref{2.45n} yields with a new
constant $c$:
      \begin{eqnarray}
          \label{4.50}
&& \hskip-1.5cm  {\mathcal N}^{(k)}_{\la}(\Ga,q_{c(\Ga)^c})\le
\;e^{-\beta   c_1 N_\Ga \zeta^2\ell_-^d + c \ga^{1/2}|{\rm
sp}(\Ga)|} \prod_{h} Z^{(h)}_{{\rm int}_{(h)}(\Ga)),\la}
(\chi^{(h)}_{B_0^{(h)}})
        \nn\\&&\hskip1cm
\times \hat Z_{\Delta_1,\la
}\Big(\mathcal{X}^{(k)}_{\Delta_1}\big|\chi^k_{B^{(k)}_0}+q_{c(\Ga)^c}\Big)
\nn\\&&\hskip1cm\times
e^{-\beta\big(F^*_{B_0,\la_\beta}(\chi_{B_0})+\phi|\La| + \sum_{h}
\mathbf{I}_{h}(\La; B_0^{(h)})\big)}
     \end{eqnarray}
Since $F^*_{B_0,\la_\beta}(\chi_{B_0})=\sum_h
F^*_{B^{(h)}_0,\la_\beta}(\chi^{(h)}_{B^{(h)}_0})$, we have
    \begin{equation*}
F^*_{B_0,\la_\beta}(\chi_{B_0})=\phi|B_0|-\mathbf{I}_k(\Delta_1;B_0^{(k)})
 -\sum_{h}\big[\mathbf{I}_h(\La; B_0^{(h)})
-\mathbf{I}_h({\rm int}_h(\Ga);B_0^{(h)})\big]
    \end{equation*}
therefore the exponent in the last term of
\eqref{4.50} becomes
      \begin{equation*}
 -\beta\Big( \phi|\La\cup B_0|-I_k(\Delta_1;B_0^{(k)})
-\sum_{h}\mathbf{I}_h({\rm int}_h(\Ga);B_0^{(h)})\Big)
     \end{equation*}
Going backwards we write
      \begin{eqnarray*}
\phi|\La\cup B_0| = F^*_{\La\cup B_0,\la_\beta}(\chi^{(k)}_{\La\cup
B_0})+ \mathbf{I}_k(\Delta_1; B_0^{(k)})+\mathbf{I}_k( {\rm
int}(\Ga);B_0)
     \end{eqnarray*}
thus after some cancelations
      \begin{eqnarray}
          \label{4.53n}
&& \hskip-1.5cm  {\mathcal N}^{(k)}_{\la}(\Ga,q_{c(\Ga)^c})\le K
 \;e^{-\beta   c_1
N_\Ga \zeta^2\ell_-^d + c \ga^{1/2}|{\rm sp}(\Ga)|} \prod_{h\ne
k}\;\frac {Z^{(h)}_{{\rm int}_{(h)}(\Ga),\la} (\chi^{h}_{B^{(h)}_0})
 e^{\beta \mathbf{I}_h({\rm
int}_h(\Ga);B_0^{(h)})}}{Z^{(k)}_{{\rm int}_{h}(\Ga),\la}
(\chi^{(k)}_{B^{(h)}_0})
 e^{\beta \mathbf{I}_k({\rm
int}_h(\Ga);B_0^{(h)})}}
       \\&&\hskip-1.5cm \nn
 K:=Z^{(k)}_{{\rm
int}_{k}(\Ga),\la}(\chi^{(k)}_{B_0^{(k)}})\prod_{h:h\ne k}
Z^{(k)}_{{\rm int}_{h}(\Ga),\la} (\chi^{(k)}_{B^{(h)}_0})
  \;\hat Z_{\Delta_1,\la}\Big(\mathcal{X}^{(k)}_{\Delta_1}\big|
\chi^k_{B^{(k)}_0}+q_{c(\Ga)^c}\Big) e^{-\beta F^*_{\La\cup
B_0,\la_\beta}(\chi^{(k)}_{\La\cup B_0}\big)}
     \end{eqnarray}

Let $\rho = [\rho^{(k)}]$,
 $ [\rho^{(k)}]$ be the  element of
 $\{ \ga^{d/2}n, n\in \mathbb N\}$ closest to $\rho^{(k)}$ and let
 $\mathcal B=\big\{q_{\La\cup B_0}:
 \rho^{(\ga^{-1/2})}(\cdot;q_{\La\cup B_0})=[\rho^{(k)}]\big\}$.
Then for all
 $\la$ such that $|\la-\la_\beta|\le c\ga^{1/2}$,
   $$
e^{-\beta{\mathcal F}^*_{\La\cup B_0,\la_\beta}(\chi^{(k)}_{\La \cup
B_0}) }\le \hat Z_{\La \cup B_0,\la}(\mathcal B) e^{c\ga^{1/2}|\La
\cup B_0|}
  $$
Let $\bar q$ be any configuration in $\La\cup B^{(k)}_0$ such that
$\rho^{(\ga^{-1/2})}(\cdot;\bar q) =[\rho^{(k)}]$, then
   $$
 Z^{(k)}_{{\rm
int}_k(\Ga),\la}  (\chi^{(k)}_{B^{(k)}_0})\le    Z^{(k)}_{{\rm
int}_k(\Ga),\la} (\bar q) e^{c \ga^{d/2} |B^{(k)}_0|}
  $$
Analogous bounds hold for the other partition functions such that
$K$ is bounded by
    \begin{equation}
    \label{6.60y}
K\le e^{c\ga^{1/2}|{\rm sp}(\Ga)|}
\mathcal{D}^{(k)}_\la(\Ga,q_{c(\Ga)^c})
    \end{equation}
From \eqref{4.53n} and \eqref{6.60y}, \eqref{I.5.1} follows with
$c_1$ as in Theorem \ref{thm3.1}. \qed

\vskip1cm

\section{{\bf Critical points and minimizers
 in a pure phase }}
        \label{sec:I.2}
In this Section we prove Theorem \ref{thm:4.3}, by
analyzing a variational problem for the free energy
functional. We need a similar result in Section \ref{sec:5}
in the proof of Theorem \ref{thm:2.14} for an interpolated
functional. We thus state the result in a way that includes
both cases.

We study the variational problem
 \begin{equation}
      \label{I.2.1.1}
\min_{\rho_\La\in \mathcal X_\La^{(k)}}f_{\La, t} (\rho_\La;\bar
\rho_{\La^c}),\qquad \bar \rho_{\La^c} \in \mathcal
X_{\La^c}^{(k)},\quad t\in[0,1]
     \end{equation}
where, recalling \eqref{a3.15},
     \begin{equation}
     \label{appK.1}
f_{\La, t} (\rho_\La;\bar \rho_{\La^c})= tF^*_{\La,\la_\beta}
(\rho_\La|\bar\rho_{\La^c})+(1-t)\Big[h(\rho_\La)-\frac
1\beta\sum_s\sum_{x\in
 \La\cap \ga^{-d/2}\mathbf{Z}^d}S(\rho_\La(x,s))\Big]
     \end{equation}
where $F^*$ is defined in \eqref{a3.11} and
    \begin{equation}
h(\rho_\La) =\sum_s(r^{(k)}_s-\la_\beta)\ga^{-d/2}\sum_{x\in
 \La\cap \ga^{-d/2}\mathbf{Z}^d}\rho_\La(x,s),\qquad
r^{(k)}_s=\sum_{s'\ne s} \rho^{(k)}_{s'}
    \label{aa2.47}
    \end{equation}
  We  prove that away from $\La^c$ the minimizer is
exponentially close to $\rho^{(k)}$.
 The constraint $\rho_\La\in\mathcal X_\La^{(k)}$, namely
$\eta(\cdot;\rho_\La)=k$, is essential as it localizes the
problem in a neighborhood of the (stable) minimum where it
is possible to prove that the critical points, i.e.\ the
solutions of the Euler-Lagrange equations, converge
exponentially to $\rho^{(k)}$.

\bigskip

Thus in this Section we prove the following result.

    \begin{thm}
    \label{thmI.3.1}
 Let $\La$ be a
$\mathcal D^{(\ell_+)}$-measurable set  and let
$\bar\rho_{\La^c}\in \mathcal X^{(k)}_{\La^c}$. Then, for
any $t\in[0,1]$, there is a unique minimizer $\hat
\rho_\La\in \mathcal X^{(k)}_{\La}$
 of the variational problem \eqref{I.2.1.1}.
Furthermore there are $c$ and $\hat \om$ both positive such that
   \begin{equation}
|\hat\rho_\La(r,s)-\rho^{(k)}_s|\le c e^{-\ga \om {\rm
dist}(r,\La^c)},\qquad \text{  and all $s\in \{1,..,S\}$ and all
$r\in\La$}
 \label{I.3.222}
     \end{equation}
 \end{thm}

\vskip.5cm

Observe that Theorem \ref{thm:4.3} is a corollary of the
above Theorem for $t=1$.

\vskip1cm

In Section 5 of \cite{DMPV2} we have studied a similar variational
problem but there the constraint was on the single variable
$\rho_\La\in \mathcal{X}^{(k)}_\La$ because the functional
considered there had as main
 term   the Lebowitz-Penrose free energy on the scale $\ell_-$. Here
 we have a ''simpler'' functional but we have to face
  the new
problem of controlling the fluctuations of $\rho_\La$ from its
average on cubes of site $\ell_-$.


\subsection{Extra notation and definitions}

       \label{subsec:6.1}


In this Section we will use both the lattices
$\ga^{-1/2}\mathbb{Z}^d$ and $\ell_-\mathbb{Z}^d$. We thus define
the following.

\vskip.5cm

 $\bullet$\; $\mathcal H_\ell$ denotes the
Euclidean space of vectors $u=\big(u(x,s), x \in \La\cap \ell\mathbb
Z^d, s \in\{1,.,S\}\big)$ with the usual scalar product
$\dis{(u,v)_\ell= \sum_{x\in \La\cap \ell\mathbb Z^d}\sum_{s=1}^S
u(x,s) v(y,s)}$. For $\ell=\ga^{-1/2}$ we simply write
$(\cdot,\cdot)$.

\smallskip

 $\bullet$\;  For any  $x\in \ga^{-1/2}\mathbb Z^d$
 we denote by $z_x\in \ell_-\mathbb Z^d$ the
point such that $x\in C^{(\ell_-)}_{z_x}$. For
 $u=(u(y,s), y\in \ga^{-1/2}\mathbb Z^d, s\in \{1,..,S\})$ we let
  \begin{eqnarray}
      \label{a6.5}
&&\hskip-1cm  {\rm Av}(u;x,s) = \frac {\ga^{-d/2}}{\ell_-^d}
\sum_{y\in C^{(\ell_-)}_{z_x}\cap \ga^{-1/2}\mathbb Z^d }u(y,s)
     \end{eqnarray}
and we observe that
    \begin{eqnarray}
      \label{appK.4}
&&\hskip-1cm \mathcal{X}^{(k)}_\La=\{\rho_\La: \big| {\rm
Av}(\rho_\La;x,s) -\rho^{k)}_{s}\big|\le \zeta, {\text {
for all $x\in \ga^{-1/2}\mathbb Z^d$, $s=1,\dots,S$}}\})
     \end{eqnarray}

\medskip
 $\bullet$\;  We fix a $\mathcal D^{(\ell_+)}$-measurable
 region $\La$ and $t\in[0,1]$
 and omitting the dependence on
 $\La$ and $t$, we rewrite the functional \eqref{I.2.1.1} as follows. Notice
 that there are two equal entropy terms:  one  in $F^*$ multiplied by $t$  and the
 other,
 explicitly written on the r.h.s. of \eqref{appK.1}, multiplied by
 $1-t$. The same holds for the terms multiplied by
 $\la_\beta$. Calling  $r^{(k)}$ the vector with components
 $r^{(k)}_s$ as in \eqref{aa2.47}, we have
                \begin{equation}
     \label{6.7}
f (\rho_\La;\bar \rho_{\La^c})=  \frac 12 (\rho_\La, V^*_{\ga,t}
\,\rho_\La) + (\rho_\La, \hat V_{\ga,t}\,\bar
\rho_{\La^c})+(1-t)(1_\La,r^{(k)} \rho_\La) -\frac{1}{\beta} (1_\La,
\mathcal I(\rho_\La))
     \end{equation}
where $(u,v)=(u,v)_{\ga^{-1/2}}$,
   \begin{eqnarray}
      \label{6.8}
&&V_{\ga,t}^*(x,s,x',s')=\text{\bf 1}_{s\ne s'} \hat
V_{\ga,t}(x,x'),\qquad x,x'\in \ga^{-1/2}\mathbb Z^d
\\&& \hat
V_{\ga,t}(x,x')=t\,\ga^{-d/2} \sum_{y\in \ga^{-1/2}\mathbb Z^d} \hat
J_\ga(x,y) \ga^{-d/2} \hat J_\ga(y,x')
    \nn
     \end{eqnarray}
$\hat J_\ga$ as in \eqref{2.8} and
   \begin{equation}
      \label{6.9}
   \mathcal I(\rho_\La)(x,s)=
   -\rho_\La(x,s)\Big(\log \rho_\La(x,s)-1-\beta \la_\beta\Big)
     \end{equation}

\medskip

 $\bullet$\; We call $ \bar J_\ga$  the kernel obtained
by averaging over the cubes of $\mathcal D^{(\ell_-)}$ while $\hat
J_\ga$ is over those of $\mathcal D^{(\ga^{-1/2})}$, thus
            \begin{equation}
      \label{A.5.00}
 \bar J_\ga(z_x,z_y)= \mintone{C^{(\ell_-)}_{z_x}}\mintone {C^{(\ell_-)}_{z_y}}
J_\ga(r,r')\; dr\,dr',\qquad x,y\in \ga^{-1/2}\mathbb Z^d
     \end{equation}
We also define
   \begin{equation}
     \label{appK.1.0.6}
V^{(\ell_-)}_{\ga,t}(z_x,z_y) =  t\int \bar J(z_x,r)\bar J(r ,
z_y)dr
     \end{equation}
Observe that  there is $c_1>0$ such that
        \begin{equation}
     \label{appK.1.0.7}
|V^{(\ell_-)}_{\ga,t}(z_x,z_y) -\hat V_{\ga,t}(x,y)| \le
c_1(\ga^{-d/2}\ga^d)(\ga\ell_-)\text{\bf 1}_{|x-y|\le 2 \ga^{-1}}
     \end{equation}

\medskip
 $\bullet$\;  We
introduce a $\eps$-relaxed constraint
   \begin{eqnarray*}
      \label{appK.3}
&&\hskip-1cm f_\eps =f +\frac {\eps^{-1}}4 \sum_s\sum_{x\in
 \La\cap\ga^{-1/2}\mathbb Z^d}\Big(
\{({\rm Av}(\rho_{\La};x,s)-[\rho^{(k)}_{s}+\zeta])_+\}^4
+\{({\rm Av}(\rho_{\La};x,s)- [\rho^{(k)}_{s}-\zeta])_-\}^4
\Big)
     \end{eqnarray*}
 where $(a)_+=a \,\text{\bf 1}_{a>0}$,
$(a)_-=a\,\text{\bf 1}_{a<0}$  and $\rho_\La\in
\mathbf{W}^{(k)}_\La$,
    \begin{equation*}
 \mathbf W^{(k)}_\La:=\Big\{\rho_\La: |{\rm
  Av}(\rho_{\La};x,s)
  -\rho^{(k)}_{s}|\le b, x\in \La\cap\ga^{-1/2}\mathbb Z^d,s\in\{1,.,S\} \Big\},
  \qquad b=\frac 12\min_{k_1\ne
  k_2}\|\rho^{(k_1)}-\rho^{(k_2)}\|_\infty
      \end{equation*}
$b$ is such that $ \mathbf
W^{(k)}_\La\cap\{\rho^{(1)},\dots,\rho^{(S+1)}\}=\rho^{(k)}$ and
$\mathbf W^{(k_1)}_\La\cap \mathbf W^{(k_2)}_\La=\emptyset$ if
$k_1\ne k_2$.

\smallskip

  $\bullet$\;  We
denote by  $\hat \rho_{\La,\eps}$ and  $\hat \rho_{\La}$ the
minimizers of $f_\eps$ in $ \mathbf{W}^{(k)}_\La$ and of $f$ in
$\mathbf{W}^{(k)}_\La$.

\smallskip
$\bullet$\; For any  $\mathcal D^{(\ga^{-1/2})}$-measurable function
$\psi(\rho)$, we denote by
     \begin{equation}
      \label{ee5.1.3}
D_\La\psi\;=\;\Big\{\;\frac{\partial\psi}{\partial
\rho(x,s)},\;
   x\in \ga^{-1/2}\mathbb Z^d \cap \La,s\in\{1,.,S\}\Big\}
     \end{equation}
We say that a $\mathcal D^{(\ga^{-1/2})}$-measurable function
$\rho_\La$ is a critical point of $f_\eps$, respectively $f$, if
$D_\La f_\eps(\rho_\La)=0$, respectively $D_\La f(\rho_\La)=0$.

\vskip2cm


\subsection{Point-wise estimates}

       \label{subsec:appK.1}


In this subsection we prove some a-priori bounds on the fluctuations
of the minimizers $\hat \rho_{\La,\eps}$ and $\hat \rho_{\La}$ from
their averages.

We start from the latter proving the following result.

\bigskip

 \begin{prop}
 \label{prop6.1}  Let $\La$ be a
$\mathcal D^{(\ell_+)}$-measurable set  and let $\bar\rho_{\La^c}\in
\mathcal X^{(k)}_{\La^c}$.  There is $c^*$ such that for all $\ga$
small enough the following holds.

\noindent {\bf (i)}. Let $C\subset \La $ any $\mathcal D^{(\ell_-)}$
cube. Let $\rho_{\La\setminus C}\in \mathcal{X}^{(k)}_{\La\setminus
C}$ and denote by $\bar \rho_{C^c}=(\rho_{\La\setminus
C},\bar\rho_{\La^c})$. If $\hat \rho_C$ is a minimizer of $f(\cdot;
 \bar\rho_{C^c})$ in $\mathcal{X}^{(k)}_C$
then $|\hat \rho_C(x,s)-\rho^{(k)}_{s}|\le c^*\zeta$ for all $x\in
C\cap \ga^{-1/2}\mathbb{Z}^d$.

\noindent {\bf (ii)}. If $\hat \rho_{\La}\in \mathcal{X}^{(k)}_\La$
 minimizes $f(\cdot:\bar \rho_{\La^c}) $
then $|\hat \rho_\La(x,s)-\rho^{(k)}_{s}|\le c^*\zeta$ for all $x\in
\La\cap \ga^{-1/2}\mathbb{Z}^d$.

  \end{prop}

\medskip We postpone the proof of Proposition \ref{prop6.1} giving
first some preliminary lemmas.

\bigskip

For any $\mu=(\mu_1,..,\mu_S) \in [-1,1]^S$ and any $\tau\in [0,1]$
for $C$ and $\bar \rho_{C^c}$ as in (i) of Proposition
\ref{prop6.1}, we let
  \begin{equation}
     \label{appK.1.0.0}
g_\mu(\rho_C;\rho_{ C^c},\tau):=  \frac \tau 2 (\rho_C,
 V^*_{\ga,t} \,\rho_C) + (\rho_C,  V^*_{\ga,t}\,
\rho_{C^c})+(1-t)(1_C,r^{(k)} \rho_C) -\frac{1}{\beta}
(1_C, \mathcal I(\rho_C)) - (\rho_C,\mu)
  \end{equation}
where
    $$
(\rho_C,\mu):= \sum_s\sum_{x\in
C\cap\ga^{-1/2}\mathbb{Z}^d}
 \mu_{s} \rho_C(x,s)
    $$
We regard $g_\mu$ as a function of $\rho_C=\{\rho_C(x,s)\ge 0, x\in
C\cap\ga^{-1/2}\mathbb{Z}^d,s\in\{1,\dots,S\}\}$. $\mu_s$ has the
interpretation of a chemical potential for the species $s$, $\tau$
is an auxiliary parameter, we will eventually set $\tau=1$, in which
case $g_\mu(\rho_C;\rho_{C^c},1):= f(\rho_C,\rho_{C^c})
-(\rho_C,\mu)$.

\vskip1cm

We let
    \begin{equation}
     \label{6.11}
\psi(x,s):=\sum_{s'\ne s}\sum_{x'\in
C^c\cap\ga^{-1/2}\mathbb{Z}^d}\hat V_{\ga,t}(x,x')\bar\rho_{
C^c}(x',s'), \quad x\in C\cap\ga^{-1/2}\mathbb{Z}^d
    \end{equation}
and define the map $T_\mu$ on $\{\rho_C:\rho_C(x,s)\ge 0\}$
by setting for $x\in C\cap\ga^{-1/2}\mathbb{Z}^d$
    \begin{equation}
     \label{6.12}
  (T_\mu\rho_C)(x,s) = \exp \Big\{-\beta \Big(\tau\sum_{s'\ne s}\sum_{x'\in
C^c\cap\ga^{-1/2}\mathbb{Z}^d}  \hat V_{\ga,t}(x,x') \rho_C(x',s')
  + \psi(x,s) -\la_\beta
  -\mu_{s}+(1-t)r^{(k)}_{s}\Big)\Big\}
     \end{equation}

\vskip1cm

            \begin{lemma}
 \label{lemma6.3}
There is $\ga_0$ such that for all $\ga\le \ga_0$ and all $\mu\in
[-1,1]^S$, $g_\mu$ has a unique minimizer $\hat
\rho_C(\cdot;\mu,\tau)$. Furthermore for all $\mu$ and $\tau$,
    \begin{equation}
     \label{appK.1.0.1}
\hat \rho_C(x,s;\mu,\tau) \le e^{\beta (\la_\beta+\mu_{s})}, \qquad
x\in C\cap\ga^{-1/2}\mathbb{Z}^d, s\in\{1,\dots,S\}
  \end{equation}

 \end{lemma}

\vskip.5cm

{\bf Proof.} Since $\{g_\mu \le n\}$ are compact for any $n>0$ and
$g_\mu$ is smooth, then $g_\mu$ has a minimum. Any minimizer (which
is strictly positive because of the entropy term in $g_\mu$) is also
a critical point, namely a fixed point of the map $T_\mu$, such that
recalling \eqref{6.12},
   \begin{equation}
     \label{appK.1.0.2}
  \rho_C(x,s) = \exp \Big\{-\beta \Big(\tau\sum_{s'\ne s}\sum_{x'\in
C^c\cap\ga^{-1/2}\mathbb{Z}^d}  \hat V_{\ga,t}(x,x') \rho_C(x',s')
  + \psi(x,s) -\la_\beta
  -\mu_{s}+(1-t)r^{(k)}_{s}\Big)\Big\}
     \end{equation}
Since  $\rho_C(x,s)> 0$,
 $\psi(x,s)> 0$, $t\le 1$  and $r^{(k)}_{s}>0$,
 then \eqref{appK.1.0.1} holds for any minimizer.
Thus any minimizer belongs to the set $O=\{\rho_C:\rho_C(x,s)\in
[0,e^{\beta (\la_\beta+1)}]\}$ and $O$ is invariant under $T_\mu$.
Moreover, for any $x\in C\cap\ga^{-1/2}\mathbb{Z}^d$,
   \begin{equation*}
  \hskip-1cm \frac{\partial (T_\mu\rho_C)(x,s)}{\partial
  \rho_C(x',s')} = -\beta \tau(T_\mu\rho_C)(x,s)  V^*_{\ga,t}(x,s,x',s'),\quad
    | \frac{\partial (T_\mu\rho_C)(x,s)}{\partial
  \rho_C(x',s')}| \le \beta  e^{\beta (\la_\beta+1)} c\ga^d\ga^{-d/2}
     \end{equation*}
such that
        \begin{equation*}
 | (T_\mu\rho''_C)(x,s)- (T_\mu\rho'_C)(x,s)|
  \le \Big(\beta e^{\beta (\la_\beta+1)} c\ga^{d/2}\Big)
   \max_{s,x\in C\cap\ga^{-1/2}\mathbb{Z}^d} | \rho''_C(x,s)-  \rho'_C(x,s)|
     \end{equation*}
For $\ga$ small enough, $\beta e^{\beta (\la_\beta+1)} c\ga^{d/2}<1$
such that $T_\mu$ is a contraction and the fixed point is  unique.
\qed

\vskip1cm

 \begin{lemma}
 \label{lemmaappK.1.3}
There is $c>0$ such that for all $\ga$ small enough   and with $\hat
\rho_C(\cdot;\mu,\tau)$ as in Lemma \ref{lemma6.3},
  \begin{equation}
     \label{6.15}
 e^{- c \zeta }\le
 \frac{\hat \rho_C(x,s;\mu,\tau)}{ \rho^{(k)}_{s}e^{ \beta \mu_{s}}}
 \le  e^{ c \zeta}
    \end{equation}

 \end{lemma}

\vskip.5cm

{\bf Proof.} Recalling the definition \eqref{appK.1.0.6}
 we use the estimate \eqref{appK.1.0.7} to
replace  $\hat V_{\ga,t}$ with $V^{(\ell_-)}_{\ga,t}$.

Thus, recalling the definition of $\psi$ in \eqref {6.11}, since
$\bar\rho_{C^c}\in \mathcal{X}^{(k)}_{C^c}$ there are $c_2$ and
$c_3$ positive such that (below we shorthand $y\in C$ for $y\in
C\cap \ga^{-1/2}\mathbb{Z}^d$)
   \begin{equation*}
\sum_{s'\ne s}\sum_{y\in C} \tau \hat V_{\ga,t}(x,y)
\hat\rho_C(y,s';\mu,\tau)+\psi(x,s)
 \ge \psi(x,s)\ge t\sum_{s'\ne s}[\rho^{(k)}_{s'} -\zeta]
- c_2(\ga\ell_-)= t \,r^{(k)}_{s}-(S-1)\zeta - c_2(\ga\ell_-)
     \end{equation*}
and by \eqref{appK.1.0.1},
        \begin{equation*}
\sum_{s'\ne s}\sum_{y\in C} \tau \hat V_{\ga,t}(x,y)
\hat\rho_C(y,s';\mu,\tau)+\psi(x,s) \le t [ \,r^{(k)}_{s}+(S-1)\zeta
]+e^{\beta (\la_\beta+1)}c_3(\ga^d\ell_-^d)+ c_2(\ga\ell_-)
     \end{equation*}
Recalling that $\dis{\rho^{(k)}_s= \exp\{-\beta [r^{(k)}_s
-\la_\beta]\}}$, we then obtain \eqref{6.15} from \eqref{appK.1.0.2}
and \eqref{2.3} for $\ga$ small enough.  \qed

\vskip1cm

In the next Lemma we prove that the minimizer $\hat
\rho_C(\cdot;\mu,\tau)$ is  smooth.

\vskip1cm

            \begin{lemma}
 \label{lemmaappK.1.4}
$\hat \rho_C(\cdot;\mu,\tau)$ is a smooth function of $\mu\in
[-1,1]^S$ and $\tau\in [0,1]$ (derivatives of all orders exist) and
there is a constant $c$ such that
                \begin{equation}
     \label{appK.1.0.2.2}
 \frac{\partial \hat \rho_C(x,s;\mu,\tau)}{\partial \mu_s}  \ge \beta e^{\beta\mu_s}\rho^{(k)}_s -
 c\zeta,\;\;\; | \frac{\partial \hat \rho_C(x,s;\mu,\tau)}{\partial \mu_{s'}} |
  \le c (\ga\ell_-)^d,\;\;s'\ne s
      \end{equation}
    \begin{equation}
     \label{appK.1.0.2.2.1}
|\frac{\partial \hat \rho_C(x,s;\mu,\tau)}{\partial \tau} | \le c
(\ga\ell_-)^d
    \end{equation}

 \end{lemma}
\vskip.5cm

{\bf Proof.} The   minimizer $\hat \rho_C(x,s;\mu)$ is implicitly
defined by the critical point equation $\partial g_\mu/\partial
\rho_C(x,s)=0$, see \eqref{appK.1.0.2}. Thus we call
        $$
G(\rho,\mu):= \rho- T_\mu\rho
    $$
and we observe that
    $$
\frac{\partial G(\rho,\mu)(x',s')}{\partial
\rho(x,s)}\Big|_{\rho=\hat\rho_C(\cdot,\mu,\tau)}:=
A(x,s,x',s'),\qquad x,x'\in C\cap\ga^{-1/2}\mathbb{Z}^d
    $$
where
    \begin{equation}
     \label{6.20}
 A(x,s,x',s')= \text{\bf 1}_{(x,s)=(x',s')}+ \beta \hat \rho_C(x,s;\mu,\tau)
 \tau  V^*_{\ga,t}(x,s,x',s')
     \end{equation}
 By \eqref{appK.1.0.1} this is a positive definite matrix and  by
 the definition of $\hat V_{\ga,t^*}$
    \begin{equation}
     \label{6.21}
| A^{-1}(x,s,x',s')- \text{\bf 1}_{(x,s)=(x',s')}|\le c
(\ga\ell_-)^d,\qquad x,x'\in C\cap\ga^{-1/2}\mathbb{Z}^d
     \end{equation}
By the implicit function theorem  we then conclude that the function
$\hat\rho_C(\cdot,\mu,\tau)$, such that $G(\hat\rho_C,\mu)=0$ is
differentiable and its derivative verifies
   \begin{equation*}
\sum_{x'\in C\cap\ga^{-1/2}\mathbb{Z}^d} A(x,s,x',s')\frac{\partial
\hat \rho_C(x',s';\mu,\tau)}{\partial \mu_{s}} = \beta \hat
\rho_C(x,s;\mu)
     \end{equation*}
Thus from \eqref{6.15} and \eqref{6.21} we get \eqref{appK.1.0.2.2}.

Same argument applies for the derivative $\partial \hat
\rho_C(x,s;\mu,\tau)/\partial \tau$. \qed

\vskip1cm

  We next define
        \begin{equation}
     \label{appK.1.0.7.0}
R_s(\mu,\tau):= (\frac{\ga^{-1/2}}{\ell_-})^d \sum_{x\in
\ga^{-1/2}\mathbb Z^d \cap C} \hat \rho_C(x,s;\mu,\tau)
     \end{equation}
Since $C\in \mathcal D^{(\ell_-)}$, the number of cubes in $\mathcal
D^{(\ga^{-1/2})}$ contained in $C$ are
$\dis{(\frac{\ell_-}{\ga^{-1/2}})^d}$, thus  $R_s(\mu,\tau)$ is the
average density of the species $s$ when the chemical potential is
$\mu$. Our purpose is to prove that for any $b=(b_1,..b_S), b_s\in
[\rho^{(k)}_s-\zeta,\rho^{(k)}_s+\zeta], s=1,..,S$, there is $\mu
\in [-1,1]^S$ such that
        \begin{equation}
     \label{appK.1.0.7.1}
R_s(\mu,1) = b_s,\, s=1,..,S
     \end{equation}

     \vskip1cm

            \begin{lemma}
            \label{lemmaappK.1.4.1}
 For any $b_s\in
[\rho^{(k)}_s-\zeta,\rho^{(k)}_s+\zeta]$, $s\in \{1,..,S\}$, the
equation \eqref{appK.1.0.7.1}  has a solution $\mu\in [-1,1]^S$ and
there is a constant $c$ such that for all $s$, $|\mu_s|\le c \zeta$.

             \end{lemma}

 {\bf
Proof.} We fix a vector $b_s\in
[\rho^{(k)}_s-\zeta,\rho^{(k)}_s+\zeta]$, $s\in \{1,..,S\}$, and we
first observe that the equation $R_s(\mu,0) = b_s$ has obviously a
solution $\mu^*(0)$, which, recalling \eqref{appK.1.0.2}, is
obtained by solving
   \begin{equation*}
(\frac{\ga^{-1/2}}{\ell_-})^d \sum_{x\in \ga^{-1/2}\mathbb Z^d \cap
C}
 e^{-\beta (
   \psi(x,s) -\la_\beta
  -\mu^*_{s}(0))-(1-t)r^{(k)}_s}= b_s
     \end{equation*}
and by the same arguments used in the proof of Lemma
\ref{lemmaappK.1.3}, $|\mu^*(0)|\le c\zeta$.

We then for look for a function $\mu(\tau)$, $\tau\in[0,1]$  such
that
        \begin{equation}
     \label{appK.1.0.7.2}
R_s(\mu(\tau),\tau) = b_s,\, s=1,..,S
     \end{equation}
By differentiating the above equation  we get the following Cauchy
problem for $\mu(\tau)$:
        \begin{equation}
     \label{appK.1.0.7.3}
\sum_{s'}\frac{ \partial R_s(\mu(\tau),\tau)}{\partial \mu_{s'}}\;
\frac{ d\mu_{s'}(\tau)}{d\tau} = - \frac{ \partial
R_s(\mu(\tau),\tau)}{\partial \tau}, \quad \mu(0)=\mu^*(0)
     \end{equation}
By \eqref{appK.1.0.2.2} the $S\times S$ matrix
 $\{\partial
R_s(\mu(\tau),\tau)/\partial \mu_{s'}, s,s'=1,\dots,S\}$ has
diagonal elements
        \begin{equation}
     \label{appK.1.0.7.3diag}
 \frac{ \partial R_s(\mu,\tau)}{\partial \mu_{s}}\ge \beta e^{-\beta}\rho^{(k)}_s -
 c\zeta
     \end{equation}
     while the non diagonal elements are bounded by
             \begin{equation}
     \label{appK.1.0.7.4}
| \frac{ \partial R_s(\mu,\tau)}{\partial \mu_{s'}}|\le
 c (\ga\ell_-)^d
     \end{equation}
Thus (for $\ga$ small)  the matrix $\{\partial
R_s(\mu(\tau),\tau)/\partial \mu_{s'}, s,s'=1,\dots,S\}$ is positive
definite and invertible and depends smoothly on $\mu$.  As a
consequence the Cauchy problem  \eqref{appK.1.0.7.3} has a unique
solution and since by \eqref{appK.1.0.2.2.1}
    \begin{equation}
     \label{appK.1.0.7.5}
|\frac{ \partial R_s(\mu,\tau)}{\partial \tau} | \le c
(\ga\ell_{-,\ga})^d
    \end{equation}
the solution $\mu(\tau)$ satisfies the bound $|\mu_s(\tau)|\le
c\zeta$ for all $s$.  \qed

 \vskip1cm
We now have all the ingredients for the proof of the
Proposition \ref{prop6.1} stated at the beginning of this
Subsection.

\medskip

{\bf Proof  of Proposition \ref{prop6.1}}.

\noindent {\it Proof of (i)}. Let $\hat \rho_C\in
\mathcal{X}^{(k)}_C$ be a minimizer of $f(\cdot;\bar
 \rho_{C^c})$  and call
   \begin{equation}
   \label{6.31}
b_s= \big(\frac{\ga^{-1/2}}{\ell_-}\big)^d \sum_{x\in
\ga^{-1/2}\mathbb Z^d \cap C} \hat \rho_C(x,s)
    \end{equation}
Since $\hat \rho_C\in \mathcal{X}^{(k)}_C$, (see \eqref{appK.4}),
$b_s$ is as in Lemma \ref{lemmaappK.1.4.1}  and therefore, for $\ga$
small enough, there is
 $\mu=(\mu_s, s=1,..,S)$ such that
   \begin{equation}
   \label{6.32}
   \big(\frac{\ga^{-1/2}}{\ell_-}\big)^d \sum_{x\in
\ga^{-1/2}\mathbb Z^d \cap C}\hat \rho_C(x,s;\mu,1)=b_s,\quad
s=1,..,S
    \end{equation}
Then writing $f(\rho_C)$ for $f(\rho_C;\bar
 \rho_{C^c})$ and using \eqref{6.31},
   $$
f(\hat\rho_C)= g_\mu(\hat\rho_C;1)+
\frac{\ell_-^d}{\ga^{-d/2}}\sum_{s=1}^S\mu_sb_s
   $$
and using \eqref{6.32}
   $$
   f(\hat\rho_C) \le  f(\hat \rho_C(\cdot;\mu,1))= g_\mu(\hat
   \rho_C(\cdot;\mu,1);1)+
\frac{\ell_-^d}{\ga^{-d/2}}\sum_{ s=1}^S b_s \mu_s
   $$
Hence  $g_\mu(\hat\rho_C) \le g_\mu(\hat \rho_C(\cdot;\mu,1))$ and
since by Lemma \ref{lemma6.3} the minimizer is unique we get that
$\hat\rho_C=\hat \rho_C(\cdot;\mu,1)$. On the other hand by
\eqref{6.15},  there is $c^*$ such that for all $\ga$  small enough,
$|\hat \rho_C(x,s;\mu,1)-\rho^{(k)}_{s(i)}|\le c^*\zeta$ for all
$x\in C\cap\ga^{-1/2}\mathbb{Z}^d$.

{\it Proof of (ii).}  Let $\hat \rho_{\La}\in \mathcal{X}^{(k)}_\La$
be a minimizer of $f(\cdot;\bar \rho_{\La^c})$ and let $C\subset\La$
be a $\mathcal{D}^{(\ell_-)}$ cube. We write $\hat \rho_{\La}=(\hat
\rho_{C}, \hat \rho_{\La\setminus C})$ and we prove that $\hat
\rho_{C}\in \mathcal{X}^{(k)}_C$ minimizes $f(\cdot; \hat
\rho_{\La\setminus C})$. In fact
        $$
f(\hat \rho_{\La};\bar \rho_{\La^c})=f(\hat \rho_{\La\setminus
C};\bar \rho_{\La^c})+f(\hat \rho_{C}; \hat \rho_{\La\setminus C})
    $$
and if $\hat \rho_{C}$ were not a minimizer than for any minimizer
$\bar\rho^*_C$, calling $\hat \rho^*_{\La}=(\hat \rho_{\La\setminus
C},\hat \rho^*_C)$, we would have
    $$
f(\hat \rho_{\La};\bar \rho_{\La^c})>f(\hat \rho_{\La\setminus
C};\bar \rho_{\La^c})+f(\hat \rho^*_{C}; \hat \rho_{\La\setminus C})
   >f(\hat \rho^*_{\La};\bar \rho_{\La^c}) $$
and this would be a contradiction. Thus $\hat \rho_{C}$ is a
minimizer and by  (i) $|\hat \rho_\La(x,s)-\rho^{(k)}_{s}|\le
c^*\zeta$, $x\in C\cap\ga^{-1/2}\mathbb{Z}^d$.  The proof of (ii)
then follows from the arbitrariness of $C\subset \La$.  \qed

\vskip1cm

We will next consider $\eps>0$ and start by  proving the analogue of
Lemma 5.2 of \cite{DMPV2}:

\vskip1cm

    \begin{lemma}
     \label{lemmaappK.1.5}

There is a constant $c>0$ such that for all $\eps>0$ small enough
any minimizer $\hat \rho_{\La,\eps}\in \mathbf{W}^{(k)}_\La$ of
$f_\eps$ is also a critical point and  for all $x\in \La$, and all
$s=1,\dots,S$
    \begin{eqnarray}
     \label{appK.1.1}
&& \big|{\rm Av}(\hat\rho_{\La,\eps};x,s)-\rho^{(k)}_{s}\big| \le
\zeta +c \{(\frac{\ell_+}{\ga^{-1/2}})^{d/4}\log
\ga^{-1}\}\eps^{1/4}
     \end{eqnarray}

       \end{lemma}

\vskip.5cm

{\bf Proof.} We denote by
    $$
\psi(\rho_\La)=\sum_{i\in \La} \{({\rm
Av}(\rho_\La;x,s)-[\rho^{(k)}_{s}+\zeta])_+\}^4 +\{({\rm
Av}(\rho_\La;x,s)- [\rho^{(k)}_{s}-\zeta])_-\}^4
    $$
Then for all $\rho_{\La}\in \mathbf{W}^{(k)}_\La$,
   $$
\frac 1{4\eps}\psi(\hat \rho_{\La,\eps})\le f(\rho_\La;\bar
\rho_{\La^c})-f(\hat \rho_{\La,\eps};\bar \rho_{\La^c})+\frac
1{4\eps}\psi(\rho_\La)
    $$
and since $\psi$ vanishes on $\mathcal{X}^{(k)}_\La$:
   $$
\frac 1{4\eps}\psi(\hat \rho_{\La,\eps})\le \inf_{\rho_{\La}\in
\mathcal{X}^{(k)}_\La}  f(\rho_\La;\bar \rho_{\La^c}) - f(\hat
\rho_{\La,\eps};\bar \rho_{\La^c})
    $$
and, calling $\dis{\phi'=\min_{\rho_{\La}\in
\mathcal{X}^{(k)}_\La}f(\rho_\La;\bar \rho_{\La^c})}$,\,\,
$\dis{\phi''=\min_{\rho_{\La}\in W^{(k)}_\La}f(\rho_\La;\bar
\rho_{\La^c})}$ we have
    \begin{equation}
    \label{6.38}
\{({\rm Av}(\hat\rho_{\La,\eps};x,s)-[\rho^{(k)}_{s}+\zeta])_+\}^4
+\{({\rm Av}(\hat\rho_{\La,\eps};x,s)- [\rho^{(k)}_{s}-\zeta])_-\}^4
\le 4\eps (\phi'-\phi'')
     \end{equation}
Let
    \begin{equation}
    \label{6.34}
R= \max_{s}(\rho^{(k)}_s+\zeta),\qquad c'>0 {\text { such that
}}|\ga^{-1/2}\mathbb Z^d\cap
    \La|\le c'\frac{\ell_+^d}{\ga^{-d/2}}
        \end{equation}
Given $x,y\in$ $\ga^{-1/2}\mathbb Z^d$ we denote by $z_x\in
\ell_-\mathbb Z^d$ the point such that $x\in C^{(\ell_-)}_{z_x}$ and
analogously  $z_y\in \ell_-\mathbb Z^d$ is such that $y\in
C^{(\ell_-)}_{z_y}$. Then for $V^{(\ell_-)}_{\ga,t}$  as in
\eqref{appK.1.0.6}, there is a constant $c$  such that
    \begin{equation}
    \label{6.35}
\hat V_{\ga,t}(x,y)\le c V^{(\ell_-)}_{\ga,t}(z_x,z_y)
    \end{equation}
Observe that
    $$
\sum_s\sum_{x\in\La\cap\ga^{-1/2}\mathbb{Z}^d}\rho_\La(x,s)\log
\rho_\La(x,s)=
\sum_s\sum_{x\in\La\cap\ga^{-1/2}\mathbb{Z}^d}V^{(\ell_-)}_\ga\star\big[\rho_\La\log
\rho_\La](z_x)
    $$
Thus by Jensen inequality,
    $$
\sum_s\sum_{x\in\La\cap\ga^{-1/2}\mathbb{Z}^d}\rho_\La(x,s)\log
\rho_\La(x,s)\le \sum_s\sum_{x\in\La\cap\ga^{-1/2}\mathbb{Z}^d}
V^{(\ell_-)}_{\ga,t}\star\rho_\La\log[V^{(\ell_-)}_{\ga,t}\star\rho_\La](z_x)
    $$
 Then, from \eqref{6.34} we get that for all $\rho_{\La}\in
\mathcal{X}^{(k)}$,
    \begin{equation}
    \label{6.37}
\sum_s\sum_{x\in\La\cap\ga^{-1/2}\mathbb{Z}^d}\rho_\La(x,s)\log
\rho_\La(x,s)\le\frac{Sc'\ell_+^d}{\beta\ell_-^d}
  \frac{R \ell_-^d}{\ga^{-d/2}} \log \frac{R \ell_-^d}{\ga^{-d/2}}
    \end{equation}
By using \eqref{6.35} we then get
    \begin{eqnarray*}
&& \phi' \le c (\rho_\La, V^{(\ell_-)}_{\ga,t} \rho_\La) +c
(\rho_\La, V^{(\ell_-)}_{\ga,t}\bar \rho_{\La^c})+
\frac{Sc'\ell_+^d}{\beta\ell_-^d}
  \frac{R \ell_-^d}{\ga^{-d/2}} \log \frac{R \ell_-^d}{\ga^{-d/2}}
+(1-t^*+1+\la_\beta) RS \frac{c'\ell_+^d}{\ga^{-d/2}}
\\&&\hskip1cm \le  cc'  \frac{\ell_+^d}{\ga^{-d/2}} SR^2+C
  \frac{\ell_+^d}{\ga^{-d/2}} \log \ga^{-1}
     \end{eqnarray*}
An upper bound for $\phi''$ is  $\dis{ \phi'' \ge
   \frac{c'\ell_{+,\ga}^d}{\ell_{0}^d}S \frac{e^{\beta \la_\beta}}{ \beta}
}$. Thus there is $c$ such that
    \begin{equation}
    \label{6.39}
\phi'-\phi''\le c  \frac{\ell_+^d}{\ga^{-d/2}} \log \ga^{-1}
    \end{equation}
From \eqref{6.38} and \eqref{6.39}, \eqref{appK.1.1} follows.

From  \eqref{appK.1.1} and by choosing $\eps$ so small that
$\zeta +c \{(\frac{\ell_+}{\ga^{-1/2}})^{d/4}\log
\ga^{-1}\}\eps^{1/4}<2\zeta<b$,  we conclude that the
minimizer is in the interior of  $\mathbf{W}^{(k)}_\La$ and
thus is a critical point. \qed

\vskip1cm

 \begin{prop}
 \label{propappK.1.3}
There is $c^*$ such that for any $\ga>0$ small enough, for all
$\eps>0$ small enough (depending on $\ga$), if $\hat
\rho_{\La,\eps}$ minimizes $f_\eps$ in $\mathbf{W}^{(k)}_\La$ then
    \begin{equation}
    \label{a6.40}
    |\hat
\rho_{\La,\eps}(x,s)-\rho^{(k)}_{s}|\le c^*\zeta,\qquad {\text{ for
all }}x\in \La, s=1,\dots,S
     \end{equation}
  \end{prop}

\vskip.5cm

{\bf Proof.}  By Lemma \ref{lemmaappK.1.5} given any $\ga$ if $\eps$
is small enough, then $\hat \rho_{\La,\eps}$ is a critical point,
namely it satisfy the following equation.
     \begin{eqnarray}
    \nn
&&\hskip-1cm\log\hat \rho_{\La,\eps}(x,s)= -\beta
\Big(\tau\sum_{s'\ne s}\sum_{y\in \La\cap\ga^{-1/2}\mathbb{Z}^d}
\hat V_{\ga,t}(x,y) \hat \rho_{\La,\eps}(y,s')
  + \psi(x,s) -\la_\beta
  -\mu_{s}+(1-t)r^{(k)}_{s}\Big)
  \\&&\hskip1cm+ \frac 1{\eps} \{({\rm Av}(\hat\rho_{\La,\eps};x,s)-
  [\rho^{(k)}_{s}+\zeta])_+\}^3
+\{({\rm Av}(\hat\rho_{\La,\eps};x,s)- [\rho^{(k)}_{s}-\zeta])_-\}^3
    \label{6.40}
     \end{eqnarray}
where
    $$\psi(x,s)=\sum_{s'\ne s}\sum_{y\in \La\cap\ga^{-1/2}
    \mathbb{Z}^d}\hat
V_{\ga,t}(x,y)\bar\rho_{ \La^c}(y,s')$$

Let $C\subset \La$ a $\mathcal{D}^{(\ell_-)}$ cube and let $x,y\in
C$. Then ${\rm Av}(\hat\rho_{\La,\eps};x)={\rm
Av}(\hat\rho_{\La,\eps};y)$ and from \eqref{6.40} we get
   \begin{equation*}
|\log \hat \rho_{\La,\eps}(x,s)-\log \hat \rho_{\La,\eps}(y,s)| \le
\sum_{s''\ne s} \sum_{z\in \La\cap\ga^{-1/2}
    \mathbb{Z}^d} |\hat V_{\ga,t}(x,z) -\hat
V_{\ga,t}(y,z)|\hat \rho_{\La,\eps}(z,s'')
     \end{equation*}
We use \eqref{appK.1.0.7}  and  we get
    \begin{equation*}
|\log \hat \rho_{\La,\eps}(x,s)-\log \hat \rho_{\La,\eps}(y,s)| \le
c (\ga\ell_-)=c \ga^{\alpha_-}
     \end{equation*}
Thus for all $x\in \La \cap\ga^{-1/2}
    \mathbb{Z}^d$
    \begin{equation}
| \rho_{\La,\eps}(x,s)-\rho^{(k)}_{s}|\le |
\rho_{\La,\eps}(x,s)-{\rm Av}(\hat\rho_{\La,\eps};x,s)|+|{\rm
Av}(\hat\rho_{\La,\eps};x,s)-\rho^{(k)}_{s}\big| \le c'
\ga^{\alpha_-}+ 2\zeta
    \label{6.41}
     \end{equation}
where we have used   \eqref{appK.1.1} and the fact that for all
$\eps$ small enough, $c \{(\frac{\ell_+}{\ga^{-1/2}})^{d/4}\log
\ga^{-1}\}\eps^{1/4} <\zeta$. \qed

\vskip1cm

    \begin{lemma}
     \label{lemmaappK.1.5bis}

$\hat \rho_{\La,\eps}$ converges by subsequences and any limit point
$\hat \rho_{\La}$ is a minimizer of $f$ in $\mathcal{X}^{(k)}_\La$.

   \end{lemma}

\vskip.5cm

{\bf Proof.} The proof is exactly the same as that of Lemma 5.3 in
\cite{DMPV2}. Compactness implies convergence by subsequences and by
\eqref{a6.40} any limiting point is in $ \mathcal{X}^{(k)}_\La$.
Since for any $\rho_\La\in  \mathcal{X}^{(k)}_\La$
$f(\rho_\La)=f_\eps(\rho_\La)\ge f_\eps(\hat\rho_{\La,\eps})$, by
taking the limit $\eps\to 0$ along a  subsequence converging to some
$\hat\rho_\La$, we get $f(\rho_\La)\ge f(\hat\rho_\La)$, thus
$\hat\rho_\La$ is a minimizer. \qed

\vskip2cm


\subsection{Convexity and uniqueness}

       \label{subsec:appK.2}


In this subsection we prove convexity of $f$ and $f_\eps$ and from
this we will deduce the uniqueness of their minimizers.

\vskip.5cm

\begin{thm}
 \label{thmappK.2.1}

Given any $b\ge 2$ and $\kappa \in (0,\kappa^*)$ ($\kappa^*$ as in
\eqref{2.4}), for all $\ga$ small enough the following holds. Let
$\rho_\La \in \mathcal{X}^{(k)}_\La$ and
$|\rho_\La(x,s)-\rho^{(k)}_{s}| <b\zeta$ for all $x\in
\La\cap\ga^{-1/2} \mathbb{Z}^d$, $s\in\{1,\dots S\}$, then the
matrix $B:=D^2_\La f(\rho_\La;\bar \rho_{\La^c})$, $\bar
\rho_{\La^c}\in \mathcal{X}^{(k)}_{\La^c}$, is strictly positive:
   \begin{equation}
      \label{Ie.2.16}
\big(u,B u\big) \ge \kappa (u,u),\;\;\text{for
 all $u \in \mathcal H$}
     \end{equation}
Same statements holds for $B_\eps:=D^2_\La
f_\eps(\rho_\La;\bar \rho_{\La^c})$, $\eps>0$.

\end{thm}

\vskip.5cm

{\bf Proof.}  The proof is analogous to that of Theorem 5.5 in
\cite{DMPV2} for completeness we sketch it.

  We will prove the theorem only in the case $\eps=0$.
Denoting by $\rho_\La^{-1}$ below the diagonal matrix with entries
$\rho_\La(x,s)^{-1}$
    \begin{equation*}
(u,B u) = (u,  V^*_{\ga,t} u)   +\frac{1}{\beta} (u,\rho_\La^{-1} u)
     \end{equation*}
Extend $u$  and $B$ as equal to 0 outside $\La$ and set
   $$
   U(x,s)=  \ga^{-d/2} \sum_{y\in  \ga^{-1/2} \mathbb Z^d} \hat J_\ga(x,y)u(y,s),
   \qquad x\in  \ga^{-1/2}\mathbb Z^d
   $$
Then, by \eqref{6.8} and using that $\hat J_\ga$ is symmetric,
            \begin{eqnarray*}
&& (u,B u) =    \sum_{s\ne s'} \sum_{x\in  \ga^{-1/2}\mathbb
Z^d}U(x,s)U(x,s') +\frac{1}{\beta} (u,\rho_\La^{-1}  u)  \nn\\&&
\hskip1.5cm =  \Big\{  \sum_{s\ne s'} \sum_{x\in  \ga^{-1/2}\mathbb
Z^d}U(x,s)U(x,s') + \sum_s\sum_{x\in  \ga^{-1/2}\mathbb Z^d}
[\frac{1}{\beta\rho^{(k)}(s)} -\kappa^*] U(x,s)^2\Big\} \nn\\&&
\hskip2cm -  \sum_s\sum_{x\in  \ga^{-1/2}\mathbb
Z^d}[\frac{1}{\beta\rho^{(k)}_s} -\kappa^*] U(x,s)^2
+\frac{1}{\beta} (u,\rho_\La^{-1}  u)
            \end{eqnarray*}
By \eqref{2.4} the curly bracket is non negative as well as
$[1/\beta\rho^{(k)}(s)-\kappa^*]$.  Since, by Cauchy
Schwartz inequality, for each $s$
      $$
\sum_{x\in  \ga^{-1/2}\mathbb Z^d}  U(x,s)^2 \le \sum_{x\in
 \ga^{-1/2}\mathbb Z^d} u(x,s)^2
       $$
then
      $$
\sum_s\sum_{x\in  \ga^{-1/2}\mathbb Z^d}[\frac{1}{\beta\rho^{(k)}_s}
-\kappa^*] U(x,s)^2 \le ( u,[\frac{1}{\beta\rho^{(k)}}-\kappa^*]u)
       $$
Thus
   \begin{equation*}
 (u,B  u) \ge   (u, [\kappa^*+ (\beta\rho_\La)^{-1} -(\beta\rho^{(k)})^{-1}]  u)
     \end{equation*}
and \eqref{Ie.2.16} is proved.

 \qed

\vskip1cm

A minimizer $\hat \rho_\La$ of $f$ is not necessary  a solution of
$D_\La f=0$. However a property analogous to the one states in
 Lemma 5.4 of \cite{DMPV2} holds.

\medskip

            \begin{lemma}
            \label{lemmaappK.2.1}
Any minimizer $\hat\rho_\La$  of $\{f(\rho_\La;\bar
\rho_{\La^c})$, $ \rho_{\La}\in \mathcal{X}_{\La}^{(k)}\}$,
is ``a critical point'' in the sense:

$\bullet$\; If for some $(x,s)$, $x\in \La$, $s=1,\dots,S$,
$|\text{\rm Av} (\hat\rho_\La;x,s)-\rho^{(k)}_{s}| <\zeta$
(strictly!), then
   \begin{equation}
      \label{I.3.2.14}
\frac{\partial f(\hat \rho_\La;\bar \rho_{\La^c})}{\partial
\rho_\La(y,s')} (y,s')  =0, \quad \text{for all $y\in
C^{(\ell_-)}_{z_x}: s'=s$}
     \end{equation}

$\bullet$\; If instead  $ \text{\rm Av}
(\hat\rho_\La;x,s)=\rho^{(k)}_{s}\pm \zeta$, then  for all $\xi$
with positive average, i.e. $\text{\rm Av} (\xi;x,s)>0$
   \begin{equation}
      \label{I.3.2.15}
\sum_{y\in  C^{(\ell_-)}_{z_x}\cap  \ga^{-1/2}\mathbb Z^d}\xi(y)
\frac{\partial f(\hat \rho_\La;\bar \rho_{\La^c})}{\partial
\rho_\La(y,s)}  (y,s) \le 0,\;\;\text{respectively $\ge 0$}
     \end{equation}
     while for all  $\xi$ with null average, i.e. $\text{\rm Av}
(\xi;x,s)=0$
        \begin{equation}
      \label{I.3.2.16}
\sum_{y\in  C^{(\ell_-)}_{z_x}\cap \ga^{-1/2}\mathbb Z^d} \xi(y)
\frac{\partial f(\hat \rho_\La;\bar \rho_{\La^c})}{\partial
\rho_\La(y,s)} (y,s)= 0
     \end{equation}

\end{lemma}

\vskip1cm

The following holds.

 \vskip.5cm

\begin{thm}
 \label{thmappK.2.2}
Given any $b\ge 2$ and $\kappa \in (0,\kappa^*)$ ($\kappa^*$ as in
\eqref{2.4})  for all $\ga$ and $\eps$ small enough, as well as for
$\eps=0$ the following holds. Let $\hat\rho_{\La,\eps}$ be either a
minimizer of $\{f_\eps(\rho_\La, \bar \rho_{\La^c})$, $\rho_\La\in
\mathbf{W}^{(k)}_\La\}$ or, if $\eps=0$, a ``critical point'' (in
the sense of Lemma \ref{lemmaappK.2.1}) satisfying  the inequality
$|\hat\rho_\La(x,s)-\rho^{(k)}_{s}|\le b\zeta$ for all $(x,s)$,
$x\in \La\cap  \ga^{-1/2}\mathbb Z^d$
  $s\in\{1,\dots, S\}$.
Then for any $\rho_\La$ such that $|\rho_\La(x,s)-\rho^{(k)}_{s}|
\le b\zeta$ for all $(x,s)$ as above,
   \begin{equation}
      \label{I.3.2.17}
f_\eps(\rho_\La; \bar \rho_{\La^c})\ge
 f_\eps(\hat\rho_\La;
\bar \rho_{\La^c})
 + \frac{\kappa}{2}\; \big(\rho_\La-\hat\rho_{\La,\eps},
 \rho_\La-\hat\rho_{\La,\eps}\big),\quad \text{both for
 }\eps>0\text{ and  for }\eps=0
     \end{equation}

\end{thm}

{\bf Proof.}  The proofs is the same as the one of Theorem 5.6 in
\cite{DMPV2}. Given any $\rho_\La$ as in the statement of the
Theorem, we interpolate by setting $\rho_{\La,\eps}(\theta)=
\theta\rho_\La+ (1-\theta) \hat \rho_{\La,\eps}$, $ \theta\in [0,1]
$, then with $ f_\eps(\theta):= f_\eps(\rho_{\La,\eps}(\theta)
;\bar\rho_{\La^c}) $
  \begin{eqnarray*}
     \label{4.16}
f_\eps(1) -  f_\eps(0)
&=& \int_0^1 \big(D_\La f_\eps(\theta),\rho_\La- \hat\rho_{\La,\eps}\big) \, d\theta\\
&=& \int_0^1  \int_0^\theta \big(D^2_\La f_\eps(\theta')[\rho_\La-
\hat\rho_{\La,\eps}],\rho_\La- \hat\rho_{\La,\eps}\big) \,d\theta'
\,d\theta + \big(D_\La f_\eps(0),\rho_\La- \hat\rho_{\La,\eps}\big)
    \end{eqnarray*}
 Since $\hat\rho_{\La,\eps}$ is a minimizer, then
 $\big(D_\La f_\eps(0),\rho_\La - \hat\rho_{\La,\eps}\big)\ge 0$. This
is immediate in the case $\eps >0$, while it follows from Lemma
\ref{lemmaappK.2.1} applied to $\xi = \rho_\La - \hat\rho_\La$ in
the case $\eps=0$.  Hence \eqref{I.3.2.17} follows from Theorem
\ref{thmappK.2.1}. \qed

\vskip.5cm

As shown in \cite{DMPV2}, an immediate consequence of
Theorem \ref{thmappK.2.2} is:

\vskip.5cm

 \begin{coro}
 \label{coroappK.2.1}
For any $\ga$ and $\eps>0$ small enough the minimizer of $f_\eps$ is
unique, same holds at $\eps=0$ for $f$. For $\eps>0$ (and small
enough) there is a unique critical point in the space $\{|\rho_{\La}
-\rho^{(k)} \big| \le 2\zeta\}$; such a critical point minimizes
$f_\eps$. Analogously, when $\eps=0$ there is a unique critical
point in the sense of Lemma \ref{lemmaappK.2.1}. Such a critical
point minimizes $f$. The minimizer of $f_\eps$, $\eps>0$, converges
as $\eps\to 0$ to the minimizer of $f$.

 \end{coro}

\vskip2cm


\subsection{Exponential decay}

       \label{subsec:appK.3}

\vskip.5cm

\begin{prop}
 \label{propappK.3.1}
For all $\eps>0$ small enough, the minimizer 
of $\{f_\eps(\rho_\La; \rho^{(k)}\mathbf{1}_{\La^c}), \rho_\La\in
\mathbf{W}^{(k)}_\La\}$ is $\rho^{(k)}\mathbf{1}_{\La}$ which is
also the minimizer of $\{f(\rho_\La; \rho^{(k)}\mathbf{1}_{\La^c}),
\rho_\La\in \mathcal{X}^{(k)}_\La\}$.

\end{prop}

\vskip.5cm

{\bf Proof.}  Case $\eps=0$. If $D_\La
f(\rho_\La;\rho^{(k)}\mathbf{1}_{\La^c})=0$,  then for all $x \in
\La\cap\ga^{-1/2}\mathbb{Z}^d$ and $s\in\{1,\dots S\}$
   \begin{equation*}
  \rho_\La(x,s) = \exp\Big\{-\beta[ \sum_{s'\ne s}\sum_{y\in
  \ga^{-1/2} \mathbb{Z}^d}\hat V_{\ga,t}(x,y) \rho(y,s')
  -\la_\beta-(1-t)r^{(k)}_{s}]\Big\}
     \end{equation*}
where $\rho(x,s)$ on the r.h.s.\ is equal to $\rho_\La(x,s)$ if
$x\in \La\cap \ga^{-1/2} \mathbb{Z}^d$ and to $\rho^{(k)}_s$ if
$x\in \La^c\cap \ga^{-1/2} \mathbb{Z}^d$. The function
$\rho_\La(x,s)=\rho^{(k)}_{s}$ solves the above equation and by
Corollary \ref{coroappK.2.1} it is unique and the only minimizer of
$\{f(\rho_\La; \rho^{(k)}\mathbf{1}_{\La^c}), \rho_\La\in
\mathcal{X}^{(k)}_\La\}$.

 Case $\eps>0$ and small. Being a critical point of $f_\eps$ as well,
$\rho^{(k)}\mathbf{1}_{\La}$ is by Corollary \ref{coroappK.2.1} also
a minimizer of $f_\eps$.
 \qed

\vskip1cm

\begin{thm}
 \label{thmappK.3.1}
There are $\om>0$ and $c>0$ such that for all $\ga$ small enough the
minimizer $\hat {\rho}_\La$ of $\{f(\rho_\La;\bar{\rho}_{\La^c}),
\rho_\La \in \mathcal{X}^{(k)}_{\La}\}$ satisfies
   \begin{equation}
      \label{appK.3.1}
|\hat {\rho}_\La(x,s)- \rho^{(k)}_{s}| \le c \ga^{-d/2}e^{-\om \ga
\text{\rm dist}(x, \La^c)},\qquad x\in \La\cap
\ga^{-1/2}\mathbb{Z}^d, \quad s\in\{1,\dots, S\}
     \end{equation}

\end{thm}

\vskip.5cm

{\bf Proof.}  By Proposition \ref{propappK.3.1},
$\rho^{(k)}\mathbf{1}_{\La}$ is
 the minimizer  of $f_\eps(\cdot; \rho^{(k)}\mathbf{1}_{\La^c})$,
(both for $\eps>0$ and for $\eps=0$)
 thus the difference $\hat {\rho}_{\La,\eps}(x,s)- \rho^{(k)}_{s}$
$x\in\La\cap \ga^{-1/2}\mathbb{Z}^d$, can be seen as the difference
of two minimizers relative to two different boundary conditions and
we can then proceed as in the proof of Theorem 5.9 of \cite{DMPV2}.
However the proof is different: the complication comes from the fact
that the constraint is here on the averages and not on the
elementary variables $\rho_\La(x,s)$ as it was in \cite{DMPV2}. Thus
the strategy of the proof of \eqref{appK.3.1} is to reduce to the
case already treated in \cite{DMPV2}.

Define for $\theta\in [0,1]$,
    \begin{equation}
     \label{appK.3.2}
\bar \rho_{\La^c}^{(\theta)}:= \theta \bar \rho_{\La^c} + (1-\theta)
\rho^{(k)}\mathbf{1}_{\La^c}
     \end{equation}
and call $\hat \rho_{\La,\eps}(x,s;\theta)$, $x\in \La\cap
\ga^{-1/2}\mathbb{Z}^d$, the minimizer of $f_\eps(\rho_\La;\bar
\rho_{\La^c}^{(\theta)})$ (both for $\eps>0$ and for $\eps=0$).  The
same argument used in the proof of Theorem 5.9 of \cite{DMPV2} shows
that $\hat \rho_{\La,\eps}(\cdot;\theta)$ is differentiable in
$\theta$ such that
   \begin{equation}
      \label{appK.3.3}
|{\hat \rho}_\La(x,s)-\rho^{(k)}_{s}| \le \lim_{\eps\to 0} \int_0^1
\Big|\frac{d\hat \rho_{\La,\eps}(x,s;\theta)}{d\theta}\Big| d\theta
     \end{equation}
We now estimate $\dis{|\frac{d\hat
\rho_{\La,\eps}(x,s;\theta)}{d\theta}|}$ uniformly in $\eps$ and
$\theta$ and this  will give \eqref{appK.3.1}.

Shorthand
   \begin{equation}
      \label{appK.3.4}
u:=\frac{d\hat \rho_{\La,\eps}(\cdot;\theta)}{d\theta}\hat ,\;\;\;\;
 v=-  \frac {d}{d\theta'}
D_\La f_{\eps}(\hat \rho_{\La,\eps}(\cdot;\theta); \bar
\rho_{\La^c}^{(\theta')}) \Big|_{\theta'=\theta}
     \end{equation}
we have that
    \begin{equation}
      \label{app7.49y}
Bu=v
          \end{equation}
where
        \begin{equation}
      \label{appK.3.5}
B:=D^2_\La f_{\eps}(\hat \rho_{\La,\eps}(\cdot;\theta))
     \end{equation}
So  we need to bound $|u(x,s)|$, $x \in \La\cap
\ga^{-1/2}\mathbb{Z}^d$. Explicitly, (recall the notation in
Subsection \ref{subsec:6.1})
   \begin{equation}
     \label{appK.3.6}
B(x,s,y,s') =  \hat V_{\ga,t}(x,y)\mathbf{1}_{s\ne s'} + \frac
{\text{\bf 1}_{(x,s)=(y,s')}} {\beta
\hat\rho_{\La,\eps}(x,s;\theta)} + \phi_{\eps,\theta}(x,s)\,\,
\text{\bf 1}_{\{C^{(\ell_-)}_{z_x}=C^{(\ell_-)}_{z_y}, s'=s\}}
     \end{equation}
and
   \begin{equation}
      \label{appK.3.7}
\phi_{\eps,\theta}(x,s)=3\eps^{-1}  \Big( \{({\rm
Av}(\hat\rho_{\La,\eps}(\cdot;\theta);x,s)
-[\rho^{(k)}_{s}+\zeta])_+\}^2 +\{({\rm
Av}(\hat\rho_{\La,\eps}(\cdot;\theta);x,s) -
[\rho^{(k)}_{s}-\zeta])_-\}^2 \Big)
     \end{equation}
Finally
    \begin{equation}
    \label{6.55}
v(x,s)=-\sum_{s'\ne s}\sum_{y\in\La^c\cap\ga^{-1/2}\mathbb{Z}^d}\hat
V_{\ga,t}(x,y)[\bar\rho_{\La^c}(y,s')-\rho_{s'}^{(k)}]
    \end{equation}

Recalling the definition of $V^{(\ell_-)}_{\ga,t}$ in
\eqref{appK.1.0.6} we define for all $x,y\in
\La\cap\ga^{-1/2}\mathbb{Z}^d$ and the corresponding $z_x,z_y\in
\ell_-\mathbb{Z}^d$,
        \begin{equation}
     \label{appK.3.8}
A(x,s,y,s') = V^{(\ell_-)}_{\ga,t}(z_x,z_y) \, \mathbf{1}_{s\ne s'}+
\frac {\text{\bf 1}_{(x,s)=(y,s')}} {\beta \rho^{(k)}_{s}} +
\phi_{\eps,\theta}(x,s)\,\,
 \text{\bf 1}_{\{C^{(\ell_-)}_{z_x}=C^{(\ell_-)}_{z_y}, s'=s\}}
     \end{equation}
calling $R(x,s,y,s')=B(x,s,y,s')-A(x,s,y,s')$. We extend these three
matrices to all $x,y\in \ga^{-1/2}\mathbb{Z}^d$ by setting their
entries equal 0 outside $\La$.

We denote by $\|E\|$, the norm of the matrix $E$ in $\mathcal{H}$
($\dis{\|E\|^2=\sup_{u\ne 0}(Eu,Eu)/(u,u)}$), and we observe that
    \begin{equation}
    \label{6.56}
\|E\|\le
\max_{x,s}\Big(\sum_{y,s'}|E(x,s,y,s')|,\sum_{y,s'}|E(y,s',x,s)|
\Big)
    \end{equation}

By Theorem \ref{thmappK.2.1}, $B$ is a positive matrix such that the
inverse $B^{-1}$ is well defined. Since the same proof applies also
to the matrix $A$, we have that $A$ as well is a positive matrix
with a well defined inverse $A^{-1}$ and $\|A^{-1}\|\le C$.

Moreover by \eqref{appK.1.0.7} and \eqref{a6.40} we get, for $\ga$
small enough,
    \begin{eqnarray*}
&&\hskip-2cm\sum_{s'}\sum_{y\in\La \cap
\ga^{-1/2}\mathbb{Z}^d}|R(x,s,y,s')|\le \sum_{s'}\sum_{y\in\La \cap
\ga^{-1/2}\mathbb{Z}^d}|V_{\ga,t}(x,s,y,s')-V^{(\ell_-)}_{\ga,t}(z_x,s,z_y,s')|
                \\&&\hskip2.8cm
+\big| \frac {1} {\beta \hat\rho_{\La,\eps}(x,s;\theta)}-\frac {1}
{\beta \rho^{(k)}_{s}}\big|\\&& \le c_1(\ga\ell_-)(\ga^{-d/2}\ga^d)
\sum_{s'}\sum_{y\in\ga^{-1/2}\mathbb{Z}^d}\text{\bf 1}_{|x-y|\le 2
\ga^{-1}}+ \hat c^*\zeta\le S\bar c(\ga\ell_-+\zeta)\le  c\zeta
    \end{eqnarray*}
Thus, from \eqref{6.56} we have that $\|R\|\le c\zeta$.

Since $\|A^{-1}\|\|R\|\le Cc\zeta<1$, for $\ga$ small enough, we
have that the series below is convergent and
        \begin{equation}
     \label{appK.3.9}
B^{-1} = \sum_{n=0}^\infty (-A^{-1}R)^n A^{-1}
     \end{equation}
For a proof of this statement see Theorem A.1 in Appendix A of
\cite{DMPV2}.

 We are going to prove that
there are $\om$ and $c$ positive such that for $x\in \La\cap
\ga^{-1/2}\mathbb{Z}^d$,
        \begin{equation}
     \label{6.58}
\sum_{s'}\sum_{y\in \La\cap
\ga^{-1/2}\mathbb{Z}^d}|B^{-1}(x,s,y,s')| e^{\om \ga |x-y|}\le C^*
     \end{equation}
We first show why the estimate \eqref{6.58} concludes the proof of
the Theorem. From the definition \eqref{app7.49y} and using
\eqref{6.58}, we have that for all $x\in \La\cap
\ga^{-1/2}\mathbb{Z}^d$, and for all $s\in\{1,\dots,S\}$,
    \begin{eqnarray}
    \nn
&&\hskip-.5cm|u(x,s)|=\sum_{s'}\sum_{y\in \La\cap
\ga^{-1/2}\mathbb{Z}^d}[ B^{-1}(x,s,y,s')e^{\om \ga |x-y|}][ e^{-\om
\ga |x-y|}v(y,s')]
            \\&&\hskip1cm
\le C^* \max_{s', y\in \La\cap \ga^{-1/2}\mathbb{Z}^d}[ e^{-\om \ga
|x-y|}v(y,s')]
    \label{6.59}
    \end{eqnarray}
Recalling \eqref{6.55}
        $$
\max_{s', y\in \La\cap \ga^{-1/2}\mathbb{Z}^d}[ e^{-\om \ga
|x-y|}v(y,s')]\le c\frac{\ga^{-d}}{\ga^{-d/2}} \max_{y\in \La\cap
\ga^{-1/2}\mathbb{Z}^d}[ e^{-\om \ga
|x-y|}\mathbf{1}_{\text{dis}(y,\La^c)\le \ga^{-1}}]
        $$
thus \eqref{appK.3.1} follows from \eqref{6.59} and
\eqref{appK.3.3}.

\bigskip {\bf Proof of \eqref{6.58}}. From \eqref{appK.1.0.7} we have
         \begin{eqnarray}
        \nn
&&\hskip-3cm \sum_{s'}\sum_{y\in \La\cap
\ga^{-1/2}\mathbb{Z}^d}|R(x,s,y,s')| e^{\om \ga |x-y|}\le
c_1(\ga\ell_-)(\ga^{-d/2}\ga^d)
            \\&&\hskip1.7cm
\times \sum_{s'}\sum_{y\in \La\cap \ga^{-1/2}\mathbb{Z}^d}e^{\om \ga
|x-y|}\text{\bf 1}_{|x-y|\le 2 \ga^{-1}}+c\zeta\le \bar c\zeta
     \label{appK.3.9.1}
     \end{eqnarray}
We will prove that
         \begin{equation}
     \label{appK.3.9.2}
\sum_{s'}\sum_{y\in \La\cap
\ga^{-1/2}\mathbb{Z}^d}|A^{-1}(x,s,y,s')| e^{\om \ga |x-y|}\le c''
     \end{equation}
then, by \eqref{appK.3.9}
        \begin{equation}
     \label{appK.3.10.0}
\sum_{s'}\sum_{y\in \La\cap
\ga^{-1/2}\mathbb{Z}^d}|B^{-1}(x,s,y,s')| e^{\om \ga |x(i)-x(j)|}\le
\frac{ c''}{1-c''\bar c\zeta}
     \end{equation}
such that we are reduced to the proof of \eqref{appK.3.9.2}.

\medskip

We define $e_{(x,s)}(y,s')=\text{\bf 1}_{(x,s)=(y,s')}$, then
    $A^{-1}(x,s,y,s')=(e_{(x,s)}, A^{-1}e_{(y,s')})
$. Thus
        \begin{equation}
     \label{appK.3.13}
\text{for } u=A^{-1}e_{(y,s')},\qquad A^{-1}(x,s,y,s')=u(x,s)
     \end{equation}
We say that two pairs $(x,s)$, $(y,s')$ are equivalent,
    \begin{equation}
    \label{6.65}
(y,s')\sim (x,s),\quad \Leftrightarrow \quad
C^{(\ell_-)}_{z_y}=C^{(\ell_-)}_{z_x}, \quad \text {and } s'=s
    \end{equation}
We call $\dis{ N= (\frac{\ell_-}{\ga^{-1/2}})^d}$, and for any
$(x,s)$, $x\in \La\cap \ga^{-1/2}\mathbb{Z}^d$ we define
    \begin{equation}
     \label{appK.3.11}
e_{(x,s)}^*=\frac 1N\sum_{(x',s')'\sim (x,s)} e_{(x',s')}, \quad
e_{(x,s)}^{\perp}=\frac{N-1}N \big(e_{(x,s)}-\frac
1{N-1}\sum_{(x',s')\sim (x,s), (x',s')\ne (x,s)} e_{(x',s')}\big)
         \end{equation}
With these definitions,
    \begin{equation}
    \label{6.66}
(e_{(x,s)}^{\perp},e_{(x,s)}^*)=0,\qquad
e_{(x,s)}=e^*_{(x,s)}+e_{(x,s)}^{\perp}
    \end{equation}
and then, according to \eqref{appK.3.13}  we need to solve $Au=
e_{(y,s')}^*+ e_{(y,s')}^{\perp}$.

Recalling \eqref{appK.3.8} and observing that for all
$x\in\La\cap\ga^{-1/2}\mathbb{Z}^d$ and all $s$
$\phi_{\eps,\theta}(x,s)=\phi_{\eps,\theta}(z_x,s)$, we call
        \begin{equation}
     \label{6.68}
U^{(\ell_-)}_{\ga,t}(z_x,s,z_y,s') =
 V^{(\ell_-)}_{\ga,t}(z_x,z_y)\mathbf{1}_{s\ne s'}+\phi_{\eps,\theta}(z_x,s)\,\,
 \text{\bf 1}_{\{C^{(\ell_-)}_{z_x}=C^{(\ell_-)}_{z_y}, s=s'\}}
     \end{equation}
and we observe that $U^{(\ell_-)}_{\ga,t}(z_x,s,z_y,s')$ assume the
same value in all points equivalent to $(x,s)$ and to $(y,s')$. Thus
$A$ can be considered also as an operator from $\mathcal H_{\ell_-}$
onto itself in the following way. Let $\bar u\in \mathcal
H_{\ell_-}$, $\bar u=(\bar u(z,s), z\in \ell_-\mathbb{Z}^d)$, and
let $ u$ its extension to $\mathcal H_{\ga^{-1/2}}$ given by $
u(x,s)=\frac 1N\bar u(z,s)$ for all $x\in
C^{(\ell_-)}_z\cap\ga^{-1/2}\mathbb{Z}^d$. Then for any
$x\in\ga^{-1/2}\mathbb{Z}^d$ and any $s$,
    \begin{eqnarray*}
Au(x,s)=N\sum_{s'\ne s}\sum_{z'\in \ell_-\mathbb{Z}^d}
U^{(\ell_-)}_{\ga,t}(z_x,s,z',s')\frac {\bar u(z',s')}N+\frac
1{\beta\rho^{(k)}_{s}} \frac {\bar u(z_x,s)}N =:A^{(\ell_-)}\bar
u(z_x,s)
    \end{eqnarray*}
where
        $$
A^{(\ell_-)}(z,s,z',s')= U^{(\ell_-)}_{\ga,t}(z,s,z',s')+\frac
1{N\beta\rho^{(k)}_{s}} \mathbf{1}_{(z,s)=(z's')}
        $$
Observe that 
$e_{(x,s)}^*(y,s')=\frac 1 N \bar e^*_{(z_x,s)}(z_y,s')$ where $\bar
e^*_{(z,s)}(z',s')=\mathbf{1}_{z=z',s=s'}$ is a vector in $ \mathcal
H_{\ell_-}$. Thus we find $u= u^*+ u^{\perp}$ by solving separately
$Au^{\perp}=e_{(x,s)}^{\perp}$ and
  $Au^*=e_{(x,s)}^*$ with $u^*=\frac 1N\bar u^*$, $\bar u^*\in \mathcal
  H_{\ell_-}$ and $A^{(\ell_-)}\bar u^*= \bar e_{(z,s)}^*$.
  By direct inspection we have
             \begin{equation}
     \label{appK.3.14}
 u^{\perp}(y,s')= e_{(x,s)}^{\perp}(y,s')\beta
\rho^{(k)}_{s}
     \end{equation}
It is easy to see that
          \begin{equation}
     \label{appK.3.15}
\sum_{s'}\sum_{y\in \La\cap\ga^{-1/2}\mathbb{Z}^d} e^{\om \ga
|x-y|}|e_{(y,s')}^{\perp}(x,s)\beta \rho^{(k)}_{s}|  \le c
     \end{equation}

We prove below that there is a constant $\bar c$ so that
    \begin{equation}
     \label{6.73}
\sum_{s'}\sum_{z'\in \La\cap\ell_-\mathbb{Z}^d} e^{\om \ga |z-z'|} |
\bar u^*(z',s)|\le \bar c,\qquad A^{(\ell_-)}\bar u^*=\bar
e_{(z',s')}^*
     \end{equation}
\eqref{6.73}  concludes the proof of \eqref{appK.3.9.2} in fact
    \begin{equation*}
\sum_{s'}\sum_{y\in \La\cap\ga^{-1/2}\mathbb{Z}^d} e^{\om \ga
|x-y|}|u^*(y,s')|  \le e^{2\om\ga\ell_-}\sum_{s'}\sum_{z'\in
\La\cap\ell_-\mathbb{Z}^d} e^{\om \ga |z_x-z'|}|\bar u^*(z',s')|\le
C
     \end{equation*}

\medskip

{\bf Proof of \eqref{6.73}} The matrix $A^{(\ell_-)}$ acting on
functions constant on the scale $\ell_-$ is the same operator
considered in Theorem 5.9 of
 \cite{DMPV2} where a statement stronger than \eqref{6.73} has been proven.
 Thus the proof of \eqref{6.73} is contained in the proof of  Theorem 5.9 of
 \cite{DMPV2}, but, for the reader convenience we sketch it.
 To have the same notation as in \cite{DMPV2},
 we  use the label $i$ for a pair $(x,s)$,
$x\in\ell_- \mathbb Z^d, s \in\{1,.,S\}$, writing $x(i)=x$, $s(i)=s$
if $i=(x,s)$ and shorthand $|i-j|$ for $|x(i)-x(j)|$ and $i\in \La$
for $x(i)\in \ell_-\mathbb Z^d\cap\La$. We call $A_0$ the matrix
    \begin{equation}
     \label{6.75}
A_0(i,j) = V^{(\ell_-)}_{\ga,t}(i,j)\text{\bf 1}_{s(i)=s(j)}+ \frac
{\text{\bf 1}_{i=j}} {N\beta \rho^{(k)}_{s(i)}},\qquad i,j\in
\La\cap \ell_-\mathbb{Z}^d
     \end{equation}
so that  $A^{(\ell_-)}=A_0+\phi_{\eps,\theta}$ and
$\phi_{\eps,\theta}$ is a diagonal matrix in $\mathcal H_{\ell_-}$.

We need to distinguish values $i$ where $\phi_{\eps,\theta}(i)$ is
large, and so we define
    \begin{equation}
      \label{6.76}
G=\big\{i\in\La: \phi_{\eps,\theta}(i) \ge b\big\},\quad \mathcal
H_G= \big\{u\in \mathcal H_{\ell_-}: u(i)=0,\; \text{for all}\;i\in
G^c\big\}
     \end{equation}
Let $Q$ be the orthogonal projection on $ \mathcal H_G$ and $P=1-Q$,
thus $Q$ selects the sites where $\phi_{\eps,\theta}$ is large and
$P$ those where it is small. Let
    \begin{equation}
      \label{6.77}
C= PAP - PA_0(QAQ)^{-1}QA_0
     \end{equation}
We look for $\bar u^*$ such that $A^{(\ell_-)}u^*=\bar e_j^*$. In
\cite{DMPV2} (see eqs (5.34)--(5.38) in \cite{DMPV2} ) it is proven
that  $C$ is invertible on the range of $P$ and that
   \begin{equation}
      \label{6.78}
P\bar u^*=C^{-1}\{ Pe_j^*-PA_0(QAQ)^{-1}Qe_j^*\}
     \end{equation}
    \begin{equation}
      \label{6.79}
    Q\bar u^*  = (QAQ)^{-1}\{Qe_j^*- QA_0Pu^*\}
         \end{equation}
The matrix $C$ satisfies the hypothesis of Theorem A.1 and Theorem
A.2 of \cite{DMPV2} so that there are $c_1$, $c_2$ and $\kappa$ such
that
    \begin{equation}
      \label{6.80}
|C^{-1}(i,j)| \le  c_1\exp\Big\{-\kappa \ga|x(i)-x(j)| \Big\}
         \end{equation}
         \begin{equation}
      \label{6.81}
\sum_i |(QAQ)^{-1}(i,j)| e^{\ga |x(i)-x(j)|} \le \frac {c_2}b
     \end{equation}
Furthermore if $b$ is large enough and $c\ge\|A_0\|$, there is $c'$
such that
   \begin{equation}
      \label{6.82}
\| PA_0(QAQ)^{-1}QA_0\|\le  \frac {2c^2}{b},\quad \|
PA_0(QAQ)^{-1}QA_0\|_\infty\le  c'\frac {2c^2}{b}
     \end{equation}
By observing that for any $i,j\in\ell_-\mathbb{Z}^d$, $\bar
e_j^*(i)=1_{i=j}$, from \eqref{6.78}--\eqref{6.82}, inequality
\eqref{6.73} follows.
 \qed

\vskip1cm

\section{{\bf Surface corrections to the pressure}}
    \label{sec:5}

In this Section we fix the chemical potential $\la=\la_{\beta,\ga}$
such that \eqref{p2.30} holds and we prove Theorem \ref{thm:2.14}.
We often omit to write explicitly the dependence on
$\la_{\beta,\ga}$.

\subsection{Interpolating Hamiltonian}
    \label{sec:5.1}
 We write an interpolation formula for the partition functions
$Z^{(k)}_{\La} (\chi^{(k)}_{\La^c}) \equiv
Z^{(k)}_{\La,\la_{\beta,\ga}} (\chi^{(k)}_{\La^c})$,  $k\in\{1,\dots
,S+1\}$ defined in \eqref{I.4.2}. Here $\La$ is a
$\mathcal{D}^{\ell_+}$- measurable set and the boundary condition
$\chi^{(k)}_{\La^c}=\rho^{(k)}\text{\bf 1}_{r\in \La^c}$.

The reference hamiltonian $h(\rho_\La)$,
$\rho_\La=\rho(\cdot;q_\La))$, is the one defined in \eqref{aa2.47}
and it has been chosen in such a way that
 for all $x$  $\in\ga^{-d/2}\cap
 \La$ and all $s\in\{1,\dots,S\}$,
     \begin{equation}
    \label{aa2.48}
\frac 1 Z \int e^{-\beta
h(\rho_\La)}\rho(x,s;q_\La)\nu(dq_\La)=\rho^{(k)}_s,\qquad Z=\int
e^{-\beta h(\rho_\La)}\nu(dq_\La)
    \end{equation}
that can be proved by using the mean field equation \eqref{2.3}.

 Recalling \eqref{I.4.3a}, we call
    \begin{equation}
     \label{aa2.49}
Z^{(k)}_{\La,0} :=\int_{\mathcal{X}_\La^{(k)}} e^{-\beta
h(\rho_\La)} \mathbf{X}^{(k)}_{\La,\la_{\beta,\ga}} ( q_{\La})
\nu(dq_\La)
    \end{equation}
For any $t\in(0,1)$ we let
 ($H_{\La}=H_{\La,\la_{\beta,\ga}}$ below)
    \begin{equation}
     \label{aa2.51}
H_{\La,t}(q_\La|\chi^{(k)}_{\La^c})=t
H_\La(q_\La|\chi^{(k)}_{\La^c})+(1-t)h(\rho_\La)
    \end{equation}
and ($\mathbf{X}^{(k)}_{\La}
=\mathbf{X}^{(k)}_{\La,\la_{\beta,\ga}}$ below)
    \begin{equation}
     \label{aa2.52}
Z^{(k)}_{\La,t}(\chi^{(k)}_{\La^c})
:=\int_{\mathcal{X}_\La^{(k)}}e^{- \beta
H_{\La,t}(q_\La|\chi^{(k)}_{\La^c})}\mathbf{X}^{(k)}_{\La} (
q_{\La}) \nu(dq_\La)
    \end{equation}

We thus have a formula for our partition function $Z^{(k)}_{\La}
(\chi^{(k)}_{\La^c})$, namely
    \begin{equation}
     \label{aa2.53}
\frac 1 \beta\log Z^{(k)}_{\La}(\chi^{(k)}_{\La^c})=\frac 1\beta
\log Z^{(k)}_{\La,0} -\int_0^1 \mathbb{E}_{\La,t}
\big(H_\La(q_\La|\chi^{(k)}_{\La^c})-h(\rho_\La)\big)dt
    \end{equation}
where $\mathbb{E}_{\La,t}$ is the expectation with respect to the
following probability measure $\mu_{\La,t,k}$ with support on the
set $\mathcal{X}_\La^{(k)}$
    \begin{equation}
     \label{aa2.54}
\mu_{\La,t,k}(dq_\La)=\frac 1 {Z^{(k)}_{\La,t}(\chi^{(k)}_{\La^c})}
e^{-\beta H_{\La,t}(q_\La|\chi^{(k)}_{\La^c})}
\mathbf{X}^{(k)}_{\La} ( q_{\La})\, {\bf
1}_{\{q_\La\in\mathcal{X}_\La^{(k)}\}}\nu(dq_\La)
    \end{equation}

\bigskip

We call
    \begin{equation}
    \label{a5.10b}
\La_0=\La\setminus \delta_{\rm{in}}^{\ell^*_+}(\La),\qquad \La_{00}=
\La_0\setminus \delta_{\rm{in}}^{\ell^*_+}(\La_0),\qquad
\ell^*_+=\ga^{-1-\alpha^*_+}, \quad\alpha^*_+>\alpha_+
    \end{equation}

\bigskip

\subsection{Estimate close to the boundaries}
    \label{sec:5.2}

In this subsection we prove the following Theorem.
\bigskip

    \begin{thm}
    \label{thm52.1}
There is $ c_b>0$ such that
    \begin{equation}
    \label{5.11b}
\Big|\int_{\mathbb R^d\setminus\La_{00}}\mathbb{E}_{\La,t}\Big(\big[
e^{\rm mf}\big(J_\ga\star q_\La\cup \chi^{(k)}_{\La^c} \big) -e^{\rm
mf}(\rho^{(k)})\big]\Big)\Big| \le
c_b\ga^{1/4}\,|\La\setminus\La_{00}|
    \end{equation}
where
           \begin{equation}
    \label{emf}
e^{\rm mf}(u)=\frac\12\sum_{s,s':s\ne s'} u_su_{s'}
             \end{equation}
Furthermore letting $\mathfrak{L}_\ga=\ga^{-d/2}\cap
 \mathbb{Z}^d$,
    \begin{equation}
    \label{a5.11b}
\Big|\sum_s\ga^{-d/2}\sum_{x\in \La\setminus\La_{00}\cap
\mathfrak{L}_\ga }\mathbb{E}_{\La,t}\Big(\hat
J_\ga\star[\rho_\La(x,s,q_\La) - \rho^{(k)}(s)] \Big)\Big|\le
c_b\ga^{1/4}\,|\La\setminus\La_{00}|
    \end{equation}
    \end{thm}

\bigskip
    {\bf Proof.}
We first observe that (recall \eqref{2.8})
    \begin{equation}
    \label{5.12b}
\big|J_\ga\star q_\La\cup \chi^{(k)}_{\La^c} -\hat
J_\ga\star\rho\big|\le c\ga^{1/2}
    \end{equation}
where, recalling \eqref{a3.1},
    \begin{equation}
      \label{5.13b}
\rho(r,s)= \rho^{(\ga^{-1/2})}(r,s;q_\La)+ \rho^{(k)}{\text{\bf
1}}_{\La^c}
     \end{equation}
We write
    \begin{equation}
[ e^{\rm mf}\big(\hat J_\ga\star\rho\big) -e^{\rm
mf}(\rho^{(k)})\big](x)= e^{\rm mf}\big(\hat J_\ga\star
[\rho-\rho^{(k)}]\big)(x)+\sum_{s,s';s\ne s'} \rho^{(k)}_s \hat
J_\ga\star [\rho-\rho^{(k)}](x,s')
    \label{5.14b}
     \end{equation}
     We call
    \begin{equation}
    \label{a5.18b}
\Delta=(\La\setminus\La_{00})\cup \big(\delta^{(2\ga^{-1})}_{\rm
{out}}[\La\setminus\La_{00}]\cap \La\big)
     \end{equation}
We define the set (recall the definition \eqref{5.13b} of $\rho$)
    \begin{equation}
\mathcal{A}_\Delta=\Big\{q_{\Delta}\in
\mathcal{X}_\Delta^{(k)}:\sum_s\int_{\mathbb R^d\setminus\La_{00}}
\big|\hat J_\ga\star [\rho-\rho^{(k)}]\big|\le \bar{C}\ga^{1/4}
|\La\setminus\La_{00}|\,\Big\},
    \label{5.15b}
     \end{equation}
where $\bar{C}$ is some large constant. We also call
     \begin{equation}
\bar{\mathcal  A}_\Delta=  \mathcal{X}_\Delta^{(k)}\setminus
\mathcal{A}_\Delta
    \label{5.15bb}
     \end{equation}
From \eqref{5.12b}, \eqref{5.14b} and the fact that the interaction
range is $2\ga^{-1}$, we have that
    \begin{equation}
    \label{5.16b}
\Big|\int_{\mathbb R^d\setminus\La_{00}}\mathbb{E}_{\La,t}\Big(\big[
e^{\rm mf}\big(J_\ga\star\ q_\La\cup \chi^{(k)}_{\La^c} \big)
-e^{\rm mf}(\rho^{(k)})\big]\Big)\Big|\le
c[\ga^{1/2}+\bar{C}\ga^{1/4}+
\mu_{\La,t,k}(\bar{\mathcal{A}}_\Delta)] \,\big|\Delta\big|
    \end{equation}
We are left with the estimate of the probability on the r.h.s of
\eqref{5.16b}. We first write
    \begin{eqnarray}
    \label{a5.17b}
&&\hskip-1cm\mu_{\La,t,k}(\bar{\mathcal A}_\Delta)=\frac
{Z^{(k)}_{\La,t}(\bar{\mathcal  A}_\Delta|\chi^{(k)}_{\La^c})}
{Z^{(k)}_{\La,t}(\chi^{(k)}_{\La^c})}
\\&&\hskip-1cm
Z^{(k)}_{\La,t}(\bar{\mathcal  A}_\Delta|\chi^{(k)}_{\La^c})
=\int_{\bar{\mathcal  A}_\Delta}e^{-\beta
H_{\La,t}(q_\La|\chi^{(k)}_{\La^c})} \mathbf{X}^{(k)}_{\La} (
q_{\La})\nu(dq_\La)
    \label{5.17b}
    \end{eqnarray}
and we estimate separately the numerator and the denominator in
\eqref{a5.17b} starting from the numerator.
    \bigskip

The following estimate holds for the weights of the contours (see
Lemma \ref{Lemma2} in Appendix \ref{app:bounds})
        \begin{equation}
    \label{5.18b}
\mathbf{X}^{(k)}_{\La\setminus \big(\Delta\cup \delta_{\rm
{out}}^{(2\ell_+)}(\Delta)\big)} \le \mathbf{X}^{(k)}_{\La} \le
\mathbf{X}^{(k)}_{\La\setminus \big(\Delta\cup \delta_{\rm
{out}}^{(2\ell_+)}(\Delta)\big)}\big[1+ e^{-c
\zeta^2\ell_{-}^d}\big]^{| \Delta\cup \delta_{\rm
{out}}^{(2\ell_+)}(\Delta)|/\ell^d_+}
     \end{equation}

 We partition $\La$ in $\La=\Delta\cup B\cup(\La_{00}\setminus B)$ where
    \begin{equation}
B=\delta_{\rm{out}}^{\ell_+}[\Delta]\cap\La_{00}%
    \label{5.20b}
    \end{equation}
We also define the set $B_0\subset B$ such that
    \begin{equation}
    \label{b0}
\text{dist}(B_0,B^c)>\ell_+/8
    \end{equation}
    We write
    \begin{equation}
    \label{5.21b}
H_{\La,t}(q_\La|\chi^{(k)}_{\La^c})=H_{\La_{00}\setminus
B,t}(q_{\La_{00}\setminus
 B})+ H_{\Delta\cup B,t}(q_{\Delta\cup B
}|q_{\La_{00}\setminus B } \cup \chi^{(k)}_{\La^c})
        \end{equation}
and we notice that from \eqref{5.18b} we get
        \begin{eqnarray}
    \nn
&&Z^{(k)}_{\La,t}(\bar{\mathcal  A}_\Delta|\chi^{(k)}_{\La^c}) \le
\int_{\mathcal{X}^{(k)}_{\La_{00}\setminus B}}e^{-\beta
H_{\La_{00}\setminus B,t}(q_{\La_{00}\setminus
 B})}\mathbf{X}^{(k)}_{\La\setminus \big(\Delta\cup \delta_{\rm
{out}}^{(2\ell_+)}(\Delta)\big)}\big[1+ e^{-c
\zeta^2\ell_{-}^d}\big]^{[| \Delta\cup \delta_{\rm
{out}}^{(2\ell_+)}(\Delta)|]/\ell^d_+}
 \\&&\hskip2.5cm\,\,\times\hat
Z_{\Delta\cup B,t}(\bar{\mathcal  A}_\Delta|q_{\La_{00}\setminus B }
\cup \chi^{(k)}_{\La^c})
        \label{5.22b}
        \end{eqnarray}
where $ \hat Z^{(k)}_{\Delta\cup B,t}$ is as in \eqref{2.41nn} but
with the energy given by $H_{\Delta\cup B,t}$. We call $\bar
\rho_{\La_{00}\setminus B}$ the density on the scale $\ga^{-1/2}$
corresponding to the configuration $q_{\La_{00}\setminus B }$ and we
call $\bar \rho_{(\Delta\cup B)^c}$ the density corresponding to the
configuration $q_{\La_{00}\setminus B } \cup \chi^{(k)}_{\La^c}$.

A result analogous to  Theorem \ref{propA.2} holds also for the
partition function with the interpolated hamiltonian. Thus,
recalling  Definition \ref{dens}, we have that
    \begin{equation}
 \log\Big( \hat
Z_{\Delta\cup B,t}(\bar{\mathcal  A}_\Delta|q_{\La_{00}\setminus B }
\cup \chi^{(k)}_{\La^c})\Big) \le -\beta \inf_{\rho\in \bar{\mathcal
A}^*_\Delta}f_{\Delta\cup B, \la_{\beta,\ga},t} (\rho|\bar
\rho_{(\Delta\cup B)^c})
 + \bar c \ga^{1/2} |\Delta\cup B|
  \label{5.23b}
     \end{equation}
were analogously to \eqref{appK.1}, for any region $\Om$
    \begin{equation}
   f_{\Om,\la, t}(\rho_\Om|\bar\rho_{\Om^c}) = tF^*_{\Om,\la}
(\rho_\Om|\bar\rho_{\Om^c})+(1-t)[h(\rho_\Om)- \frac 1\beta\mathcal
S(\rho_\Om)]
     \label{flat}
    \end{equation}
Since we want to use Theorem \ref{thmI.3.1} where $\la=\la_\beta$,
we first change the chemical potential by using that there is $c$
such that
        $$
| F^*_{\Delta\cup B,\la_{\beta,\ga}}(\rho_{\Delta\cup B}|\bar
\rho_{(\Delta\cup B)^c})-F^*_{\Delta\cup
B,\la_\beta}(\rho_{\Delta\cup B}|\bar\rho_{(\Delta\cup B)^c})|\le c
\ga^{1/2}|\Delta\cup B|
        $$
Thus \eqref{5.23b} holds with $f_{\Delta\cup B,t}\equiv
f_{\Delta\cup B,\la_\beta,t}$ in place of $f_{\Delta\cup
B,\la_{\beta,\ga},t}$ and with a new constant $c$.
 We have
    \begin{equation}
    \label{5.24b}
    \inf_{\rho\in
\bar{\mathcal  A}^*_\Delta} f_{\Delta\cup B,t}(\rho|\bar
\rho_{(\Delta\cup B)^c})=\inf_{\rho_\Delta\in \bar{\mathcal
A}^*_\Delta}\Big\{\inf_{\rho_B\in {\mathcal{X}}^{(k)}_B}
f_{B,t}(\rho_B|\bar \rho_{\La_{00}\setminus B} +\rho_\Delta)+
f_{\Delta,t}(\rho_\Delta|\chi^{(k)}_{\La^c})\Big\}
    \end{equation}

Recalling that $\rho_\Delta\in  \mathcal X^{(k)}_\Delta$, we can use
Theorem \ref{thmI.3.1} concluding  that there is a function
$\rho^*_B\in \mathcal{X}^{(k)}_B$ (that depends on $\rho_\Delta$)
such that $\rho^*_B=\rho^{(k)}$ in $B_0$ and
    \begin{equation}
    \label{5.25b}
\inf_{\rho_B\in {\mathcal{X}}^{(k)}_B} f_{B,t}(\rho_B|\bar
\rho_{\La_{00}\setminus B}+\rho_\Delta) \ge f_{B,t}(\rho^*_B|\bar
\rho_{\La_{00}\setminus B}+\rho_\Delta) -ce^{-\omega \ell_+}
    \end{equation}
Calling $B_1$ the subset of $B\setminus B_0$ that is connected to
$\Delta$, we observe that the functional on ${\Delta\cup B_1}$
depends on the boundary conditions only in $B_0$ and $\La^c$ where
they are equal to the pure phase  $\rho^{(k)}$.  Thus
    \begin{eqnarray}
    \nn
&&\hskip-2cm f_{B,t}(\rho^*_B|\bar \rho_{\La_{00}\setminus
B}+\rho_\Delta)+ f_{\Delta,t}(\rho_\Delta|\chi^{(k)}_{\La^c})=
f_{\Delta\cup B_1,t}(\rho_{\Delta\cup B_1}|\chi^{(k)}_{(\Delta\cup
B_1)^c}) \\&&\hskip2cm+ f_{B\setminus (B_1\cup
B_0),t}(\rho^*_{B\setminus (B_1\cup B_0)}|\chi^{(k)}_{B_0}+\bar
\rho_{\La_{00}\setminus B})+f_{B_0,t}(\chi^{(k)}_{B_0})
    \label{5.26b}
    \end{eqnarray}
To get rid of the dependence on $\rho_\Delta$ of $\rho_{B\setminus
(B_1\cup B_0)}^*$, we minimize the second term on the right hand
side of \eqref{5.26b} using again Theorem \ref{thmI.3.1}. We thus
get the existence of $\rho^{**}_{B\setminus (B_1\cup B_0)}$ such
that, going back to \eqref{5.24b},
    \begin{eqnarray}
   \nn
 &&  \hskip-2cm \inf_{\rho\in
\bar{\mathcal{A}}^*_\Delta} f_{\Delta\cup B,t}(\rho|\bar
\rho_{(\Delta\cup B)^c})\ge f_{B\setminus (B_1\cup
B_0),t}(\rho^{**}_{B\setminus (B_1\cup B_0)}|\chi^{(k)}_{B_0}+\bar
\rho_{\La_{00}\setminus B})-ce^{-\omega \ell_+}
\\&&\hskip2cm +\inf_{\rho\in \bar{\mathcal{A}}^*_\Delta}
f_{\Delta\cup B_1,t}(\rho|\chi^{(k)}_{(\Delta\cup
B_1)^c})+f_{B_0,t}(\chi^{(k)}_{B_0})
     \label{5.27b}
    \end{eqnarray}
We are left with the estimate of the $\inf$ on the r.h.s. of
\eqref{5.27b}. Recalling \eqref{flat}, we estimate separately the
first term, namely $tF^*_{\Delta\cup B_1}$ and the second namely
$(1-t)[h(\rho)-\beta^{-1}\mathcal{S}(\rho)]$. For the former we use
Lemma \ref{lemma4.16} with $\Omega= \Delta\cup B_1$. Then, calling $
\mathcal R=\hat J_\ga*(\rho_\Om+\chi^{(\und k)}_{\Omega^c})$, from
\eqref{I.1.9.4}, \eqref{D.8b}, \eqref{D.5c}, \eqref{6.48y} and the
fact that $F_2\ge 0$, we get that
    \begin{equation}
    \label{5.28b}
F^*_{\Delta\cup B_1,\la_\beta}(\rho_{\Delta\cup
B_1}|\chi^{(k)}_{(\Delta\cup B_1)^c})|\ge
\ga^{-d/2}\mathbf{I}_k((\Delta\cup
B_1)^c)+\ga^{-d/2}\sum_{x\in\Omega\cap \mathfrak L_\ga}F_{\la_\beta
}^{\rm mf}( \mathcal R(x,1),..,\mathcal R(x,S))
    \end{equation}
where $F_{\la_\beta }^{\rm mf}$, the mean field free energy, is
quadratic around $\rho^{(k)}$. Since $\rho_{\Delta\cup B_1}\in
\mathcal{X}^{(k)}_{\Delta\cup B_1}$, we can take Taylor expansion up
to the second order, getting
    \begin{eqnarray}
   \nn
&&\hskip-2.6cm\ga^{-d/2}\sum_{x\in\Delta\cup B_1\cap \mathfrak
L_\ga}F_{\la_\beta }^{\rm mf}( \mathcal R(x,1),..,\mathcal
R(x,S))\ge \phi|\Delta\cup B_1|\\&&\hskip1cm+\frac{\kappa^*}2\sum_s
\ga^{-d/2}\sum_{x\in(\Delta\cup B_1)\cap \mathfrak L_\ga}\hat
J_\ga*\big[\rho_{\Delta\cup B_1}+\chi^{(\und k)}_{(\Delta\cup
B_1)^c}-\rho^{(k)}\big]^2
     \label{5.29b}
    \end{eqnarray}
If $\rho_{\Delta\cup B_1}\in \bar{\mathcal{A}}_\Delta^*$, observing
that the integral over $\mathbb R^d\setminus\La_{00}\cap
\mathfrak{L}_\ga$ in \eqref{5.15b} can be restricted to a sum over
$\Delta\cap \mathfrak{L}_\ga$ we get
    \begin{eqnarray}
    \nn
&&\hskip-1.6cm \bar{C}\ga^{1/4}|\Delta\cup B_1|\le
\sum_s\ga^{-d/2}\sum_{x\in\Delta\cap \mathfrak
L_\ga}J_\ga*\big|\rho_{\Delta\cup B_1}+\chi^{(\und k)}_{(\Delta\cup
B_1)^c}-\rho^{(k)}\big|\\&&\le |\Delta\cup B_1
|^{1/2}\sum_s\ga^{-d/2}\Big\{\sum_{x\in\Delta\cap \mathfrak
L_\ga}\hat J_\ga*\big[\rho_{\Delta\cup B_1}+\chi^{(\und
k)}_{(\Delta\cup B_1)^c}-\rho^{(k)}\big]^2\Big\}^{1/2}
     \label{5.30b}
    \end{eqnarray}
We next observe that
    \begin{equation}
    \label{a5.30b}
\ga^{-d/2}\mathbf I_k((\Delta\cup B_1)^c)+\phi|\Delta\cup
B_1|=F^*_{\Delta\cup B_1,\la_\beta}(\chi^{(k)}_{\Delta\cup
B_1}|\chi^{(k)}_{(\Delta\cup B_1)^c})
    \end{equation}
Using again that $|\la_\beta-\la_{\beta,\ga}|\le \ga^{1/2}$ and
taking $\bar{C}>2$, from \eqref{5.28b}, \eqref{5.29b},\eqref{5.30b}
and \eqref{a5.30b} we get that
    \begin{equation}
    \label{aa5.33b}
    F^*_{\Delta\cup B_1,\la_{\beta,\ga}}(\rho_{\Delta\cup
B_1}|\chi^{(k)}_{(\Delta\cup B_1)^c})\ge F^*_{\Delta\cup
B_1,\la_{\beta,\ga}}(\chi^{(k)}_{\Delta\cup
B_1}|\chi^{(k)}_{(\Delta\cup B_1)^c})+ 2\bar C|\Delta\cup
B_1|\ga^{1/2}
    \end{equation}
    To estimate the other term in $f_{\Delta\cup B_1,t}$  we observe
that the function of $u=(u_1,\dots,u_S)$,
$-\frac{1}{\beta}\mathcal{S}(u)+h(u)$  is strictly convex and by
\eqref{aa2.47} and \eqref{2.3} the only minimum is in $\rho^{(k)}$.
Thus, calling
$\phi^*_k=-\frac{1}{\beta}\mathcal{S}(\rho^{(k)})+h(\rho^{(k)})$
there is $c^*$ such that
    \begin{eqnarray}
&&\hskip-4cm \ga^{-d/2}\sum_{\Delta\cup B_1\cap\mathfrak L_\ga}
\Big[\frac {-1}\beta \mathcal{S}(\rho)+h(\rho)\Big]\ge
\phi^*_k|\Delta\cup B_1|+c^*\sum_s \ga^{-d/2}\sum_{\Delta\cup
B_1\cap\mathfrak L_\ga}[\rho-\rho^{(k)}]^2
    \label{5.32b}
    \end{eqnarray}

For $\rho\in \bar{\mathcal A}^*_\Delta$, observing that the integral
over $\mathbb R^d\setminus\La_{00}\cap \mathfrak{L}_\ga$ in
\eqref{5.15b} can be restricted to a sum over $\Delta\cap
\mathfrak{L}_\ga$, and using that $\rho-\rho^{(k)}=0$ in
$(\Delta\cup B_1)^c$ and that $\hat J_\ga\star 1=1$, (see
\eqref{6.41y})  we get
    \begin{eqnarray}
    \nn
&&\hskip-2cm\frac{\bar{C}\ga^{1/4}|\Delta|}{|\Delta|^{1/2}}\le
\sum_s \Big\{\ga^{-d/2}\sum_{x\in\Delta \cap \mathfrak L_\ga}\hat
J_\ga*\big[\rho-\rho^{(k)}\big]^2\Big\}^{1/2}
\\&& \nn\le \sum_s
\Big\{\ga^{-d/2}\sum_{x\in\Delta \cap \mathfrak
L_\ga}\ga^{-d/2}\sum_{y\in \Delta \cup B_1}\hat
J_\ga(x,y)\big[\rho(y,s)-\rho^{(k)}\big]^2\Big\}^{1/2}
\\&& \le \sum_s
\Big\{\ga^{-d/2}\sum_{y\in \Delta \cup B_1}
\big[\rho(y,s)-\rho^{(k)}\big]^2\ga^{-d/2}\sum_{x\in \mathfrak
L_\ga}\hat J_\ga(x,y)\Big\}^{1/2}
    \label{5.33b}
    \end{eqnarray}
Using again that $\hat J_\ga* 1=1$ and provided $\bar C>2/c^*$, from
\eqref{5.32b} and \eqref{5.33b}, we get that
    \begin{equation}
    \label{5.34b}
 \ga^{-d/2} \sum_{\Delta\cup B_1\cap\mathfrak
L_\ga} \Big[\frac {-1}\beta \mathcal{S}(\rho)+h(\rho)\Big]\ge
\phi^*_k|\Delta\cup B_1|+ 2\bar{C}\ga^{1/2}|\Delta|,\qquad \forall
\rho\in \bar{\mathcal A}^*_\Delta
    \end{equation}
From \eqref{aa5.33b} and \eqref{5.34b}, we thus get for all
$\rho_\Delta\in \mathcal{A}_\Delta^c$
    \begin{eqnarray}
     \nn
&&\hskip-1cm f_{\Delta\cup B_1,t}(\rho_{\Delta\cup
B_1}|\chi^{(k)}_{(\Delta\cup B_1)^c})\ge tF^*_{\Delta\cup
B_1,\la_{\beta,\ga}}(\chi^{(k)}_{\Delta\cup
B_1}|\chi^{(k)}_{(\Delta\cup B_1)^c})\\&&\nn\hskip
3cm+(1-t)\sum_s\ga^{-d/2}\sum_{\Delta\cup B_1\cap\mathfrak
L_\ga}[\frac{-1}{\beta}\mathcal{S}(\rho^{(k)})+h(\rho^{(k)})]+
4\bar{C} |\Delta|\ga^{1/2}
\\&& \hskip3cm \geq  f_{\Delta\cup B_1,t}
(\chi^{(k)}_{\Delta\cup B_1}|\chi^{(k)}_{(\Delta\cup B_1)^c})+
4\bar{C} |\Delta|\ga^{1/2},
     \label{5.41b}
    \end{eqnarray}

\vskip .5cm

By estimating $ce^{-\omega \ell_+}\le c|\Delta|\ga^{1/2}$, and
calling
        $$\bar\rho_{\Delta\cup
B}^*:=\rho^{**}_{B\setminus (B_1\cup B_0)}+\chi^{(k)}_{\Delta\cup
B_1\cup B_0}
    $$
from \eqref{5.27b} and \eqref{5.41b} we get that for all $\rho\in
\bar{\mathcal{A}}^*_\Delta$, and $\bar C> c$,
    \begin{eqnarray}
    \nn
  && \hskip-1cm  \inf_{\rho\in \mathcal{\bar A}^*_\Delta} f_{\Delta\cup
B,t}(\rho|\bar \rho_{(\Delta\cup B)^c})\ge f_{B\setminus (B_1\cup
B_0),t}(\rho^{**}_{B\setminus (B_1\cup B_0)}|\chi^{(k)}_{B_0}+\bar
\rho_{\La_{00}\setminus B})
\\&&\nn\hskip2.5cm+ f_{\Delta\cup
B_1,t}(\chi^{(k)}_{\Delta\cup B_1}|\chi^{(k)}_{(\Delta\cup
B_1)^c})+f_{B_0,t}(\chi^{(k)}_{B_0}) + (4\bar{C}-c)\ga^{1/2}|\Delta|
\\&&\nn\hskip1cm \ge
f_{B\setminus (B_1\cup B_0),t}(\rho^{**}_{B\setminus (B_1\cup
B_0)}|\chi^{(k)}_{B_0}+\bar \rho_{\La_{00}\setminus B})+
f_{\Delta\cup B_1\cup B_0,t}(\chi^{(k)}_{\Delta\cup B_1\cup
B_0}|\chi^{(k)}_{(\Delta\cup B_1)^c})
\\&&\nn\hskip3cm+3\bar{C}\ga^{1/2}|\Delta|
\\&&\nn\hskip1cm \ge f_{\Delta\cup
B,t}(\bar\rho_{\Delta\cup B}^*|\bar \rho_{(\Delta\cup B)^c})
+3\bar{C}\ga^{1/2}|\Delta|,
     \label{5.31b}
      \end{eqnarray}

\vskip .5cm

Going back to \eqref{5.22b}-\eqref{5.23b} and observing that we can
choose $\bar{C}>\bar c$ so large that
    \begin{eqnarray*}
  -3\bar{C}\ga^{1/2}|\Delta|+ \bar c \ga^{1/2} |\Delta\cup B| +{| \Delta\cup
\delta_{\rm {out}}^{(2\ell_+)}(\Delta)|\ell^{-d}_+}\log\Big(\big[1+
e^{-c \zeta^2\ell_{-}^d}\big]\Big)\le - 2\bar{C}\ga^{1/2}|\Delta|,
   \end{eqnarray*}
we get
    \begin{equation}
\hat Z^{(k)}_{\Delta\cup
B,t}(\mathcal{A}_\Delta^c|\chi^{(k)}_{\La^c}) \le e^{ - 2\bar{C}
\ga^{1/2}|\Delta| }\int_{\mathcal{X}^{(k)}_{\La_{00}\setminus
B}}e^{-\beta H_{\La_{00}\setminus B,t}(q_{\La_{00}\setminus
 B})}\mathbf{X}^{(k)}_{\La\setminus \big(\Delta\cup \delta_{\rm
{out}}^{(2\ell_+)}(\Delta)\big)}e^{-\beta f_{\Delta\cup
B,t}(\bar\rho_{\Delta\cup B}^*|\bar \rho_{(\Delta\cup B)^c})}
       \label{5.36b}
        \end{equation}
Observe that the function $\bar\rho_{\Delta\cup B}^*$ is in the set
$\mathcal{X}^{(k)}_{\Delta\cup  B}$, thus from Theorem \ref{propA.2}
we get that
    \begin{eqnarray}
-\beta f_{\Delta\cup B,t}(\bar\rho_{\Delta\cup B}^*|\bar
\rho_{(\Delta\cup B)^c})
    \le \log \hat Z_{\Delta\cup
B}(\bar \rho_{(\Delta\cup B)^c})+ \bar c\ga^{1/2}|\Delta\cup B|
          \label{5.37b}
    \end{eqnarray}
 We next observe that the denominator in
\eqref{a5.17b} can be bounded using the first inequality in
\eqref{5.18b} getting (recall that $\bar \rho_{(\Delta\cup B)^c}$ is
the density corresponding to the configuration $q_{\La_{00}\setminus
B } \cup \chi^{(k)}_{\La^c}$)
    \begin{eqnarray}
\hskip-1cm Z^{(k)}_{\La,t}(\chi^{(k)}_{\La^c}) \ge
\int_{\mathcal{X}^{(k)}_{\La_{00}\setminus  B}}e^{-\beta
H_{\La_{00}\setminus B,t}(q_{\La_{00}\setminus
 B})}\mathbf{X}^{(k)}_{\La\setminus \big(\Delta\cup \delta_{\rm
{out}}^{(2\ell_+)}(\Delta)\big)} \hat Z_{\Delta\cup B}(\bar
\rho_{(\Delta\cup B)^c})
    \label{5.38b}
    \end{eqnarray}
Thus from \eqref{5.36b}, \eqref{5.37b}, \eqref{5.38b} using again
that $\bar C>\bar c$, we get
            \begin{equation}
     \label{8.43}
\hat Z^{(k)}_{\Delta\cup
B,t}(\mathcal{A}_\Delta^c|\chi^{(k)}_{\La^c}) \le e^{ - \bar{C}
\ga^{1/2}|\Delta| } Z^{(k)}_{\La,t}(\chi^{(k)}_{\La^c})
   \end{equation}
Using that there is $c$ such that $|\Delta|\le c |\La\setminus
\La_{00}|$,   \eqref{8.43}, \eqref{a5.17b} and \eqref{5.16b},
conclude the proof of \eqref{5.11b}.

The proof of \eqref{a5.11b} is similar: with $\Delta$ as in
\eqref{a5.18b} and $\mathcal{A}_\Delta$ as in \eqref{5.15b}, we get
that
    \begin{eqnarray*}
 &&  \hskip-2cm  \Big|\sum_s\ga^{-d/2}\sum_{x\in \La\setminus\La_{00}\cap
\mathfrak{L}_\ga }\mathbb{E}_{\La,t}\Big(\hat
J_\ga\star[\rho_\La(x,s,q_\La) - \rho^{(k)}(s)] \Big)\Big| \le
c[\ga^{1/4}+ \mu_{\La,t,k}(\mathcal{A}_\Delta^c)] \,|\Delta|
\\&&\hskip4cm \le c\,\,|\Delta|\,[\ga^{1/4}+ e^{-\bar{C}\ga^{1/2}|\Delta|}],
    \end{eqnarray*}
 \qed

\bigskip

\subsection{Infinite volume limit}
    \label{sec:5.3}

 The following result will be used to control the
first term on the r.h.s. of \eqref{aa2.53}.

        \begin{thm}
    \label{thm:press}
For any van Hove sequence $\La_n\to\mathbb{R}^d$ of $\mathcal
D^{\ell_+}$- measurable region the following limit exists
        \begin{equation}
        \label{product}
\lim_{n\to\infty} \frac{1}{\beta|\La_n|} \log
Z^{(k)}_{\La_n,0}=:p_{0}^{(k)}
        \end{equation}
Furthermore for any $c_{\rm pol}<c_w$ there is $c$ and for any $\ga$ small enough
    \begin{equation}
        \label{a5.45b}
|R_\La|\le c \big|\delta_{\rm in}^{\ell_+}[\La]\big|\,e^{-c_{\rm
pol}\ell^d_-},\qquad R_\La:=\frac{1}{\beta} \log
Z^{(k)}_{\La,0}-|\La|p_{0}^{(k)}
        \end{equation}
for any $\mathcal
D^{\ell_+}$- measurable region $\La$
\end{thm}

\vskip.5cm

{\bf Proof.} For any $\mathcal D^{(\ell_+)}$ measurable region $\La$
we write $Z_{\La,0}^{(k)} = Z_{\La,0}^{(k,0)} \mathbb E_{\La,0}
({\mathbf X}^{(k)}_{\La,\la_{\beta,\ga}})$ where $ \mathbb
E_{\La,0}$ denotes expectation w.r.t.\ $\dis{d\mu_\La:=\frac
{e^{-\beta h(\rho_{\La})}} {Z_{\La,0}^{(k,0)}} \nu(dq_{\La})}$ and
$\dis{Z_{\La,0}^{(k,0)}= \int_{\mathcal{X}_{\La}^{(k)}} e^{-\beta
h(\rho_{\La})} \nu(dq_{\La})}$.

Let $\und\Ga\in \mathcal  B^{k}_{\La}$. Since $\dis{d\mu_\La = \prod_{C^{(\ell_-)}\subset \La}d\mu_{C^{(\ell_-)}}}$ and since
 for
$\Ga\ne \Ga'$ in $\und \Ga$,
$\dis{\delta_{\rm out}^{2\ga^{-1}}[{\rm sp}(\Ga)]\cap \delta_{\rm out}^{2\ga^{-1}}[{\rm sp}(\Ga')]=\emptyset}$
    \begin{equation}
       \label{5.2.3}
 \mathbb E_{\La,0}\Big( \prod_{\Ga\in \und \Ga}W_{\la_{\beta,\ga}}^{(k)}(\Ga|q) \Big)=
\prod_{\Ga\in \und \Ga} \mathbb E_{\La,0}\Big( W_{\la_{\beta,\ga}}^{(k)}(\Ga|q) \Big)=:\prod_{\Ga\in \und \Ga} \psi^{(k)}(\Ga)
        \end{equation}
We  have
   \begin{equation}
       \label{5.2.3a}
 \psi^{(k)}(\Ga) \le e^{-c_w \zeta^2 \ell_-^d N_\Ga}
        \end{equation}
because $\psi^{(k)}(\Ga)$ is the expectation of $W^{(k)}(\Ga|q) $ which satisfies the same bound independently of
$q$.  By \eqref{5.2.3a} as a direct consequence of the cluster expansion, see for instance Theorem 11.4.3.1 of \cite{leipzig}, for all $\ga$ small enough
    \begin{equation}
       \label{5.2.3b}
 \sum_{\und\Ga\in \mathcal
             B^{k}_{\La}} \prod_{\Ga\in \und \Ga} \psi^{(k)}(\Ga) = e^{-\beta K^{(k)}_\ga(\La)}
        \end{equation}
 where the ``hamiltonian''     $K^{(k)}_\ga(\La)$ can be written as
$\dis{ K^{(k)}_\ga(\La) = \sum_{\Delta  \subseteq \La}
U^{(k)}_{\ga,\Delta}}$, $\Delta$ ranging over the connected
$\mathcal D^{(\ell_+)}$-measurable sets; the  potentials
$U^{(k)}_{\ga,\Delta}$ are translational invariant (in
$\ell_+\mathbb Z^d$) and satisfy the bound: for any  $c_{\rm pol}<
c_w$ and for any $b>0$
     \begin{equation}
       \label{8z.3.2.17}
\beta  \sum_{\Delta \ni r} e^{b N_\Delta} |U^{(k)}_{\ga,\Delta}
|  \leq e^{-\beta c_{\rm pol} (\zeta^2
\ell_{-}^d)},\quad r\in \La
    \end{equation}
for all   $\ga$
  small   enough ($N_\Delta$ being the number of $\mathcal D^{(\ell_+)}$ cubes in $\Delta$).     We are now ready to conclude the proof of Theorem \ref{thm:press} that we will prove with
            \begin{equation}
       \label{5.2.3c}
p_0^{(k)} = \frac {\log Z_{C^{(\ell_-)}_0,0}^{(k,0)}}{\beta |C^{(\ell_-)}_0|} - \ell_+^{-d} \sum
_{\Delta \ni 0}  \frac{U^{(k)}_{\ga,\Delta}}{N_\Delta}
        \end{equation}
because $\dis{\frac {\log Z_{C^{(\ell_-)}_0,0}^{(k,0)}}{\beta |C^{(\ell_-)}_0|}= \frac {\log Z_{\La,0}^{(k,0)}}{\beta |\La|}}$.  We have
$\dis{|R_\La| \le \sum_{C^{(\ell_+)}\in \delta_{\rm out}^{\ell_+}[\La]} \sum_{\Delta\supset C^{(\ell_+)}}|U^{(k)}_{\ga,\Delta}|}$
        such that \eqref{a5.45b} follows from \eqref{8z.3.2.17}.
        \qed

\bigskip

The next results are the main tools for dealing with the ''bulk''
part of the expectation on the  integral on the r.h.s. of
\eqref{aa2.53}.

\medskip
 The
following theorem is a corollary of Theorem 3.1 of \cite{DMPV2}
whose statement is given in the proof below.

\bigskip

\begin{thm}
        \label{thm:infgibbs}
There are $\ga_0$, $\om>0$ and $c$ such that for all $\ga\le \ga_0$
and all $t\in [0,1]$ there is a probability measure
$\mu_{\infty,t,k}$ on $\mathcal{X}^{(k)}$ which is invariant under
translations  in $\ell_+\mathbb{Z}^d$ and such that the following
holds. For any bounded $\mathcal D^{\ell_+}$- measurable region
$\La\in \mathbb{R}^d$ and for any $\Delta\subset \La_{00}$,
($\La_{00}$ is defined in \eqref{a5.10b}) and any cylinder function
$f$ with basis in $\Delta$,
            \begin{equation}
     \label{5.46b}
|\mathbb{E}_{\La,t,k}(f)-\mathbb{E}_{\infty,t,k}(f)|\le c
\|f\|_\infty \,\,e^{-\omega\ga^{1+\alpha_+} {\rm
dist}(\Delta_{\ell_+},\La^c)}
    \end{equation}
 where $\Delta_{\ell_+}$ is the smallest $\mathcal
 D^{(\ell_+)}$-measurable set that contains $\Delta$ and
where $\mathbb{E}_{\La,t,k}$, respectively $\mathbb{E}_{\infty,t,k}$
denote the expectation w.r.t $\mu_{\La,t,k}$, respectively
$\mu_{\infty,t,k}$.
        \end{thm}

\vskip.5cm

{\bf Proof.} In Theorem 3.1 of \cite{DMPV2} it has been proved that
for any bounded, $\mathcal D^{(\ell_+)}$-measurable regions $\La$
and $\La'\supset \La$ and any boundary conditions ${\bar q}_{\La^c}$
and ${\bar q}'_{{\La'}^c}$ the following holds. Let
$\mu_{\La,t,k}(dq_\La|{\bar q}_{\La^c})$ and
$\mu_{\La',t,k}(dq_{\La'}|{\bar q}'_{{\La'}^c})$ be the
probabilities on $\mathcal{X}_\La^{(k)}$ and respectively on
$\mathcal{X}_{\La'}^{(k)}$ defined as in \eqref{aa2.54} but with the
boundary conditions ${\bar q}_{\La^c}$ and ${\bar q}'_{{\La'}^c}$
instead of $\chi^{(k)}$. Then there is a coupling $dQ$ of
$\mu_{\La,t,k}(dq_\La|{\bar q}_{\La^c})$ and
$\mu_{\La',t,k}(dq_{\La'}|{\bar q}'_{{\La'}^c})$ such that if
$\Delta$ is any $\mathcal D^{(\ell_+)}$-measurable subset of $\La$:
    \begin{equation}
      \label{e3.6.1}
      \dis{
  Q\Big(\{\text{\,$(q'_\La,\und \Ga')$ and $(q''_{\La'},\und \Ga'')$
  agree in $ \Delta$} \}\Big)
  \ge 1- c_1 e^{- c_2 \frac {{\rm
  dist}(\Delta,\La^c)}{\ell_+}}}
     \end{equation}
where $(q,\und \Ga)$ agrees with $(q',\und \Ga')$ in $\Delta$
 if all
$\Ga\in \und \Ga$  such that the closure of sp$(\Ga)$ intersects
$\Delta$ are also in $\und \Ga'$ and viceversa and moreover
    \begin{equation}
      \label{e3.6.0}
q\cap \Delta^* = q'\cap \Delta^*,\quad \Delta^*:= \Delta
\bigcup_{\Ga\in \und\Ga} \{{\rm sp}(\Ga)\cup \delta_{\rm
out}^{(\ell_+)} [{\rm sp}(\Ga)]\}
    \end{equation}
Inequality \eqref{e3.6.1} implies that for all cylinder functions
$f$ with basis $\Delta\subset\La$, for $\Delta_{\ell_+}\supset
\Delta$ as in the statement of the Theorem,
    \begin{eqnarray}
    \nn
&&|\mathbb{E} _{\La,t,k}(f|\bar q_{\La^c})-\mathbb{E}
_{\La',t,k}(f|\bar q_{\La'^c})|\le c \|f\|_\infty
Q\Big(\{\text{\,$(q'_\La,\und \Ga')$ and $(q''_{\La'},\und \Ga'')$
 do not agree in $ \Delta_{\ell_+}$} \}\Big)
    \\&&\hskip4cm
 \le c
\|f\|_\infty e^{-\omega\ga^{1+\alpha_+} \,{\rm
dist}(\Delta_{\ell_+},\La^c)}
    \label{aa2.56}
    \end{eqnarray}
Then for any sequence $\La_n\to\mathbb{R}^d$ of $\mathcal
D^{\ell_+}$- measurable region, and for any $f$ as above, the
sequence $\{\mu_{\La_n,t,k}(f|{\bar q}_{\La_n^c})\}$ is a Cauchy
sequence and the limit $\mu_{\infty,t,k}(f)$ defines a probability
on $\mathcal{X}^{(k)}$ which is invariant under translations  in
$\ell_+\mathbb{Z}^d$. From the uniformity on the boundary conditions
defining $\mu_{\La',t,k}(f|\bar q_{\La'\setminus\La}\cup \bar
q'_{\La'^c})$ it follows that for any $\La_n\supset\La$ ($\La_n$ an
element of the sequence defining $\mu_{\infty,t,k}$)
    \begin{equation*}
|\mathbb{E}_{\La,t,k}(f|\bar
q_{\La^c})-\mathbb{E}_{\La_n,t,k}(f)|\le \|f\|_\infty
e^{-\omega\ga^{1+\alpha_+} \,{\rm dist}(\Delta_{\ell_+},\La^c)}
    \end{equation*}
and this implies \eqref{5.46b}.\qed

\bigskip
We will use the following consequence of Theorem \ref{thm:infgibbs}.

    \begin{coro}
    \label{thm:la00}
For any $\mathcal D^{\ell_+}$- measurable region $\La\subset
\mathbb{R}^d$, for any $r\in \La_{00}$
    \begin{equation}
    \label{5.44b}
\Big| \mathbb{E}_{\La,t,k}\big(e^{\rm mf}\big(J_\ga\star q_\La
\big)(r)\big)-\mathbb{E}_{\infty,t,k}\big(e^{\rm mf}\big(J_\ga\star
q_\La \big)(r)\big)\Big|\le c \,\,e^{-\omega\ga^{1+\alpha_+}\, {\rm
dist}(r,\La^c)}
    \end{equation}
where $e^{\rm mf}$ is defined in \eqref{emf}.
        \end{coro}

\vskip.5cm

{\bf Proof.} We denote by  $\Delta_r$ the smallest  $\mathcal
D^{\ell_+}$- measurable set that contains the set $\{r':|r-r'|\le
2\ga^{-1}\}$. We have that
    \begin{equation}
    \label{5.45b}
\sup_{q_\La\in\mathcal{X}_\La^{(k)}} J_\ga\star q_\La(r,s)\le
\ga^d\|J\|_\infty
\,\ell_-^d(\rho^{(k)}_s+\zeta)\,\frac{\ga^{-d}}{\ell_-^d}
    \end{equation}
where $\ell_-^d(\rho^{(k)}_s+\zeta)$ is a bound for the number of
particles in a cube of $\mathcal D^{\ell_-}$ and
$\dis{\frac{\ga^{-d}}{\ell_-^d}}$ is a bound for the number of cubes
in $\mathcal D^{\ell_-}$ that intersect the set $\Delta_r$. Thus
from \eqref{5.46b} and \eqref{5.45b} we get
    \begin{equation}
     \label{5.50b}
     \Big| \mathbb{E}_{\La,t,k}\big(e^{\rm mf}\big(J_\ga\star q_\La
\big)(r)\big)-\mathbb{E}_{\infty,t,k}\big(e^{\rm mf}\big(J_\ga\star
q_\La \big)(r)\big)\Big| \le c \,\,e^{-\omega\ga^{1+\alpha_+} {\rm
dist}(\Delta_{r},\La^c)}
    \end{equation}
Since ${\rm dist}(\Delta_{r},\La^c)\ge {\rm dist}(r,\La^c)-
3\ell_+$, \eqref{5.44b} follows from \eqref{5.50b}. \qed

\subsection{Proof of Theorem \ref{thm:2.14}}
    \label{sec:7n}

As a consequence of Theorem \ref{thm:press} and Corollary
\ref{thm:la00} we have the following result.

        \begin{thm}
        \label{thm5.5}
Let $P_{\la_{\beta,\ga}}$ and $p_{0}^{(k)}$ be the pressures defined
in Theorem \ref{plabeta} and Theorem \ref{thm:press} respectively.
Then
        \begin{equation}
        \label{eqn:infty}
p_{0}^{(k)} = P_{\la_{\beta,\ga}} +  \int_{0}^{1}
\mathbb{E}^*_{\infty,t,k}(A) dt
        \end{equation}
where,
        \begin{equation*}
A(q)  = \mintone{C_{0}^{(\ell_+)}}  e^{\rm mf}[J_\ga \star q] -
\sum_{s}(\varphi_{k}(s)+\la_{\beta,\ga}) \frac{| q(s)\cap
C_{0}^{(\ell_+)} |}{\ell_+^d}.
        \end{equation*}
with
     \begin{equation}
        \label{var}
 \varphi_{k}(s)= r^{(k)}_s\big)-\la_\beta,\qquad
  r^{(k)}_s=\sum_{s'\ne s}\rho^{(k)}_{s'}
     \end{equation}

\end{thm}

\begin{proof}[{\bf Proof.}] We will get \eqref{eqn:infty} by dividing
\eqref{aa2.53} by $|\La|$ and letting $\La \to \infty$. We first
write (below we set $(\mathfrak{L}_{\ell_+}=\ell_+\mathbf{Z}^d$),
    \begin{equation*}
h(\rho_{\La}) + \la_{\beta,\ga}|q_{\La}| =
\ell_{+}^{d}\sum_{x\in\mathfrak{L}_{\ell_+}\cap\La}\sum_{s}
(\varphi_{k}(s)+\la_{\beta,\ga}) \frac{| q(s)\cap C_{x}^{(\ell_+)}
|}{\ell_{+}^{d}} =:
\ell_{+}^{d}\sum_{x\in\mathfrak{L}_{\ell_+}\cap\La} a_x(q_{\La}),
\end{equation*}
We write
        \begin{eqnarray}
            \nn
&&\hskip-3cm H_{\La} (q_\La| \chi_{\La^c}^{(k)} ) + \la_{\beta,\ga}
|q_\La| = \int_{(\La\cup\delta_{\rm
out}^{2\ga^{-1}}[\La])\setminus\La_{00}} e^{\rm mf}[J_{\ga}\star
q_\La\cup \chi_{\La^c}^{(k)}] - e^{\rm mf}[J_{\ga}\star
\chi_{\La^c}^{(k)}]
        \\&&
+\sum_{x\in\mathfrak{L}_{\ell_+}\cap\La_{00}}
\int_{C_{x}^{(\ell_+)}}  e^{\rm mf}[J_{\ga}\star q_\La]
    \label{5.54b}
        \end{eqnarray}
Taking expectation w.r.t $\mu_{\La,t,k}$ we get
\begin{eqnarray}
            \nn
&&\hskip-2cm
    \mathbb{E}_{\La,t,k}\Big(H_{\La} (q_\La| \chi_{\La^c}^{(k)} )
    -h(\rho_\La)\Big)=\int_{(\La\cup\delta_{\rm
out}^{2\ga^{-1}}[\La])\setminus\La_{00}}
\mathbb{E}_{\La,t,k}\Big(e^{\rm mf}[J_{\ga}\star q_\La\cup
\chi_{\La^c}^{(k)}] - e^{\rm mf}[J_{\ga}\star
\chi_{\La^c}^{(k)}]\Big)
\\&&+ \ell_+^d \sum_{x\in\mathfrak{L}_{\ell_+}\cap\La_{00}}
     \mathbb{E}_{\La,t,k}\Big(\mintone{C_{x}^{\ell_+}}
     e^{\rm mf}[J_\ga \star q]
- \sum_{s}(\varphi_{k}(s)+\la_{\beta,\ga}) \frac{| q(s)\cap
C_{x}^{\ell_+} |}{\ell_+^d}\Big)
    \label{5.55b}
        \end{eqnarray}
As in the proof of \eqref{5.45b} we have that
    \begin{equation}
    \label{5.56b}
\sup_{q_\La\in\mathcal{X}_\La^{(k)}} J_\ga\star q_\La(r,s)\le
\ga^d\|J\|_\infty
\,\ell_-^d(\rho^{(k)}_s+\zeta)\,\frac{\ell_+^d}{\ell_-^d}
    \end{equation}
Thus
    \begin{equation*}
\lim_{\La\to\mathbb{R}^d}\frac 1{|\La|}\int_{(\La\cup\delta_{\rm
out}^{2\ga^{-1}}[\La])\setminus\La_{00}}
\mathbb{E}_{\La,t,k}\Big(e^{\rm mf}[J_{\ga}\star q_\La\cup
\chi_{\La^c}^{(k)}] - e^{\rm mf}[J_{\ga}\star
\chi_{\La^c}^{(k)}]\Big)=0
    \end{equation*}
From Theorem \ref{thm:infgibbs}, using the invariance under
$\ell_+$-translation of the limiting measure $\mu_{\infty,t,k}$ and
noticing that $|\mathfrak{L}_{\ell_+}\cap\La|=|\La|/\ell_+^d$, we
get
    \begin{eqnarray}
 \lim_{\La\to\mathbb{R}^d}\frac 1{|\La|} \mathbb{E}_{\La,t,k}
 \Big(H_{\La} (q_\La| \chi_{\La^c}^{(k)} )
    -h(\rho_\La)\Big)= \mathbb{E}_{\infty,t,k}(A)
    \label{a5.55b}
        \end{eqnarray}
From Theorem \ref{plabeta}, \eqref{product} and \eqref{a5.55b},
\eqref{eqn:infty} follows.
\end{proof}

\vskip.1cm

    \begin{coro}
    \label{cor5.7}
We call
    \begin{equation}
        \label{aa5.10b}
\Psi(J_\ga\star q_\La;\chi^{(k)}_{\La^c})= \int_{\La^c}  \left(
e^{\rm mf}\big[J_\ga\star q_\La\cup \chi^{(k)}_{\La^c}\big] -e^{\rm
mf}\big[J_\ga\star\chi^{(k)}_{\La^c}\big]\right) -\int_\La e^{\rm
mf}\big[J_\ga\star\chi^{(k)}_{\La^c}\big]
        \end{equation}
Then the following holds.
\begin{equation}
    \label{5.58b}
\frac 1 \beta\log
Z^{(k)}_{\La}(\chi^{(k)}_{\La^c})=P_{\la_{\beta,\ga}}|\La|+ R_\La+
\mathbf{I}_k(\La^c) - \int_0^1[\mathcal{R}_1+\mathcal{R}_2 +
\mathcal{R}_3- \mathcal{R}_4- \mathcal{R}_5]
    \end{equation}
where $R_\La$ is defined in \eqref{a5.45b} and $\mathbf{I}_k$ in
\eqref{I.1.3},
    \begin{eqnarray*}
        &&\hskip-1cm
\mathcal{R}_1= \int_{\La_{00}}\Big\{\mathbb{E}_{\La,t,k} \big(
e^{\rm mf}\big(J_\ga\star q_\La  \big)- \mathbb{E}_{\infty,t,k}(
e^{\rm mf}\big(J_\ga\star q_\La \big)\Big\}dt
                \\&&\hskip-1cm
\mathcal{R}_2=\int_{\mathbb
R^d\setminus\La_{00}}\big[\mathbb{E}_{\La,t,k}\Big(\big[ e^{\rm
mf}\big(J_\ga\star q_\La\cup \chi^{(k)}_{\La^c} \big) -e^{\rm
mf}(\rho^{(k)})\big]\Big)\\&&\hskip1.6cm-
\mathbb{E}_{\infty,t,k}\Big(\big[ e^{\rm mf}\big(J_\ga\star
q_\La\cup \chi^{(k)}_{\La^c} \big) -e^{\rm mf}(\rho^{(k)})\big]\Big)
    \\&&\hskip-1cm
    \mathcal{R}_3=\mathbb{E}_{\La,t,k}
    \Big(\Psi(J_\ga\star q_\La;\chi^{(k)}_{\La^c})\Big)-\mathbf{I}_k(\La^c)
        \end{eqnarray*}
Finally, recalling  \eqref{2.8},
    \begin{eqnarray*}
        &&\hskip-1cm
\mathcal{R}_4=\sum_s[\varphi_{k}(s)+\la_{\beta,\ga}]\ga^{-d/2}
\sum_{x\in
\La_{00}\cap\mathfrak{L}_\ga}[\mathbb{E}_{\La,t,k}\big(\hat
J_\ga\star \rho_\La\big)-\mathbb{E}_{\infty,t,k}\big(\hat J_\ga\star
\rho_\La\big)]
    \\&&\hskip-1cm
\mathcal{R}_5=\sum_s[\varphi_{k}(s)+\la_{\beta,\ga}]\ga^{-d/2}
\sum_{x\in
\La\setminus\La_{00}\cap\mathfrak{L}_\ga}[\mathbb{E}_{\La,t,k}\big(\hat
J_\ga\star
(\rho_\La-\rho^{(k)})\big)-\mathbb{E}_{\infty,t,k}\big(\hat
J_\ga\star (\rho_\La-\rho^{(k)})\big)]
    \end{eqnarray*}
        \end{coro}

{\bf Proof.} Coming back to \eqref{aa2.53} and using
\eqref{eqn:infty}  and \eqref{a5.45b} we get
        \begin{equation}
        \label{5.11bb}
\frac 1 \beta\log Z^{(k)}_{\La}(\chi^{(k)}_{\La^c}) =
P_{\la_{\beta,\ga}}|\La| + R_\La  - \int_0^1
\left[\mathbb{E}_{\La,t,k} \big(H_\La(q_\La|\chi^{(k)}_{\La^c})
-h(\rho_\La)\big)-|\La|
\mathbb{E}_{\infty,t,k}(A) \right]dt \\
            \end{equation}
Recalling \eqref{5.54b} and definition \eqref{aa5.10b}, we get
    \begin{equation}
        \label{5.10b}
H_{\La}(q_\La|\chi^{(k)}_{\La^c})+\la_{\beta,\ga}|q_\La|=\int_{\La}
e^{\rm mf}\big(J_\ga\star [q_\La \cup \chi^{(k)}_{\La^c}] \big) +
\Psi(J_\ga\star q_\La;\chi^{(k)}_{\La^c})
    \end{equation}
Thus adding and subtracting $\mathbf{I}_k(\La^c) $ we get that the
term corresponding to the energy $e^{\rm mf}$ in the integral on the
r.h.s. of \eqref{5.11bb} is given by
    \begin{eqnarray*}
&&\hskip-1cm \int_{\La}\mathbb{E}_{\La,t,k}(e^{\rm
mf}\big(J_\ga\star [q_\La \cup \chi^{(k)}_{\La^c}]\big)-
\mathbb{E}_{\infty,t,k}(e^{\rm mf}\big(J_\ga\star q\big) +
 \mathcal{R}_3+\mathbf{I}_k(\La^c)
 \\&&=\int_{\La_{00}}\big[\mathbb{E}_{\La,t,k}(e^{\rm
mf}\big(J_\ga\star q_\La \big)- \mathbb{E}_{\infty,t,k}(e^{\rm
mf}\big(J_\ga\star q_\La\big)\big] + \mathcal{R}_2+
 \mathcal{R}_3+\mathbf{I}_k(\La^c)
    \end{eqnarray*}
We have used ``back'' $\ell_+$-translation invariance of
$\mu_{\infty,t,k}$ and to get the term $\mathcal{R}_2$ we have added
and subtract $e^{\rm mf}(\rho^{(k)})$.

For the remaining terms in the expectation on the right hand side of
\eqref{5.11bb}, we consider  $\rho_\La=\rho(\cdot;q_\La))$ as a
function defined on the whole lattice $\mathfrak{L}_\ga$ by setting
$\rho_\La=0$ outside $\La$, so as to we write
    \begin{equation}
    \label{a5.12}
\la_{\beta,\ga}|q_\La|+h(\rho_\La)=\sum_s[\varphi_{k}(s)+\la_{\beta,\ga}]\ga^{-d/2}
\sum_{x\in \mathfrak{L}_\ga}\hat J_\ga\star \rho_\La
    \end{equation}
where we used that $\hat J_\ga\star 1=1$. We then split the
sum for $x\in \mathfrak{L}_\ga$ in a sum over $\La_{00}$
plus the sum over $\La\setminus\La_{00}$. In this last sum
we add and subtract $\rho^{(k)}$ thus getting
$\mathcal{R}_4$ and $\mathcal{R}_5$. \qed

\vskip1cm

From \eqref{5.58b} it follows that in order to conclude the proof of
Theorem \ref{thm:2.14}, namely of \eqref{aa2.57}, we need to show
that
            \begin{equation}
     \label{5.65}
\Big|\mathcal{R}_1+\mathcal{R}_2+\mathcal{R}_3-
\mathcal{R}_4- \mathcal{R}_5+R_\La\Big|\le c\ga^{1/4}
|\delta_{\rm in}^{\ell_+}(\La)|
    \end{equation}

\bigskip

{\bf Proof of \eqref{5.65}}. Since the number of cubes $C\in\mathcal
D^{\ell_-}$ that are in $\delta_{\rm{in}}^{\ga^{-1}}[\La]$ is
bounded by
    \begin{equation*}
\frac {\delta_{\rm{in}}^{\ell_+}[\La]}{\ell^d_+}\frac{\ell_+^{d-1}}
{(\ga^{-1})^{d-1}}\frac{(\ga^{-1})^d}{\ell_-^d}\le
c\,\delta_{\rm{in}}^{\ell_+}[\La]\,\ga^{(1-\alpha_-)d +\alpha_+}
    \end{equation*}
From \eqref{5.44b} we then get
    \begin{equation}
    \label{5.59b}
|\mathcal{R}_1|\le c |\delta_{\rm{in}}^{\ell_+}[\La]
|\ga^{(1-\alpha_-)d+\alpha_+}\int_{|r|\ge\ell^*_+}
e^{-\ga^{1+\alpha_+}\omega r}\le c \ga^{-\alpha_-d+\alpha_+}
|\delta_{\rm{in}}^{\ell_+}[\La]|e^{-\omega\ga^{-(\alpha^*_+-\alpha_+)}}
    \end{equation}
and also
 \begin{equation}
    \label{a5.59b}
|\mathcal{R}_4|\le c |\delta_{\rm{in}}^{\ell_+}[\La]
|\ga^{(1-\alpha_-)d+\alpha_+}\sum_{\substack{x \in
\mathfrak{L}_{\ell_+}\\|x|\ge\ell_+}} e^{-\ga^{1+\alpha_+}\omega
x}\le c \ga^{-\alpha_-d+\alpha_+}
|\delta_{\rm{in}}^{\ell_+}[\La]|e^{-\omega\ga^{-(\alpha^*_+-\alpha_+)}}
    \end{equation}

From \eqref{5.11b} we then get
\begin{equation}
|\mathcal{R}_2|\le \bar c\ga^{1/4}\,|\La\setminus\La_{00}|
\end{equation}
while \eqref{a5.11b} yields
\begin{equation}
|\mathcal{R}_5|\le \bar c\ga^{1/4}\,|\La\setminus\La_{00}|
\end{equation}

In order to estimate $\mathcal{R}_3$ we observe that it is equal to
the expectation of
\begin{align*}
\Psi(J_\ga\star q_\La;\chi^{(k)}_{\La^c}) - \mathbf{I}_k(\La^c) & =
\int_{\delta_{\rm out}^{\ell_+}[\La]}
\{e^{\rm {mf}}(J_\ga*[q_\La \cup \chi^{(k)}_{\La^c}])- e^{\rm {mf}}(\rho^{(k)})\} \\
& \quad +  \int_{\delta_{\rm out}^{\ell_+}[\La] \cup \delta_{\rm
in}^{\ell_+}[\La]} \{e^{\rm {mf.}}(\hat J_\ga*\chi^{(k)}_{\La^c})-
e^{\rm {mf}}( J_\ga*\chi^{(k)}_{\La^c} )\}
\end{align*}
By \eqref{5.11b}, the expectation of the first term is  bounded by
$c\ga^{1/4}|\La\setminus\La_{00} |$, while  \eqref{5.13b} shows that
the expectation of the second term is  bounded by
$c\ga^{1/2}|\La\setminus\La_{00} | $ so $|\mathcal{R}_3|\le
2c\ga^{1/4}\,|\La\setminus\La_{00}|$.\qed

\vskip2cm

\appendix

\setcounter{equation}{0}

\section{{\bf Bounds on the weights of the contours and on the
energies}}
    \label{app:bounds}

In this appendix we will prove lower and upper bounds on the weights
$\mathbf{X}^{(k)}_{\La,\la} ( q_{\La})$ defined in \eqref{I.4.3a},
see Lemma \ref{Lemma1} and Lemma \ref{Lemma2} below. These results
are quite general thus their proof is equal to the one for the LMP
model given in \cite{leipzig}. We also give bounds on the energy in
Lemma \ref{Lemma3} below.

The subsets of $\mathbf{R}^d$ that we consider here are all bounded
$\mathcal{D}^{(\ell_+)}$ measurable regions. We will often drop the
dependence on $\la$ and $q_{\La}$ when no ambiguity may arise, thus
calling $\mathbf{X}^{(k)}_{\La}=\mathbf{X}^{(k)}_{\La,\la} (
q_{\La})$.

 We
extend the definition \eqref{I.4.3a} by setting for $\La\subseteq
\La'$ and $N\in \mathbb N\cup \infty$,
      \begin{equation}
                \label{4.3}
X^{(k),N}_{\La;\La'}=\sum_{\und \Ga\in \mathcal B^k_{\La'}}
\prod_{\Ga\in \und \Ga} \text{\bf 1}_{N_\Ga\le N, {\rm sp}(\Ga)\cap
\La \ne \emptyset}\,\, W_\la^{(k)}(\Ga|q)
    \end{equation}
Observe that, since sp$(\Ga)\cap\La\ne \emptyset$ if and only if
sp$(\Ga)\cap\delta_{\rm{in}}^{(\ell_+)}[\La]\ne \emptyset$, then
$\mathbf{X}^{(k),\infty}_{\La;\La}=\mathbf{X}^{(k)}_{\La}$, hence
\eqref{4.3} indeed extends the definition \eqref{I.4.3a}.

\vskip.5cm

    \begin{lemma}
        \label{Lemma1} [Lower bounds]
For any $N\ge 0$
      \begin{equation}
         \label{4.4}
X^{(k)}_\La\ge X^{(k),N}_{\La,\La},\quad X^{(k),0}_{\La,\La}=1
    \end{equation}
       \begin{equation}
         \label{4.5}
X^{(k)}_{\La,\La'}\ge X^{(k)}_{\Delta,\La'},\qquad \Delta\subset
\La\subseteq \La'
    \end{equation}
       \begin{equation}
         \label{4.6}
 X^{(k)}_\La\ge X^{(k)}_{\La\setminus
\Delta}X^{(k)}_{\Delta}
    \end{equation}
    \end{lemma}

\vskip.5cm

{\bf Proof.} See Lemma 11.1.1.1 in \cite{leipzig}.\qed

\vskip.5cm

    \begin{lemma}
    \label{Lemma2} [Upper bounds]
$\mathbf{X}^{(k)}_{\La;\La'}$, $\La\subset\La'$, is a non decreasing
function of $\La'$, namely
            \begin{equation}
            \label{A.9}
\mathbf{X}^{(k)}_{\La;\La'}(q_{\La'})\le
\mathbf{X}^{(k)}_{\La;\La''}(q_{\La''}), \qquad
\La\subset\La'\subset\La'',\quad q_{\La''}\cap\La'=q_{\La'}
    \end{equation}
and for any $\Delta\subset\La\subseteq\La'\subseteq \La''$,
    \begin{equation}
            \label{A.10a}
\mathbf{X}^{(k)}_{\La;\La'}\le
\mathbf{X}^{(k)}_{\La\setminus\Delta;\La''}X^{(k)}_{\Delta;\La'}
    \end{equation}
Moreover there is a constant $b>0$ such that
      \begin{equation}
         \label{A4.7}
 \mathbf{X}^{(k)}_{\La;\mathbb R^d}
 \le \Big(1+ e^b e^{ -\beta c_w \,
(\zeta^2 \ell_-^d)\, 3^d }\Big)^{|\La|/\ell_+^d}
    \end{equation}
For any $\Delta\subset \La$
    \begin{equation}
         \label{A4.7y}
 \mathbf{X}^{(k)}_{\La}\le \mathbf{X}^{(k)}_{\La\setminus\Delta}
 \Big(1+ e^b e^{ -\beta c_w \,
(\zeta^2 \ell_-^d)\, 3^d }\Big)^{|\Delta|/\ell_+^d}
    \end{equation}
      \begin{equation}
         \label{A4.8}
 \mathbf{X}^{(k)}_{\La;\mathbb R^d}
 \le  \mathbf{X}^{(k)}_\La
 \Big(1+ e^b e^{ -\beta c_w \,
(\zeta^2 \ell_-^d)\, 3^d
}\Big)^{|\delta_{\rm{in}}^{(\ell_+)}[\La]|/\ell_+^d}
    \end{equation}
Finally,
      \begin{equation}
         \label{A4.9}
 \mathbf{X}^{(k)}_\La
\le  \mathbf{X}^{(k),N}_\La \Big(1+ [e^b e^{ -\beta c_w \, (\zeta^2
\ell_-^d)}]^ N\Big)^{|\La|/\ell_+^d}
    \end{equation}
\end{lemma}

\vskip.5cm

{\bf Proof.}  See Lemma 11.1.1.2 of \cite{leipzig} \qed

\vskip.5cm

We now give bounds on the energy.

 \vskip.5cm

    \begin{lemma}
    \label{Lemma3}
Let $\rho_{max}$ be as in \eqref{2.6.0}. There is $c'$ such that for
any $q_\La$ such that $ \rho^{(\ell_-)}(r,s;q_\La ) \le\rho_{max}$,
for all $r\in \La$, and $s\in\{1,\dots S\}$ and for any particle
configuration or density function $\bar q_{\La^c}$ such that $
\rho^{(\ell_-)}(r,s;q_{\La^c}) \le\rho_{max}$, $r\in \La^c$ (in
particular if $q_\La\in \mathcal X^{(k)}_\La $ and $\bar
q_{\La^c}\in \mathcal X^{(k)}_{\La^c} $ for some $k$),
      \begin{equation}
         \label{A.10}
 |   H_{\La,\la}(q_\La |\bar q_{\La^c})|\le c'|\La|,\;\;\text{ for all
$|\la-\la_\beta|\le 1$}
    \end{equation}
    If also $\bar q'_{\La^c}$ is such that
    $\rho^{(\ell_-)}(r,s;q'_{\La^c}) \le\rho_{max}$,
  $r\in \La^c$, then
      \begin{equation}
         \label{A.11}
 |H_{\La,\la}(q_\La |\bar q_{\La^c})
 - H_{\La,\la}(q_\La|\bar q'_{\La^c})|\le c'
 |\partial \La|\ga^{-1},\;\;\text{ for all
$\la$}
   \end{equation}
Finally for all $q_\La$, $\bar q_{\La^c}$ and all $\la$,
         \begin{equation}
         \label{A.12}
  H_{\La,\la}(q_\La|\bar q_{\La^c}) \ge \la|q_\La|
   \end{equation}
    \end{lemma}

{\bf Proof.} First notice that for any $r,r'$
    \begin{equation}
         \label{A.13}
J_\ga(r,r') \le \|J\|_\infty \ga^d \text{\bf 1}_{{\rm
dist}(r,C_{r'}^{(\ga^{-1})})\le \ga^{-1}}
  \end{equation}
Fix $r\in\La$, since there are at most $3^d$ cubes in $\mathcal
D^{(\ga^{-1})}$ at distance $\le \ga^{-1}$ from $r$, we have
      \begin{equation}
         \label{9.6e.2.3}
 J_\ga*q(r,s) \le 3^d\|J\|_\infty \rho_{\max},\qquad \forall s
    \end{equation}
 Thus recalling \eqref{z2.7} we have that
    \begin{eqnarray*}
&&\hskip-1cm|H_{\La,\la}(q_\La |\bar q_{\La^c})| = \Big|\int_{\La
\cup \delta_{\rm out}^{\ga^{-1}}[\La]} \{e_\la^{\rm mf}(
J_\ga*(q_\La\cup\bar q_{\La^c}))-e_\la^{\rm mf}( J_\ga*\bar
q_{\La^c})\}\Big|
\\&&\hskip2cm \le c 3^d\|J\|_\infty \rho_{\max}\,|\La|
  \end{eqnarray*}
thus proving \eqref{A.10} with a constant $c$ independent of $\la$
 if  $|\la-\la_\beta|\le 1$. Analogously we have
    \begin{eqnarray*}
&& |H_{\La,\la}(q_\La |\bar q_{\La^c})
 - H_{\La,\la}(q_\La|\bar q'_{\La^c})|\le \int_{\delta_{\rm
in}^{\ga^{-1}}[\La] \cup \delta_{\rm out}^{\ga^{-1}}[\La]} \Big|
\{e^{\rm mf}( J_\ga*(q_\La\cup\bar q_{\La^c}))-e^{\rm mf}(
J_\ga*\bar q_{\La^c})\}
\\&&\hskip1cm -\{e^{\rm mf}(
J_\ga*(q_\La\cup\bar q'_{\La^c}))-e^{\rm mf}( J_\ga*\bar
q'_{\La^c})\} \Big| \le c'
 |\partial \La|\ga^{-1}
   \end{eqnarray*}
 Finally notice that   for
 $q_\La=(\dots, r_i,s_i,\dots)$ and $\bar q_\La=(\dots, \bar r_i,\bar
 s_i,\dots)$,
    we can write
    \begin{equation*}
H_{\La,\la}(q_\La |\bar q_{\La^c}) = \frac 12 \sum_{i\ne j:
r_i,r_j\in \La} (J_\ga\star J_\ga)(r_i,r_j) \text{\bf 1}_{s_i\ne
s_j}+ \sum_{i: r_i\in \La} \sum_{ j: \bar r_j\in \La^c}(J_\ga\star
J_\ga)(r_i,\bar r_j) \text{\bf 1}_{s_i\ne \bar s_j}-\la|q_\La|\ge
-\la|q_\La|
     \end{equation*}
thus proving \eqref{A.12}. \qed.

\vskip1cm

\section{{\bf Thermodynamic pressures}}
    \label{sec:a2.10}

In this Appendix we  prove Theorem \ref{plabeta}.

\vskip.5cm

$\bullet$\; {\sl For any $\la\in [\la_\beta-1,\la_\beta+1]$ and
$k\in \{1,..,S+1\}$ there is $p(k,\la)$ such that for any   van Hove
sequence $\La_n\to \mathbb R^d$
 of $\mathcal D^{(\ell_+)}$-measurable regions and any sequence $\bar q^{(n)}_{\La_n^c}\in \mathcal X^{(k)}_{\La_n^c}$}
      \begin{eqnarray*}
  \lim_{n\to \infty} \frac{1}{\beta|\La_n|}
\;\log{Z^{(k)}_{\La_n,\la} (\bar q^{(n)}_{\La_n^c})}  = p(k,\la)
    \end{eqnarray*}

\vskip.5cm

The proof is the same as the analogous one for the LMP model in
Subsection 11.7 of \cite{leipzig}. The latter in fact is based on
bounds on the energy and on the weights of the contours which are
the same as those proved in Appendix \ref{app:bounds}.  Existence of
pressure when the  phase space  is non compact it is not an easy
problem in general, the simplifying feature in  the LMP   and the
Potts hamiltonian being the bound
 $H_{\La,\la}(q_\La|\bar q_{\La^c}) \ge -b|q_\La|$ uniform on the boundary conditions,
see  \eqref{A.12}, which in general cannot be
expected to hold.  In this way the problem is essentially reduced to the case
of compact spins. With  the bounds proved in
Appendix \ref{app:bounds} the  contours weights are also easily controlled, the
argument is standard in statistical mechanics.

\vskip1cm

$\bullet$\; {\sl  $p(k,\la)=p(1,\la)$, $k\in \{1,..,S\}$. }

\vskip.5cm

Let $\psi_k(q)$   be the configuration obtained from $q$ by interchanging
spin $1$ and spin $k$, leaving all the other spins and all the positions unchanged.  By the symmetry of
the hamiltonian  $H_{\La,\la}(\psi_k(q_{\La})|\psi_k(\bar
q_{\La^c}))=H_{\La,\la}(q_{\La}|\bar q_{\La^c})$ and the Jacobian
$d\nu(q_\La)/d\nu(\psi_k(q_\La))=1$.  Moreover
$\psi_k: \mathcal{X}_\La^{(1)}$ $\to \mathcal
X_\La^{(k)}$ one-to-one  and onto and
$\mathbf{X}^{(k)}_{\La,\la} ( q_{\La})=\mathbf{X}^{(1)}_{\La,\la} (
\psi_k(q_{\La}))$.  Then $Z^{(k)}_{\La_n,\la} (\psi_k(\bar q^{(n)}_{\La_n^c}))
=Z^{(1)}_{\La_n,\la} ( \bar q^{(n)}_{\La_n^c})$, hence the thesis as we have already proved
independence on the boundary conditions.

\vskip1cm

$\bullet$\; {\sl  $P_\la^{(\rm
ord)}:=p(1,\la)$ and $P_\la^{(\rm
disord)}:=p(S+1,\la)$
 are continuous functions of $\la$.}

\vskip.5cm
By the bounds in Appendix \ref{app:bounds}
we reduce to the same setup as in LMP and the proof becomes the same as in   Subsection 11.7.3 of \cite{leipzig}.
Notice that the dependence on $\la$ is explicit in the hamiltonian but also
implicit in the contours weights.  The dependence on the former is differentiable while the dependence of the
cutoff weights
on $\la$ is only proved to be continuous.  The whole argument is quite standard.

\vskip1cm

$\bullet$\; {\sl There are $c_0$  and $\eps$ positive and
$\rho^{(k)}(\la)=\{\rho_{s}^{(k)}(\la)\}$, $s\in \{1,..,S\}$, $k\in
\{1,..,S+1\}$, $|\la-\la_\beta|\le c_0\eps$, such that: for all $s$,
$\rho^{(S+1)}_s(\la)=\rho^{(S+1)}_1(\la)$; for $k\le S$ and $s\ne
k$, $\rho^{(k)}_s(\la) = \rho^{(1)}_2(\la) <
\rho^{(1)}_1(\la)=\rho^{(k)}_k(\la)$; $\rho^{(k)}(\la)$ are
differentiable in $\la$; $ F^{\rm mf}_{\la }(\cdot)$ has local
minima at  $\rho^{(k)}(\la)$ and}
     \begin{equation}
    \label{2.433}
\frac{d}{d\la}\Big (F^{\rm mf}_{\la
}(\rho ^{(1)}(\la))-F^{\rm mf}_{\la
}(\rho ^{(S+1)}(\la))\Big)\Big|_{\la=\la_\beta} = \sum_{s=1}^S\{\rho^{(S+1)}_s-
\rho^{(1)}_s\}<0
    \end{equation}

\vskip.5cm

This is proved in \cite{DMPV2}.

\vskip1cm

$\bullet$\; {\sl There are  $c'_1>0$ and for any $\kappa>0$ there is
$\ga_1>0$ such that
 for any $\ga\le \ga_1$,
any $k\in \{1,\dots ,S+1\}$  and any $\la$ such that
$|\la-\la_\beta|\le \kappa \ga^{1/2}$
    \begin{equation}
    \label{2.44}
|P_\la^{(\rm ord)}-p_{\rm mf,\la}^{({\rm ord})}|\le
c'_1\ga^{1/2},\quad k=1,\dots S,\qquad |P_\la^{(\rm disord)}-p_{\rm
mf,\la}^{({\rm disord})}|\le c'_1\ga^{1/2}
    \end{equation}
where $p_{\rm
mf,\la}^{({\rm disord})}=-F^{\rm mf}_{\la
}(\rho ^{(S+1)}(\la))$ and
$p_{\rm
mf,\la}^{({\rm ord})}=-F^{\rm mf}_{\la
}(\rho ^{(1)}(\la))$.}

\vskip.5cm

By \eqref{4.4} and \eqref{A4.7}
            \begin{equation*}
      \label{4.122}
 \hat  Z_{\La,\la}(\mathcal X^{(k)}_\La|\chi^{(k)}_{\La^c}) \le    Z^{(k)}_{\La,\la}(\chi^{(k)}_{\La^c}) \le 2^{|\La|/\ell_+^d}
  \hat  Z_{\La,\la}(\mathcal X^{(k)}_\La|\chi^{(k)}_{\La^c})
     \end{equation*}
By
Theorem \ref{propA.2}
            \begin{equation*}
\Big| \log
 \hat  Z_{\La,\la}(\mathcal X^{(k)}_\La|\chi^{(k)}_{\La^c}) +\beta\inf_{\rho_\La: \eta(\rho_\La;r)=k, r\in \La} F^*_{\La,\la}(\rho_{\La}|\chi^{(k)}_{\La^c})\Big|\le
c \ga^{1/2} |\La|
     \end{equation*}
Postponing the proof that $\dis{ F^*_{\La,\la}(\rho_{\La}|\chi^{(k)}_{\La^c})
\ge F^*_{\La,\la}(\chi^{(k)}_{\La}|\chi^{(k)}_{\La^c})}$ we get
            \begin{equation*}
\frac 1 { |\La|} \Big| \log Z^{(k)}_{\La,\la}(\chi^{(k)}_{\La^c}) +\beta F^*_{\La,\la}(\chi^{(k)}_{\La}|\chi^{(k)}_{\La^c})\Big|
\le
c \ga^{1/2}+ \ell_+^{-d} \log 2
     \end{equation*}
Choose  $\La$ as a cube of side $L$,
then $|F^*_{\La,\la}(\chi^{(k)}_{\La}|\chi^{(k)}_{\La^c})-|\La| F^{\rm mf}_{\la
}(\rho ^{(k)}(\la))| \le c\ga^{-1}L^{d-1}$ and \eqref{2.44}
follows letting $L\to \infty$.  It thus remains to prove that
for any $|\la-\la_\beta|\le \kappa \ga^{1/2}$ and $\ga$ small enough,
  $\dis{ F^*_{\La,\la}(\rho_{\La}|\chi^{(k)}_{\La^c})
\ge F^*_{\La,\la}(\chi^{(k)}_{\La}|\chi^{(k)}_{\La^c})}$.
The proof is taken from Proposition 11.1.4.1 in \cite{leipzig}.

 Call $\rho$ the
function equal to $\rho_\La$ on $\La$ and to $\chi^{(k)}_{\La^c}$ on
${\La^c}$.  Since
 $\mathcal S(0)=0$,
$\mathcal S(\chi^{(k)}_{\La^c})=\mathcal
S(\chi^{(k)}_{\La^c})\text{\bf 1}_{\La^c}$,
       \begin{eqnarray*}
&&  F^*_{\La,\la}(\rho_\La|\chi^{(k)}_{\La^c})= \int_{\mathbb
R^d} \{F^{\rm mf}_{\la}(\hat J_\ga
* \rho)- F^{\rm mf}_{\la}(\hat J_\ga*\chi^{(k)}_{\La^c})\} +\frac 1
\beta \{ {\mathcal S(\hat J_\ga*\rho)} - \mathcal S(\rho)\}\nn\\&&
\hskip3cm -\frac 1 \beta \{ \mathcal S(\hat
J_\ga*\chi^{(k)}_{\La^c}) -\mathcal S(\chi^{(k)}_{\La^c})\}
     \end{eqnarray*}
We can write the integral of the sum as the sum of the integrals and
in the integral with $\{ {\mathcal S(\hat J_\ga*\rho)} - \mathcal
S(\rho)\}$ we can replace $\mathcal S(\rho)$ by $\hat J_\ga*\mathcal
S(\rho)$. Then $F^*_{\La,\la}(\rho_\La|\chi^{(k)}_{\La^c})$ becomes
       \begin{eqnarray*}
&& \int_{\mathbb R^d} \{F^{\rm mf}_{\la}(\hat J_\ga
* \rho)- F^{\rm mf}_{\la}(\hat J_\ga*\chi^{(k)}_{\La^c})\} +\frac 1
\beta \int_{\mathbb R^d} \{ {\mathcal S(\hat J_\ga*\rho)} - \hat
J_\ga* \mathcal S(\rho)\}\nn\\&& \hskip3cm -\frac 1 \beta
\int_{\mathbb R^d}\{ \mathcal S(\hat J_\ga*\chi^{(k)}_{\La^c})
-\mathcal S(\chi^{(k)}_{\La^c})\}
     \end{eqnarray*}
Since  $\eta(\rho_\La;\cdot)\equiv k$, for all $\ga$ small enough the
first curly bracket is minimized by setting $\rho_\La=
\chi^{(k)}_{\La}$; the second curly bracket by
convexity is  non negative and vanishes when  $\rho_\La=
\chi^{(k)}_{\La}$; the third one is independent of
$\rho_\La$, hence  $\dis{ F^*_{\La,\la}(\rho_{\La}|\chi^{(k)}_{\La^c})
\ge F^*_{\La,\la}(\chi^{(k)}_{\La}|\chi^{(k)}_{\La^c})}$.

\vskip1cm

The proof of \eqref{p2.30} follows because:
$P_{\la_{\beta,\ga}}^{(\rm ord)}-P_{\la_{\beta,\ga}}^{(\rm disord)}$
is continuous and there is $c>0$ such that for all $\ga$ small
enough
      \begin{equation}
         \label{tB.1}
 P_\la^{(\rm ord)}- P_\la^{(\rm disord)} \begin{cases} <0
 &\text{if $\la=\la_\beta - c  \ga^{1/2} $}\\>0
 &\text{if $\la=\la_\beta + c   \ga^{1/2}$} \end{cases}
    \end{equation}
\eqref{tB.1} holds because: $|\{P_\la^{(\rm ord)}- P_\la^{(\rm
disord)}\}-\{p_{\rm mf,\la}^{({\rm ord})}-p_{\rm mf,\la}^{({\rm
disord})}\} |\le 2c'_1\ga^{1/2}$.  By \eqref{2.433} and the
smoothness of $\{p_{\rm mf,\la}^{({\rm ord})}-p_{\rm mf,\la}^{({\rm
disord})}\}$, there is $a>0$ such that for any $\kappa$ and all
$\ga$ correspondingly small,
      \begin{equation*}
 \{p_{\rm mf,\la}^{({\rm ord})}-p_{\rm mf,\la}^{({\rm disord})}\} \ge a (\la-\la_\beta),\;\; \la \in [\la_\beta,\la_\beta+\kappa \ga^{1/2}]
    \end{equation*}
 Hence $ P_\la^{(\rm ord)}- P_\la^{(\rm disord)}\ge (a\kappa-c'_1)
 \ga^{1/2}$.

\vskip1cm

{\bf ACKNOWLEDGMENTS}  One of us (YV) acknowledges very kind
hospitality at the Math. Depts. of  Roma TorVergata and L'Aquila,
partially supported by PRIN  $2004028108/005$  and $2004028108/008$,
GREFI-MEFI (GDRE 224 CNRS-INdAM), CPT (UMR 6207) and Universit\'e de
la M\'editerran\'ee. I.M. acknowledges partial support of the GNFM
young researchers project ``Statistical mechanics of multicomponent
systems".

\bigskip

\bibliographystyle{amsalpha}

\begin{thebibliography}{99}


\bibitem{BKMP} P. Baffioni, T.Kuna,  I Merola, E. Presutti:
{A liquid vapor phase transition in quantum statistical mechanics.
Submmitted to {\em Memoirs AMS } (2004).}

\bibitem{BMP} P. Butt\`{a}, I. Merola, E. Presutti: On the validity
of the van der Waals theory in Ising systems with long range
interactions {\em Markov Processes and Related Fields \bf 3} (1977)
63--88


\bibitem{BMPZ} A. Bovier, I. Merola, E. Presutti, M. Zahradn\`\i k:
{On the Gibbs phase rule in the Pirogov-Sinai regime. {\em J. Stat.
Phys. \bf 114} (2004), 1235--1267.}

 \bibitem{BZ} A. Bovier, M. Zarhadnik
 {The low temperature phase of Kac-Ising models
{\em J.Stat. Phys. \bf 87}  (1997), 311-332.}

\bibitem{CP} M. Cassandro, E. Presutti
{Phase transitions in Ising systems with long but finite range
interactions {\em Markov Processes and Related Fields \bf 2}(1996)
241--262.}


 \bibitem{chayes} M. Biskup, L. Chayes, N. Crawford
{Mean-field driven first-order phase transitions in systems with
long-range interactions {\em J. Stat. Phys. \bf 122(6) } (2006),
1139--1193. }

\bibitem{DMPV2} A. De Masi, I. Merola, E. Presutti,
Y. Vignaud: Potts models in the continuum. Uniqueness and
exponential decay  in the restricted ensembles {\em J. Stat. Phys.}
in press (2008)

\bibitem{DS}{R.L. Dobrushin, S.B. Shlosman:}
{Completely analytical interactions: constructive description {\em
J.Stat. Phys. \bf 46(5--6)} (1987) 983--1014.}


\bibitem{GP} D. J. Gates, O. Penrose:
The van der {W}aals limit for classical systems.
I. A variational principle. {\em Comm. Math. Phys. \bf 15(4)} (1969) 255--276.

\bibitem{GMRZ}
{Hans-Otto Georgii, Salvador Miracle-Sole, Jean Ruiz, Valentin
Zagrebnov: Mean field theory of the Potts Gas {\em J. Phys. A} {\bf
39}  (2006) 9045--9053.}


\bibitem{GH} {Hans-Otto Georgii, O. H\"{a}ggstr\"{o}m:
Phase transition in continuum Potts models. {\em  Comm. Math. Phys.}
{\bf  181} (1996)  507--528.}



\bibitem{GM} T. Gobron, I. Merola: First order
phase transitions in Potts models with finite range interactions
{\em J. Stat. Phys., \bf 126} (2006).

\bibitem{KP} R. Koteck\'{y}, D Preiss: Cluster expansion for
abstract polymer models {\em Comm. Math. Phys.,} {\bf {103}} (1986),
491--498.

\bibitem{LP}  J.L. Lebowitz, O. Penrose: Rigorous Treatment of
the Van Der Waals-Maxwell Theory of the Liquid-Vapor Transition,
{\em J. Math. Phys. \bf 7} (1966) 98--113.

\bibitem{LMP} J.L. Lebowitz, Mazel, E. Presutti: Liquid vapour
phase transitions for systems with finite range interactions {\em J.
Stat. Phys} (1999).


\bibitem{leipzig} E. Presutti:
Scaling limits in Statistical Mechanics and microstructures in
Continuum Mechanics. {\em Springer Verlag} (2008).

\bibitem{ruelle} D. Ruelle: Widom-Rowlinson:
{Existence of a phase transition in a continuous classical system.
{\em Phys. Rev. Lett. \bf 27} (1971) 1040--1041.}


\bibitem{cluster-expansion} {M. Zahradn\`\i k:
A short course on the Pirogov-Sinai theory.
 {\em Rend. Mat. Appl.} {\bf 18}, 411--486 (1998).}

\end{thebibliography}

\end{document}